\documentclass[arunningheads,a4paper]{llncs}

\usepackage{amssymb,xspace,latexsym,epic,eepic,graphicx}
\usepackage{epsfig,amsmath}
\usepackage{colortbl}

\usepackage{epstopdf}

\usepackage[boxruled,vlined]{algorithm2e}

\usepackage{xcolor}

\usepackage{tikz}
\usetikzlibrary{arrows}

\newtheorem{fact}{Fact}
\newtheorem{observation}{Observation}

\newcommand{\ryc}[1]{}
\newcommand{\jd}[1]{}
\newcommand{\jr}[1]{}

\newcommand{\Nulls}{\mathcal{N}}

\newcommand{\nulls}{\textit{Nulls}}

\renewcommand{\epsilon}{\varepsilon}
\newcommand{\Sch}{{\cal S}}
\newcommand{\Inst}{{\cal I}}

\newcommand{\schema}{\text{Schema}}

\newcommand{\Pro}{P}
\newcommand{\A}{\mathcal{A}}

\newcommand{\C}{\mathcal{C}}
\newcommand{\Const}{\mathbb{C}}

\renewcommand{\P}{\mathbb{P}}

\newcommand{\Que}{\mathbb{Q}}
\newcommand{\Q}{\mathcal{Q}}
\newcommand{\pre}{\text{in}}
\newcommand{\post}{\text{out}}
\newcommand{\safe}{\text{pres}}

\newcommand{\Qsafe}{\Q_\safe}

\newcommand{\outcome}{\textit{outcomes}}
\newcommand{\dynoutcome}{\textit{dyn-outcomes}}

\newcommand{\rep}{\textit{rep}}

\newcommand{\Cpre}{\mathcal{C}_\text{in}}
\newcommand{\Cpost}{\mathcal{C}_\text{out}}
\newcommand{\Scope}{\textit{Scope}}

\newcommand{\skb}{\mathcal K}
\newcommand{\chase}{\textit chase}
\newcommand{\removerels}{\textsc RemoveRelations}
\newcommand{\Reps}{\mathbf{W}}

\newcommand{\certain}{\textit certain}

\newcommand{\evisits}{\textit{EVisits}\xspace}
\newcommand{\locvisits}{\textit{LocVisits}\xspace}
\newcommand{\patients}{\textit{Patients}\xspace}

\newcommand{\inid}{\textit{insId}\xspace}
\newcommand{\fid}{\textit{facility}\xspace}
\newcommand{\pid}{\textit{pId}\xspace}
\newcommand{\timestamp}{\textit{timestp}\xspace}
\newcommand{\age}{\textit{age}\xspace}

\newcommand{\ptime}{\textbf{{\sf Ptime}}\xspace}
\newcommand{\np}{\textbf{{\sf NP}}\xspace}
\newcommand{\pspace}{\textbf{{\sf Pspace}}\xspace}
\newcommand{\conp}{\textbf{{\sf coNP}}\xspace}
\newcommand{\exptime}{\textbf{{\sf Exptime}}\xspace}
\newcommand{\nexptime}{\textbf{{\sf NExptime}}\xspace}
\newcommand{\twonexptime}{\textbf{{\sf N2Exptime}}\xspace}
\newcommand{\exptimenp}{\textbf{{\sf Exptime$^\np$}}\xspace}

\newcommand{\OMIT}[1]{}

\begin{document}

\mainmatter  % start of an individual contribution

% first the title is needed
\title{Assessing Achievability of Queries and Constraints}

% a short form should be given in case it is too long for the running head
\titlerunning{Assessing Achievability of Queries and Constraints}

% the name(s) of the author(s) follow(s) next
%
% NB: Chinese authors should write their first names(s) in front of
% their surnames. This ensures that the names appear correctly in
% the running heads and the author index.
%
\author{Rada Chirkova$^{1}$\and Jon Doyle$^{1}$\and Juan L. Reutter$^{2}$}
\authorrunning{Assessing Achievability of Queries and Constraints}
% (feature abused for this document to repeat the title also on left hand pages)

% the affiliations are given next; don't give your e-mail address
% unless you accept that it will be published
\institute{$^{1}$ Computer Science Department, North Carolina State University\\
North Carolina, USA\\
$^{2}$ Pontificia Universidad Cat\'olica de Chile\\
%\mailsa\\
\texttt{chirkova@csc.ncsu.edu, Jon\_Doyle@ncsu.edu, jreutter@ing.puc.cl}  \\
}

%
% NB: a more complex sample for affiliations and the mapping to the
% corresponding authors can be found in the file "llncs.dem"
% (search for the string "\mainmatter" where a contribution starts).
% "llncs.dem" accompanies the document class "llncs.cls".
%

\toctitle{Lecture Notes in Computer Science}
\tocauthor{Authors' Instructions}
\maketitle

\begin{abstract}
Assessing and improving the quality of data in data-intensive systems are fundamental challenges that have given rise to numerous applications targeting transformation and cleaning of data. However, while schema design, data cleaning, and data migration are nowadays reasonably well understood in isolation, not much attention has been given to the interplay between the tools that address issues in these areas. Our focus is on the problem of determining whether there exist sequences of data-transforming procedures that, when applied to the (untransformed) input data, would yield data satisfying the conditions required for performing %deemed prerequisite to 
the task in question. Our goal is to develop a framework that would address this problem, starting with the relational setting.

In this paper we abstract data-processing tools as black-box procedures. This abstraction describes procedures by a specification of which parts of the database might be modified by the procedure, as well as by the %two sets of 
constraints that specify the required states of the database before and after applying the procedure. We then proceed to study fundamental algorithmic questions arising in this context, such as understanding when one can guarantee that sequences of procedures apply to original or transformed data, when they succeed at improving the data, and when knowledge bases can represent the outcomes of procedures. Finally, we turn to the problem of determining whether the application of a sequence of procedures to a database results in the satisfaction of properties specified by either queries or constraints. We show that this problem is decidable for some broad and realistic classes of procedures and properties, even when procedures are allowed to alter the schema of instances.
\end{abstract}

\section{Introduction}\label{sec:intro}

A common approach to ascertaining and improving the quality of data 
is to develop procedures and workflows for repairing or improving data sets 
with respect to quality constraints. 
The community has identified a wide range of %data-management 
problems that are vital in this respect, leading to the creation of several lines of research, which 
have normally been followed by the development of toolboxes %of  applications 
that practitioners can use to solve their problems. %This has been the case, for instance, 
%for the %the notoriously labor- and knowledge-intensive
%Extract-Transform-Load (ETL) \cite{Devlin97,Kimball04} process in
%business applications, or for the development of %automatic 
%tools to reason about the completeness or cleanliness 
%of the data \cite{2012FanGeertsBook}. 

As a result, organizations facing data-improvement problems now have access to a variety of data-mana- gement 
tools to choose from; the tools can be assembled to construct what can be called %called 
workflows of data operations. 
However, in contrast to the considerable body of research on specific %data 
operations or %even 
entire business workflows (see, e.g., \cite{DHPV09,BGHLS07,2012Deutsch,BerardiCGHLM05}), previous research appears to have not
focused explicitly %either 
on the assembly process itself nor on
providing guarantees that the desired data-quality constraints will be
satisfied once the assembled workflow of procedures has been applied
to the %available 
data.  

We investigate %study the problem of 
constructing workflows from available procedures. 
That is, we consider a scenario in which an organization needs to meet a certain data-quality criterion or goal using 
available data-improvement procedures. %In this case, t
The problem is  to understand whether the procedures can be 
assembled into a %data-improvement 
workflow in a way that would guarantee that the data produced by the 
workflow will %effectively 
meet the desired quality goal. 

\medskip
\noindent
\textbf{Motivating example:} Suppose a medical analyst wishes to know the emergency rooms 
that are used by patients with a certain medical insurance. 
The data owned by the analyst reside in a relation \locvisits(\fid,\pid,\timestamp), with the attributes standing, respectively, 
%\fid, containing 
for the id of the facility where the emergency room is, 
%\pid, containing 
the social-security number 
of a patient, and a timestamp marking the date of the emergency visit. 

The analyst has also been given two procedures he can execute as-is but not modify: A procedure $P_\text{migrate}$, 
which is supposed to migrate data  into %the relation 
\locvisits from relation \evisits owned by another analysis company, 
and a procedure $P_\text{insur}$, which augments the relation $\locvisits$ with an attribute $\inid$ containing the 
insurance of patients, and whose data are drawn from relation $\patients(\pid,\inid)$ owned by the local authority. 

Given an insurance id $I$, the analyst can capture the desired information via query 
\texttt{SELECT facility FROM LocVisits WHERE insId = I}, posed over %a relation 
$\locvisits$ with an additional attribute $\inid$ containing the 
insurance of patients. It is natural for the analyst to ask: 
Can I use any or all of the above procedures to transform my data so that this query can be %effectively 
posed on my 
database? Or is there a way to apply these procedures so that I can guarantee that my database satisfies certain 
quality criteria?
\subsection{Contributions}

Our specific focus is on the problem of determining whether there exist sequences 
of data-transforming procedures that, when applied to  the given data, would yield data satisfying certain 
given conditions. 

%As our first contribution, w
We propose a formal framework in which data-processing tools are abstracted as 
black-box procedures, describing them %only 
by means of the following information:  
\par
\noindent
\begin{itemize}
\item A specification of which parts of the database the procedure is modifying; 
\item A set of conditions that need to be satisfied in order for the procedure to be applied; 
\item A set of conditions that are guaranteed to be satisfied once the procedure has been applied; and 
\item In some cases, additional guarantees that certain pieces of data will not be deleted or modified. 
\end{itemize}
\noindent 

We also define the notion of outcome of applying a procedure to an instance of data, and consider  
sets of such outcomes. %Since our abstraction of procedures 
%cannot describe them in complete detail, we may, in general, have an infinite number of possible outcomes for a given instance and procedure. 
%
%: 
As our goal is to reason about workflows of data-processing procedures, we also study the notion of outcome (and sets of 
outcomes) of a {\em sequence} of applications of procedures. 

These definitions naturally lead us to two fundamental decision problems in our framework. The first problem is applicability: Given an 
instance and a sequence of procedures, is one guaranteed to be able to apply successive procedures in the sequence? The second problem is non-emptiness: Is one guaranteed to obtain at least one outcome of applying a given sequence of procedures? 
We show that our definitions are too general to guarantee efficient algorithms for these problems, but also identify interesting and realistic 
classes of procedures that lead to the tractability of these basic problems. 

Next, for sequences of procedures belonging to the well-behaved classes we have identified, we focus on representing the sets of their outcomes. We show that these sets can be represented by a knowledge base in which the knowledge is given by tuple-generating dependencies, and %in which 
some of the relations are closed to adding more data. %Interestingly, 
We show that such knowledge bases form a strong representation system, in the sense of \cite{IL84}, for application of procedures. 
We also show how to reason about such knowledge bases, studying in particular the problems of query answering and constraint satisfaction. % of constraints as well as . 

Finally, we use our toolbox %into use 
to study what we call the data-readiness problem: Given an instance $I$, a set $\Pi$ of procedures, and 
a specification of a property over instances, is there a way to construct a workflow with procedures from $\Pi$ so that each instance in 
the outcome %set of this workflow 
satisfies this property? Once again, while undecidable in its general form, we show that 
this problem is decidable for some broad classes of procedures. 

%We remark that our framework, while studied here in a relational setting, might be general enough to allow for much more 
%expressive classes of procedures, including procedures calculating statistics of the data, transformation on semistructured or 
%even unstructured text data, and, why not, machine learning operations.

\medskip
\noindent
\textbf{Structure of the Exposition.}
To simplify the formal exposition, in this paper we 
restrict our attention to relational data.  The general methods, however, 
seem promising for application to other forms of data as well, including semistructured %data
and %unstructured 
text data. 

Within the scope of relational data, in most of the discussion in this paper 
we further restrict our attention to transformations of data that do not 
change the schema.  In this setting, one can formalize many types of 
data-adequacy conditions in terms of dependencies, and treat the above planning 
task in terms of chase.  We also briefly consider transformations that change both the schema and contents of the data, and sketch the use of such transformations in treating schema updates.

%too much I think
%Subsequent parts of the general framework, not treated in this paper,
%would provide means for formalizing tasks that serve as the purposes of
%analyses and derivation or construction of adequacy conditions on data
%from such tasks.  Theoretical treatment of these problems would bring in
%formal concepts and techniques from decision theory, deontic and
%temporal logics, and statistics.  Practical treatment of these methods
%would bring in the full armamentarium of automated planning and model-checking techniques.

\jd{Add some references to planning and model checking methods.}

\subsection{Related Work}

\jd{I don't know this literature, and this should be further revised
  to take the ICDT reviewer comments.}

\ryc{This is ok for the moment, I will bring in next week some relevant work by AnHai Doan, by Vardi, and by the data-quality community. Also SIGMOD-17 and PODS/ICDT-17, I suppose.}

\textbf{Data quality:}  Numerous works treat issues in the broad spectrum of data quality.
\cite{WangS96} provides a widely acknowledged study on
eliciting and defining specific dimensions of quality of the data; see
also \cite{KahnSW02,LeeSKW02}.  In this space, many works %treating data quality
view data-quality measures as objective properties unconnected
to specific uses of the data. %, rather than as largely determined by
%the needs of the anticipated uses of the data.  
The role of purpose in determining data
quality is more visible in \cite{LeePWF09,WangLPS98,ChengalurSP98},
where quality data are regarded as being fit for their intended use,
taking both context and use (i.e., tasks to be performed) into
account when evaluating and improving the quality of data.

% \cite{} report on measuring the impact of information quality on
% decision processes.

\textbf{Task dependence:}  Recent efforts have put an emphasis on
data-quality policies and strategies w.r.t. specific tasks to be performed on the data.  \cite{LeePWF09} presents
%a groundbreaking set of 
general information-quality policies that
structure decisions on information, and \cite{Wang98} presents an 
%quality-
improvement cycle %consisting of define, measure,
%analyze, and improve steps 
for data quality. 
 \cite{LeeS04} moves
toward integration of process measures with information-quality
measures.  Our work in the present paper differs from these lines of
research in that we assume that task-oriented data-quality
requirements are already given in the form of %task-adequacy
constraints on the data, and that procedures for improving data
quality are also specified and available.

\textbf{Data preparation:}  The work \cite{2012FanGeertsBook}
introduces a unified framework covering formalizations and approaches
for a range of problems in data extraction, cleaning, repair, and
integration, and also supplies an excellent survey of related work in
these areas. More recent work in data preparation includes
\cite{BergmanMNTsigmod15,BergmanMNT15,KrishnanWFGKM015,RazniewskiKNS15,NuttPS15}.

\textbf{Workflows:} Research on business processes \cite{2012Deutsch}
studies both the environment in which data are generated and
transformed, including processes, users of data, and goals of using
the data, and automatic composition of services into business
processes under the assumption that the assembly needs to follow a
predefined workflow of executions of actions or services
\cite{BerardiCGHLM05,BerardiCGHM05,BerardiCGLM05}.  Our work, in
contrast, begins with the data properties that the workflow should
ensure, rather than with the outlines of the workflow itself. In that sense, our work is in line with the efforts of, e.g., \cite{Chatterjee0V15}, while differing from those works in the nature of the specifications and of the components from which 
workflows are assembled.  The
work of \cite{DHPV09,BGHLS07} stands closer to reasoning about static
properties of business-process workflows, but does not pursue the goal
of constructing workflows.

Some recent systems work, e.g., \cite{KondaDCDABLPZNP16}, emphasizes the importance of data-transforming workflows assembled from individual procedures, while advocating for users to choose from sets of preassembled workflows. % the ones that are the best fit for the user's task. 
In this paper, we focus on providing tools for assembling individual  data-transforming workflows as needed, which complements nicely the efforts of the line of work of \cite{KondaDCDABLPZNP16}.

\section{Preliminaries}
\label{sec:prelim}

\subsection{Schemas and Instances}

Assume a countably infinite set of attribute names
$\mathcal{A} = \{ A_1, A_2, \ldots \}$ totally ordered by
$\leq_\mathcal{A}$, a countably infinite domain of values (or elements)
$D$ disjoint from $\mathcal{A}$, and a countably infinite set of
relation names $\mathcal{R} = \{ R_1, R_2, \ldots \}$ disjoint from
both $\mathcal{A}$ and $\mathcal{R}$.
A relational schema over $\mathcal{A}$ and $\mathcal{R}$ is a partial
function $\Sch: \mathcal{R} \to 2^{\mathcal{A}}$ with finite domain, 
which associates a finite set of attributes with a finite set of relation
symbols.  Abusing the notation, we say that $R$ is in $\Sch$ if $\Sch(R)$ is defined.

An instance $I$ of schema $\Sch$ %, as defined in the simplest way, %, in the simplest conception, 
assigns
a set $R^I$ of tuples to each relation $R$ in  $\Sch$, so that if
$\Sch(R) = \{A_1,\dots,A_n\}$, then $R^I \subseteq D^n$, with the set of 
tuples structured so that the elements of each tuple $(a_1, \dots, a_n)$
appear in the assumed attribute order, that is,
$A_1 <_{\mathcal{A}} \cdots <_{\mathcal{A}} A_n$.  
%We write
%$\schema(I)$ to denote the schema of an instance $I$.

Regarding instances as sets of tuples as above % in this way 
suffices when we consider 
%treating 
data transformations that do not change the schema.  When treating
transformations that change the schema of the data, we can no longer treat the functions
in $R^I$ as lists of values in $D$, and must replace this
\emph{unnamed} perspective with a \emph{named} perspective
that explicitly notes the attributes connected with each tuple element.
Following \cite{AHV95}, we regard $R^I$ as a set of functions
from $\Sch(R)$ to $D$, and each tuple $t$ in $R^I$ as a
sequence of functions $\bar t = t(A_1), \dots, t(A_n)$ that lists the values
in the attribute order, writing $t(A_i)$ to denote the element of 
$t$ corresponding to attribute $A_i$.  In the named
perspective, we denote a tuple $t:\{A_1,\dots,A_n\} \rightarrow D$
using an expression of the form $(A_1:d_1, \dots, A_n:d_n)$.

\subsection{Queries and Constraints across Schemas} 
\label{subsec-multiple-schemas}

Queries are usually 
%It is usually assumed that queries are 
defined with a particular schema
in mind, but in preparing and transforming data one sometimes has to deal
with queries that might be valid for several schemas.  
Consider, for instance, the relation  \locvisits introduced  in %the example in 
Section \ref{sec:intro}.  
In languages such as SQL, one can  retrieve 
the IDs of the facilities in % the relation
 \locvisits by issuing  the query 
\texttt{SELECT facility from locVisits}. This query can be applied over instances of multiple % several %relational 
schemas, as long as the schema has a relation 
\locvisits with attribute \fid. 

Since our goal is to model a framework where schemas may change depending on which 
procedures are applied to the data, we need the flexibility of being able to 
specify queries that may  be posed over multiple schemas. To formalize such queries,  
we assume the named perspective on schemas and data, as is explained next. 

%Example \ref{exa-intro-one}; to retrieve patient ages by patient IDs, we could pose on that relation the query \textit{LocVisists(age: x, patInsur: y).} Notice that to formulate this query, we do not need to be aware of the presence in the schema of \textit{LocVisists} of any other attributes of that relation.

%The named perspective on data provides the means for querying
%instances of different schemas.

\textbf{Named atoms:} 
A {\it named atom} is an expression of the form $R(A_1:x_1,\dots,A_k:x_k)$, where 
$R$ is a relation name, each $A_i$ is an attribute name, and each $x_i$ is a variable. We say that 
the variables mentioned by such an atom are $x_1,\dots,x_k$, and the attributes mentioned are 
$A_1,\dots,A_k$. 
A named atom $R(A_1:x_1,\dots,A_k:x_k)$ is \emph{compatible} with  
schema $\Sch$ if $\{A_1,\dots,A_k\} \subseteq \Sch(R)$. 

Given a named atom $R(A_1:x_1,\dots,A_k:x_k)$, an instance $I$ of schema $\Sch$ 
%that is compatible with the atom and an assignment $\tau: \{x_1,\dots,x_k\} \rightarrow D \cup \Nulls$, 
that is compatible with the atom, and an assignment $\tau:
\{x_1,\dots,x_k\} \rightarrow D$ assigning values to variables,  
we say that $(I,\tau)$ \emph{satisfies} $R(A_1:x_1,\dots,A_k:x_k)$ if there is a 
tuple $\bar a:\mathcal A \rightarrow D$
%\jd{``function'' was ``tuple'' in the original, but what is described
%  here is a function on an infinite domain, unlike the tuples
%  discussed earlier, and unlike many common uses of the word
%  ``tuple''.  Change it back if this usage is standard for the
%  audience.}
matching values in $\tau$ with attributes in $R$, in the sense that
$\bar a(A_i) = \tau(x_i)$ for each $1 \leq i \leq k$.  (Under the unnamed
perspective, we would require the presence of a tuple $a$ in $R^I$ such that its
projection $\pi_{A_1,\dots,A_k}\ \bar a$ over %the attributes
$A_1,\dots,A_k$ is precisely the tuple $\tau(x_1),\dots,\tau(x_k)$.)

%\noindent
\textbf{Conjunctive queries:} 
A \emph{conjunctive query} (CQ) is an expression of the form  
$\exists \bar z \phi(\bar z,\bar y)$, where $\bar z$ and $\bar y$ are tuples of variables, 
and $\phi(\bar z,\bar y)$ is a conjunction of named atoms that use the variables in $\bar z$ and $\bar y$. 
A CQ is compatible with $\Sch$ if all its named atoms are compatible. 

The usual semantics of conjunctive queries is obtained  from the semantics of named atoms in the usual
way.
Given a conjunctive query $Q$ that is compatible with $\Sch$, the result $Q(I)$ of evaluating $Q$ over $I$ 
is the set of all the tuples $\tau(x_1),\dots,\tau(x_k)$ such that $(I,\tau)$ satisfy $Q$. 

\textbf{Total queries:}  A \emph{total query}, which we define to be
an expression of the form $R$ for some relation name $R$, extracts all
the tuples stored in %the relation 
$R$, regardless of the schema and 
arity of $R$, as is done in SQL with  \texttt{SELECT *
  FROM R}.  A total query of this form is compatible with schema
$\Sch$ if $\Sch(R)$ is defined and the result of evaluating this
query over an instance $I$ over a compatible schema $\Sch$ is %simply
the set of all (unnamed) tuples in $R^I$.
%We use $CQ_*$ to denote the set of all conjunctive queries and all total queries. 

%We also extend CQs with an additional construct 
%used to differentiate constant values from null values. Formally, 
%let $C$ be a symbol not in $\mathcal{A}$, $\cal R$, $D$, or $\cal N$. A $C$-conjunctive query is a conjunctive query that may additionally use 
%the atoms $C(x)$, %and $\neg C(x)$, 
%for a variable $x$. The semantics is that an assignment $\tau$ satisfies $C(x)$ if $\tau(x)$ is a domain 
%value in $D$ (not a null in $\Nulls$). We use $CQ_C$ to denote the class of all $C$-conjunctive queries, and $CQ_{*,C}$ the class of all $C$-conjunctive queries and 
%all total queries.

%\noindent
\textbf{Data constraints:} 
We consider data constraints %considered in this paper take
that are (i) \emph{tuple-generating dependencies (tgds)}, i.e., expressions of the form 
$\forall \bar x \big(\exists \bar y\phi(\bar x,\bar y) \rightarrow \exists \bar z \psi(\bar x,\bar z)\big)$, for CQs %conjunctive queries
$\exists \bar y\phi(\bar x,\bar y)$ and $\exists \bar y\psi(\bar
x,\bar z)$, and (ii)  \emph{equality-generating dependencies (egds)}, i.e., expressions of the 
form $\forall \bar x \big(\exists \bar y\phi(\bar x,\bar y) \rightarrow x = x'\big)$, for a CQ %conjunctive query
$\exists \bar y\phi(\bar x,\bar y)$ and variables $x,x'$ in $\bar x$.  
As usual, for readability we sometimes omit the universal quantifiers of tgds and egds. 
%Furthermore, a $C$-tgd is a tgd built from $C$-conjunctive queries, and likewise for 
%$C$-egds. 

An instance $I$ \emph{satisfies} a set $\Sigma$ of tgds and egds, written $I \models \Sigma$, if 
(1) each CQ %conjunctive query 
in each dependency in $\Sigma$ is compatible with the schema of $I$, and (2) every assignment 
$\tau: \bar x \cup \bar y \rightarrow D$ 
such that $(I,\tau) \models \phi(\bar x,\bar y)$ 
can be extended into a %n assignment 
$\tau': \bar x \cup \bar y \cup \bar z \to D$ such that 
$(I,\tau') \models \psi(\bar x,\bar z)$.  

% An instance $I$ satisfies a set $\Sigma$ of ($C$)-tgds and
% ($C$)-tgds if (1) the schema of $I$ is compatible with each
% ($C$)-conjunctive query in each dependency in $\Sigma$, and (2) $I$
% satisfies $\Sigma$ in the usual logical way.

%Finally, we also treat \emph{total tgds}, which are constraints of the
%form $R \rightarrow S$, for relation names $R$ and $S$, that indicate
%that every tuple in $R$ must also be in $S$.
%Formally, an instance $I$ satisfies the constraint 
%$R \rightarrow S$ if the schema of $I$ has both $R$ and $S$, with the same number of attributes, and if the evaluation of the total query $R$ over $I$ is a subset of the evaluation of the total query $S$ over $I$. 

A tgd is \emph{full} if it does not use existentially 
quantified variables on the right-hand side, and {\em acyclic} if none of the relations on 
the right-hand side appear on the left-hand side.  %(that is, if $\bar z$ is empty). 
A set $\Sigma$ of tgds is {\em full} if each tgd in $\Sigma$ is full.
$\Sigma$ is {\em acyclic} if an acyclic graph is formed by
representing each relation mentioned in a tgd in $\Sigma$ as a node and
by adding an edge from node $R$ to $S$ if a tgd in $\Sigma$
mentions $R$ on the left-hand side and $S$ on the right-hand side.

%Moved to sec 3
%\noindent
%\textbf{Structure Constraints}:
%Structure constraints are used to specify that schemas need to contain a certain 
%relation or certain attributes. A structure constraint is a formula of the form 
%$R[\bar s]$ or $R[*]$, where $R$ is a relation symbol,  
%$\bar s$ is a tuple of attributes, and $*$ is a symbol not in 
%$\mathcal A$ or $\mathcal R$ intended to function as a 
%wildcard. 
%A schema $\Sch$ satisfies a structure constraint $R[\bar s]$, denoted 
%by $\Sch \models R[\bar s]$, if $\Sch(R)$ is defined, and each attribute in $\bar s$ belongs to $\Sch(R)$ 
%The schema satisfies the constraint $R[*]$ if $\Sch(R)$ is defined.
%For an instance $I$ over a schema $\Sch$ and a set $\Sigma$ of tgds, egds, and structure constraints, 
%we write $(I, \Sch) \models \Sigma$ if $I$ satisfies each data constraint in $\Sigma$ and $\Sch$ satisfies 
%each structure constraint in $\Sigma$. 

\section{Procedures under static schemas}\label{sec:procedures}

In this section we formalize the notion of procedures that transform
data.  We view procedures as black boxes, and assume no knowledge
of or control over their inner workings.  Our reasoning about 
procedures is based on the following information: The input conditions, or
\emph{preconditions,} on the state of the data that must hold for a procedure to be applicable; the output conditions, or
\emph{postconditions,} on the state of the data that must hold 
after an application of the procedure; and the set of relations affected by the application. 
To specify that some %pieces 
of the data will not be deleted, 
we also allow the inclusion of some queries whose answer needs to be preserved during the application 
of the procedure. 

\begin{example}
\label{exa-proc-1} 
Let us return to the procedure $P_\text{migrate}$ outlined in Section \ref{sec:intro}. The intent of $P_\text{migrate}$ is to define migration of 
 data from relation \textit{EVisits} into \textit{LocVisits}. %This procedure 
 $P_\text{migrate}$ can be described by the following information: 
 
\smallskip
\noindent
\underline{Scope}: Since $P_\text{migrate}$ %the procedure 
migrates tuples into $\textit{LocVisits}$, we specify that the 
%\emph{scope} of the 
procedure only changes %is just 
this relation.

\smallskip
\noindent
\underline{Precondition}: We specify that $P_\text{migrate}$ %this procedure 
requires a schema with relations 
$\locvisits$ and $\evisits$, both with attributes \fid, \pid and \timestamp. 

\smallskip
\noindent
\underline{Postcondition}: After $P_\text{migrate}$ is applied, it must be that each tuple in \evisits 
is in \locvisits. 

\smallskip
\noindent
\underline{Preserved queries}: The existing tuples in \locvisits are not deleted 
during the migration process. We specify this by stating that the answers to the query 
\texttt{SELECT facility, ssn, timestp FROM locVisits} are preserved in each application of $P_\text{migrate}$.  
\end{example}

In the following we present notation for formally defining these types of procedures. 
We start by introducing ``structure constraints,'' which we use to define the 
scopes of procedures. We will also use these in Section \ref{sec:dynamic}, when 
working with schema-altering procedures. % that alter schemas. % of databases. 

\subsection{Structure Constraints}

A structure constraint is a formula of the form 
$R[\bar s]$ or $R[*]$, where $R$ is a relation symbol,  
$\bar s$ is a tuple of attributes names from $\A$, and $*$ is a symbol not in 
$\mathcal A$ or $\mathcal R$ intended to function as a 
wildcard. 
A schema $\Sch$ satisfies a structure constraint $R[\bar s]$, denoted 
by $\Sch \models R[\bar s]$, if $\Sch(R)$ is defined
 and each attribute in $\bar s$ belongs to $\Sch(R)$. 
The schema satisfies the constraint $R[*]$ if $\Sch(R)$ is defined.

Given a set $\C$ of structure constraints and a  schema $\Sch$, we denote by 
$Q_{\Sch \setminus \C}$ the conjunctive query formed by the conjunction of the following atoms: 
\begin{itemize}
\item For each relation $R$ such that 
$\Sch(R) = \{ A_1,\dots,A_m \}$ but $R$ is not mentioned in $\C$, 
$Q_{\Sch \setminus \C}$ includes an atom $R(A_1:z_1,\dots,A_m:z_m)$, where  $z_1,\dots,z_m$ are fresh variables.
\item For each %relation 
$T$ %such that $T$ is 
mentioned in $\C$ but such that $T[*]$ is not in $\C$, 
$Q_{\Sch \setminus \C}$ includes an atom
$T(B_1:z_1,\dots,B_k:z_k)$, where $\{B_1,\dots,B_k\}$ is the set of
all the attributes in $\Sch(T)$ that are not mentioned in any constraint of the form $T[\bar s]$ in $\C$, 
and $z_1,\dots,z_k$ are fresh variables.
\end{itemize}

Intuitively, $Q_{\Sch \setminus \C}$ is intended to retrieve the projection of the entire database over all the 
relations and attributes not mentioned in $\C$.  Note that $Q_{\Sch \setminus \C}$ is unique up to the renaming of
 variables and order of conjuncts.  

As an example, let schema $\Sch$ have relations $R$, $S$, and $T$, such that 
 $R$ has attributes $A_1$ and $A_2$, $T$ has attributes $B_1$, $B_2$, and $B_3$, and 
 $S$ has $A_1$ and $B_1$. Let set $\C$ comprise constraints $R[*]$ and $S[B_1]$. 
 Then $Q_{\Sch \setminus C}$ is the query $T(B_1: z_1, B_2:  z_2, B_3: z_3) \wedge S(A_1: w_1)$. 

\subsection{Formal Definition of Procedures}

We define procedures w.r.t. %with respect to 
a class $\Const$ of FO constraints and a class $\Que$ of queries, but we mostly 
consider tgds, egds, structure constraints, and CQ %conjunctive
 queries.
\begin{definition}
A {\em procedure} $P$ over $\Const$ and $\Que$ is a tuple $(\Scope,\Cpre,\Cpost,\Q_\safe)$, 
where 
\begin{itemize}
\item $\Scope$ is a set of structure constraints that defines the scope (i.e., the relations and attributes) within which  
the procedure operates; 
\item $\Cpre$ and $\Cpost$ are constraints in 
$\Const$ that describe the pre- and postconditions of $P$, %the procedure,
respectively; and 
\item $\Q_\safe$ is a set of queries in $\Que$ that serve as a preservation %safety 
guarantee for the procedure. \end{itemize}
\end{definition}

\begin{example}[Example \ref{exa-proc-1} continued]
\label{exa-proc-2} 
The procedure $P_\text{migrate}$ is formally described as follows:  

 \smallskip
\noindent
\underline{$\Scope$}: The scope is the constraint $\textit{LocVisits}[*]$. 

\smallskip
\noindent
\underline{$\C_\pre$}: We use the structure constraints $\textit{EVisits}[\fid,$ $\pid,$ $\timestamp]$ and $\textit{LocVisits}$ $[\fid,$ $\pid,$  $\timestamp]$, to ensure that the data have the correct attributes.  

\smallskip
\noindent
\underline{$\C_\post$}: The postcondition comprises the tgd
\vspace{-5pt}
\begin{multline*}
\textit{EVisits}(\fid:x, \pid:y, \timestamp:z) \rightarrow \\ 
\textit{LocVisits}(\fid:x, \pid:y, \timestamp:z). 
\end{multline*}
 This says that, once $P_\text{migrate}$ has been applied, 
the projection of \textit{EVisits} over \fid, \pid, and \timestamp is a subset of the respective projection of \textit{LocVisits}.

%\smallskip
\noindent
\underline{$\Q_\safe$}: We use the query  
$\textit{LocVisits}(\fid:x,\pid:y,\timestamp:z)$, whose intent is to to state that all the answers to this query on \textit{LocVisits} that are present  before the application of $P_\text{migrate}$  
must be preserved. 
\end{example} 

\medskip
\noindent
\textbf{Semantics:} A procedure $P = (\Scope,\Cpre,\Cpost,\Q_\safe)$ is \emph{applicable} on an instance $I$ over schema $\Sch$ if 
%(1) The query  $Q_{\Sch \setminus \Scope}$ and each query in $\Q_\safe$ are compatible with $\Sch$, and 
(1) Each query in $\Q_\safe$ is compatible with $\Sch$; and 
(2) $I \models \C_\pre$. 
We can now proceed with the semantics of procedures. 
\begin{definition} 
\label{sec-procedure-semantics}
%Let $I,\Sch$ and $I',\Sch'$ be two instances with their corresponding schemas.  
Let $I$ be an instance over schema $\Sch$. 
An instance $I'$ over $\Sch$ is a {\em possible outcome} of applying 
procedure $P$ to $I$ if all of the following holds: 
\begin{enumerate}
\item $P$ is applicable on $I$; 
\item $I' \models \C_\post$; 
\item The answers of the query $Q_{\Sch \setminus \Scope}$ do not change: 
$Q_{\Sch \setminus \Scope}(I) = Q_{\Sch \setminus \Scope}(I')$; and 
\item The answers to each query $Q$ in $\Q_\safe$ over $I$ are preserved: 
$Q(I) \subseteq Q(I')$. 
\end{enumerate} 
\end{definition}

\begin{example}[Example \ref{exa-proc-2} continued]
\label{exa-proc-3}
Recall procedure $P_\text{migrate} = (\Scope,$ $\Cpre,$ $\Cpost,$ $\Q_\safe)$  defined in Example \ref{exa-proc-2}. %; we refer to $P_\text{migrate}$ as $P$ here. 
Consider instance $I$ over schema $\Sch$ with relations \evisits and \locvisits, each 
with attributes \fid, \pid, and \timestamp, as shown in Figure \ref{fig-exa-1} (a). 
Note first that $P_\text{migrate}$ is indeed applicable on $I$. 
When applying $P_\text{migrate}$ to $I$, we know from $\Scope$ that the only relation whose content can change is 
\locvisits, while \evisits is the same across all possible outcomes. 
Furthermore, we know from $\Cpost$ that in all possible outcomes the projection of 
\evisits over attributes \fid, \pid, and \timestamp must be the same as the projection of 
\locvisits over the same attributes. Finally, from $\Q_\safe$ we know that 
the projection of \locvisits over these three attributes must be preserved. 

Perhaps the most obvious possible outcome of applying $P_\text{migrate}$ to $I$ is that of the instance $J_1$ in Figure \ref{fig-exa-1} (b), 
corresponding to the outcome where the tuple in \evisits that is not 
yet in \locvisits is migrated into the latter relation. However, since we assume no control over the actions performed by  $P_\text{migrate}$, 
it may well be that it is also migrating data from a different relation that we are not aware of, 
producing an outcome whose relation \evisits is the same as in $I$ and $J_1$, but 
\locvisits has additional tuples, as depicted in Figure \ref{fig-exa-1} (c). 

Figure \ref{fig-exa-1} (d) depicts a situation where the postcondition is satisfied and, in addition, the schema is also altered due to the application of the procedure. We will revisit this situation in Section \ref{sec:dynamic}. 
%Moreover,
%and again since we do not control the procedure $P$, 
%it may also be the case that the procedure alters the schema of \locvisits, adding an extra attribute \age, importing the information from an unknown source, as shown in Figure \ref{fig-exa-1} (d). 
\end{example}  

\begin{figure*}[t]
\centering
{\small
\begin{tabular}{cc}
\fbox{\parbox{0.42\textwidth}{
\centering
\begin{tabular}{c}
\evisits  \\ 
\begin{tabular}{|ccc|}
\hline
\fid & \pid & \timestamp \\
\hline 
1234 & 33 & {070916 12:00} \\
2087 & 91 & {090916 03:10} \\
\hline
\end{tabular} \\
\locvisits \\ 
\begin{tabular}{|ccc|}
\hline
\fid & \pid & \timestamp \\
\hline
1234 & 33 & {070916 12:00} \\
1222 & 33 & {020715 07:50} \\
\hline
\end{tabular}
 \\ 
(a) Instance $I$\\
\end{tabular}}} & 

\fbox{\parbox{0.42\textwidth}{
\centering
\begin{tabular}{c}
\evisits  \\ 
\begin{tabular}{|ccc|}
\hline
\fid & \pid & \timestamp \\
\hline 
1234 & 33 & {070916 12:00} \\
2087 & 91 & {090916 03:10} \\
\hline
\end{tabular} \\
\locvisits \\ 
\begin{tabular}{|ccc|}
\hline
\fid & \pid & \timestamp \\
\hline
1234 & 33 & {070916 12:00} \\
1222 & 33 & {020715 07:50} \\
2087 & 91 & {090916 03:10} \\
\hline
\end{tabular}
 \\ 
(b) Possible outcome $J_1$ of applying $P_\text{migrate}$ to $I$\\
\end{tabular}}}
\end{tabular}

\smallskip

\begin{tabular}{cc}
\fbox{\parbox{0.42\textwidth}{
\centering
\begin{tabular}{c}
\locvisits \\ 
\begin{tabular}{|ccc|}
\hline
\fid & \pid & \timestamp \\
\hline 
1234 & 33 & {070916 12:00} \\
1222 & 33 & {020715 07:50} \\
2087 & 91 & {090916 03:10} \\
4561 & 54 & {080916 23:45} \\
\hline
\end{tabular} \\ 
(c) relation \locvisits in $J_2$
\end{tabular}}} & 
\fbox{\parbox{0.50\textwidth}{
\centering
\begin{tabular}{c}
\locvisits \\ 
\begin{tabular}{|cccc|}
\hline
\fid & \pid & \timestamp & \age \\
\hline 
1234 & 33 & {070916 12:00} & 21\\
1222 & 33 & {020715 07:50} & 45 \\
2087 & 91 & {090916 03:10} & 82 \\
\hline
\end{tabular} \\ 
(d) relation \locvisits in $J_3$
\end{tabular}}} 
\end{tabular}
}
\vspace{-5pt}
\caption{Instance $I$ of Example \ref{exa-proc-3} (a), a complete possible outcome of applying $P_\text{migrate}$ to $I$ (b), and 
the relation \locvisits of two other possible outcomes, one in which \locvisits contains additional tuples not mentioned in \evisits (c), and one where an extra attribute is added to \locvisits (d).}
\label{fig-exa-1}
\end{figure*}

As seen in Example \ref{exa-proc-3}, in general the number of possible outcomes that result from applying a procedure %is executed 
is infinite. %For this reason, 
Thus, we are generally more interested in properties shared by all possible outcomes, which motivates the following definition. 

\begin{definition}
The {\em outcome set} of applying a procedure $P$ to $I$ is defined as the set:
\begin{multline*}
\outcome_P(I) =\\
 \{I' \mid I' \text{ is a possible outcome of applying } P \text{ to } I \}.
 \end{multline*}
 The outcome of applying 
a procedure $P$ to a set of instances $\Inst$ is  the union of the outcome sets of applying $P$ to all the instances in $\Inst$:  
$$\outcome_P(\Inst) = \bigcup_{I \in \Inst} \outcome_P(I).$$ 
\end{definition}

Finally, since in general we are interested in (perhaps repeated) applications of multiple procedures, we extend the definitions  
to enable talking about the outcomes of sequences of procedures. 

\begin{definition}
The outcome of applying a sequence $\Pro_1,\dots,\Pro_n$ of procedures to instance 
$I$ is the set 
\begin{multline*}
\outcome_{\Pro_1,\dots,\Pro_n}(I) = \\ 
\outcome_{P_n}(\outcome_{P_{n-1}}(\cdots (\outcome_{P_1}(I)) \cdots )).
\end{multline*}
\end{definition}

\section{Fundamental Decision Problems}
\label{sec-fundamental-problems}

As promised, we begin our study with two decision problems on outcomes of sequences of procedures. 

\subsection{Applicability of Procedures}

In our proposed framework, the focus is on transformations of data sets given by sequences of procedures. 
Because we treat procedures as black boxes, the only description we have of the 
results of these transformations is that they ought to satisfy the output constraints of the procedures. 
In this situation, how can one guarantee that all the procedures in a given sequence will be applicable? 
Suppose that, for instance, we wish to apply procedures $P_1$ and $P_2$ 
to instance $I$ sequentially: First $P_1$, then $P_2$. 
The problem is that, since 
output constraints do not fully determine the outcome of $\Pro_1$ on $I$, we cannot 
immediately guarantee that this outcome %the outcome of applying $P_1$ over $I$ 
%is an instance that 
would satisfy  
the preconditions of $P_2$. 

Given that the set of outcomes is in general infinite, our focus is on  
guaranteeing that \emph{any} possible outcome of applying $P_1$ to $I$ will satisfy the preconditions 
of $P_2$. This gives rise to our first problem of interest:  
\vspace{-12pt}
\begin{center}
\fbox{
\begin{tabular}{ll}
\multicolumn{2}{l}{\textsc{Applicability:}}\\
\vspace{-8pt}
& \\
\textbf{Input}: & A sequence $\Pro_1,\dots,\Pro_n$ of procedures  \\
&  and a schema $\Sch$; \\
\textbf{Question}: & Can $P_n$ be applied to each instance \\
& in $\outcome_{\Pro_1,\dots,\Pro_{n-1}}(I)$, regardless \\ %no matter \\
& of the choice of instance $I$ of $\Sch$? % is chosen?
\end{tabular}
}
\end{center}
%It is not difficult to show that t
The \textsc{Applicability} problem is intimately related to the 
problem of implication of dependencies, defined as follows: Given a set $\Sigma$ of dependencies  
and an additional dependency $\lambda$, is it true that all the instances that satisfy $\Sigma$ also satisfy $\lambda$ ---  
that is, does $\Sigma$ imply $\lambda$?  
Indeed, consider a class $\mathcal L$ of constraints for which the implication problem is known 
to be undecidable. 
Then one can easily show that the applicability problem is also 
%The applicability problem is clearly 
undecidable for those procedures whose pre- and postconditions 
are in $\mathcal L$: Intuitively, if we let $\Pro_1$ be a procedure with a set $\Sigma$ 
of postconditions, and $\Pro_2$ a procedure with a dependency $\lambda$ as a precondition, then it is not difficult 
to come up with proper scopes and preservation queries so that the set 
$\outcome_{\Pro_1}(I)$ satisfies $\lambda$ for every instance $I$ over schema $\Sch$ if and only if 
$\lambda$ is true in all the instances that satisfy $\Sigma$. 

However, as the following result shows, the \textsc{Applicability} problem 
is undecidable already for very simple procedures, and even when we consider the \emph{data-complexity} 
view of the problem, that is, when we fix the procedure and take a particular input instance. 

\begin{proposition}
\label{prop-ap-und}
There are fixed procedures $\Pro_1$ and $\Pro_2$ that only use tgds for their constraints, and such that the following problem 
is undecidable: Given an instance $I$ over schema $\Sch$, is it true that all the instances in 
$\outcome_{\Pro_1}(I)$ satisfy the preconditions of $\Pro_2$? 
\end{proposition}
%\jr{Write proof, follows from results in \cite{APR13,KPT06}, see if something can be done with other combinations of dependencies.}

The proof of Proposition \ref{prop-ap-und} is by reduction from the embedding problem for finite semigroups, 
shown to be undecidable in \cite{KPT06}. %, and it is itself an adaptation of the proof of Theorem 7.2 in \cite{APR13}. 

There are several lines of work aiming to identify practical classes of constraints for which the implication problem 
is decidable, and we believe that the corresponding results can be applied in our framework. 
However, we  opt for a stronger restriction, %Since all of our examples so far use only structure constraints 
%as preconditions, for the remainder of the paper 
by focusing on procedures without preconditions. 
%whose preconditions comprise 
%structure constraints. 
In this setting, we have the following trivial result. 

%\begin{proposition}
%\label{ref-rep-decidable}
%\textsc{Applicability} is in polynomial time for sequences of relational procedures whose preconditions contain only structure constraints.
%\end{proposition}

\begin{fact}
\label{ref-ap-trivial}
If $P_1,\dots,P_n$ do not have preconditions, then {\sc Applicability} is always true, regardless of the schema. 
\end{fact}

When studying schema-altering procedures %that alter schemas 
in Section \ref{sec:dynamic}, 
we extend this class of procedures to include structure constraints as a means for 
specifying that certain procedures must be applied on schemas with certain properties. We do not do it here because 
structure constraints as preconditions do not make much sense under the static semantics of procedures. 

\subsection{Nonemptiness}
\label{ref-sec-nonempt}

%The other important problem at the time of considering the application of (sequences of) procedures is 
%to be able to gauge whether the outcome of such a sequence is nonempty. 
The other important problem is determination of whether
  the outcome of a sequence of procedures will be nonempty. 
We remark that even without preconditions, 
the outcome of a procedure may be empty if it is not possible to transform an instance in a way that would satisfy  
the postconditions of a procedure (and to ensure that the scope and preservation queries are respected). 
Perhaps surprisingly, we can show that this problem is undecidable even if we just have one fixed procedure. 

\begin{proposition}
\label{prop-rep-undec-weak}
There exists a procedure $\Pro$ that does not use preconditions and 
uses only tgds in its postconditions, such that the following is 
undecidable: Given an instance $I$, is the set 
$\outcome_{\Pro}(I)$ nonempty? 
\end{proposition}

%The reason we view Proposition \ref{prop-rep-undec} as a negative result is because it rules out the 
%possibility of using any ``reasonable'' representation system. Indeed, one would expect that deciding non-emptiness 
%should be decidable in any reasonable way of representing infinite sets of instances. 
%Proposition \ref{prop-rep-undec} is probably not surprising, since reasoning about tgds in general is 
%known to be a hard problem. Perhaps more interestingly, in our case one can show that the above fact remains true even if 
%one allows only \emph{acyclic} tgds, which are arguably one of the most well-behaved classes of dependencies in the literature. 
%The idea behind the proof is that one can simulate cyclic tgds via procedures with only acyclic tgds and no scope. 
The proof of this result makes use of arbitrary tgds. What is even more striking is that we can 
reproduce the undecidability proof even when considering acyclic tgds, albeit this time we need two fixed procedures. 

\begin{proposition}
\label{prop-rep-undec}
There exist procedures $\Pro_1$ and $\Pro_2$ that do not use preconditions and 
only use acyclic tgds in their postconditions, such that the following problem is 
undecidable: Given an instance $I$, is the set 
$\outcome_{\Pro_1,\Pro_2}(I)$ nonempty? 
 \end{proposition}

The idea of the proof is to manipulate the scope of procedures to create cases when we force some procedures to 
take care of transformations specified in subsequent procedures. We illustrate this idea with the following exaple. 
\begin{example}
\label{exa-scope}
Consider two procedures $\Pro_1$ and $\Pro_2$, where $P_1 = (\Scope^1, \Cpre^1,\Cpost^1, \Q_\safe^1)$, with 
$\Scope^1 = \{R[*],T[*]\}$, $\Cpre^1 = \emptyset$, $\Cpost^1 = \{R(A:x) \rightarrow T(A:x)\}$, and  $\Q_\safe^1 = R(A:x) \wedge T(A:x)$;  
$\Pro_2$ has empty scope, preconditions, and preservation queries, and has the postcondition set $\{T(A:x) \rightarrow R(A:x) \}$. 
Let $I$ be an instance over schema with relations $R$ and $T$, each with attribute $A$. By definition, 
the possible outcomes of applying $\Pro_1$ to $I$ are all the instances $J$ that extend $I$ and satisfy 
$R(A:x) \rightarrow T(A:x)$. Now the set $\outcome_{P_1,P_2}(I)$ corresponds to all the instances 
$I'$ that extend $I$ and satisfy both 
%all instances in the set $\outcome_{P_1,P_2}(I)$ actually satisfy both 
$R(A:x) \rightarrow T(A:x)$ and $T(A:x) \rightarrow R(A:x)$. (In other words, 
we can use $\Pro_2$  to \emph{filter out} all those instances $J$ where $T^J  \not\subseteq R^J$.) Intuitively, this happens 
because the outcome set of applying $\Pro_2$ to any instance not satisfying $T(A:x) \rightarrow R(A:x)$ is empty, and 
we define $\outcome_{P_1,P_2}(I)$ as the union of all the sets $\outcome_{P_2}(K)$ for each %instance 
$K \in \outcome_{P_1}(I)$. 
\end{example} 
%By applying the idea of this example to the proof of Proposition \ref{prop-rep-undec}, we show: 

\medskip
\noindent
\textbf{Procedures with safe scope:}
We could continue restricting the types of constraints we allow in procedures (for example, 
nonemptiness is decidable if we allow procedures made just from full, acyclic constraints). 
However, we choose to adopt a different strategy: 
We restrict the interplay between the postconditions  of procedures, their scope, and their preservation queries. This will 
allow us to rule out the undesirable behaviours illustrated in Example \ref{exa-scope}, and will become 
crucial at the time of reasoning about sequences of procedures.

%Since restricting the classes of dependencies allowed in procedures 
%may not be enough to guarantee outcomes that can be represented by reasonable systems, we now adopt a different strategy: 
%We restrict the interplay between the postconditions  of procedures, their scope, and their preservation queries. This will 
%allow us to rule out the undesirable behaviours illustrated in Example \ref{exa-scope}. 
%
We say that procedure $P = (\Scope,\Cpre,\Cpost,\Q_\safe)$ has \emph{safe scope} if the following 
holds: 
\begin{itemize}
\item $\Cpre$ is empty; 
\item $\Cpost$ is an acyclic set of tgds; 
\item The set $\Scope$ contains exactly one constraint $R[*]$ for each relation $R$ that appears 
on the right-hand side of a tgd in $\Cpost$; and 
\item The set $Q_\safe$ contains one total query $R$ for each constraint $R[*]$ in $\Scope$. 
That is, it binds precisely all the 
relations in the scope of $P$. 
\end{itemize}
%(For instance, procedure $P$ in Example 
%\ref{exa-proc-1} is essentially a procedure with safe scope, as it can 
%easily be transformed into one by slightly altering the safety query.) 

%We also define a class of procedures that ensure that certain attributes or relations  
%be present in the schema. Formally, we say that a 
%procedure $P = (\Scope,\Cpre,\Cpost,\Q_\safe)$ is an \emph{alter-schema procedure} if the following holds: 
%\begin{itemize}
%\item Both $\Scope$ and $\Q_\safe$ are empty; and 
%\item $\Cpost$ is a set of structure constraints. 
%\end{itemize} 

%Let $\P^{\textit{safe},\textit{alter}}$ be the class of all the procedures that are either safe scope or alter-schema procedures. 
%The class $\P^{\textit{safe},\textit{alter}}$ allows for practically-oriented interplay between migration and schema-alteration tasks  
%and, as we will see in this section, is more manageable from the point of view of reasoning tasks, in terms of complexity. 
%To begin with, deciding the non-emptiness of a sequence of procedures is essentially tractable for $\P^{\textit{safe},\textit{alter}}$: 

Note that the procedure $P_\text{migrate}$ of Example \ref{exa-proc-2} is not a procedure with safe scope, but can 
easily be transformed into a procedure with safe scope. 
Once again we have an easy result that makes a case for the good behaviour of procedures with safe scope. 
%The proof 
 %is based on the idea of chasing instances with the dependencies in the procedures, in order to create 
% a minimal instance. We revisit this idea in the following sections, when considering querying the outcomes of procedures.  

\begin{proposition}
\label{prop-rep-safe-scope}
For every instance $I$ and sequence $\Pro_1,\dots,\Pro_n$ of procedures with safe scope, the set 
$\outcome_{\Pro_1,\dots,\Pro_n}(I)$, is not empty. 
\end{proposition}

\section{Analyzing the outcomes of procedures} 
\label{analyze-sec}

We have seen that deciding properties of outcomes of sequences of procedures (or even of a single procedure)  
can be a nontrivial task. One of the reasons is that procedures do not completely define their outcomes: We 
do not really know the outcome of applying a sequence $\Pro_1,\dots,\Pro_n$ of procedures to an 
instance $I$, we just know that it will be an instance from the collection $\outcome_{\Pro_1,\dots,\Pro_n}(I)$. 
This collection may well be of infinite size, but can it still be represented finitely? The backdrop to this question is the variety of formalisms that have been %successfully 
developed by the community  for representing sets of instances, from tables with incomplete information \cite{IL84} to 
knowledge bases (see, e.g., \cite{BO15}). 
%database-theory 
%has  

We now focus on developing a 
representation for outcomes of (sequences of) procedures, so that the usual data-oriented tasks could be performed over these outcomes. 
In pursuing this objective, we also connect our framework with the important topics of %in the literature:
 knowledge bases and incomplete 
information. We also show how our %proposed 
framework can be used to formalize and study new natural problems related to these areas.

%%need to improve this explanation
Due to the results in the previous sections, we shall focus mostly on procedures with safe scope. 

\subsection{Representing Outcomes and Evolution of Knowledge Bases}

The first question we ask is related to the representation of the outcome of a sequence of procedures: 
Is it possible to represent the set $\outcome_{P_1,\dots,P_n}(I)$ for an instance $I$ and  
sequence $P_1,\dots,P_n$ of procedures? 

In representing instances, we use the notion of representation system, understood as finite representation of an 
infinite set of instances. Following \cite{IL84}, a \emph{representation system} is a set $\Reps$ of representatives and a 
function $\rep$ that assigns a set of instances to each representative in $\Reps$. For now, we will assume that the 
function $\rep$ is uniform, in the sense that for each $W \in \Reps$ the function $\rep(W)$ maps $W$ to instances of the same schema. 
We then say that a set $\mathcal I$ of instances can be represented by a representation system $(\Reps,\rep)$ 
if there is a representative $W \in \Reps$ such that $\rep(W) = \mathcal I$. 

%In general terms, a \emph{knowledge base} is a tuple $(I,\Gamma)$, where $I$ is a database instance and 
%$\Gamma$ is a set of constraints over the schema of $I$. This knowledge base naturally represents 
%a set $\rep(I)$ of instances, given by $\rep(I) = \{J \mid I \subseteq J$ and $J \models \Gamma\}$. 

To represent outcomes of procedures, we extend the usual notion of knowledge base, so that we can define a relational scope 
over knowledge bases. That is, we use the following representation system: 

\begin{definition}
A \emph{scoped knowledge base} (SKB) over a schema $\Sch$ is a tuple $\skb = (I,\Gamma, \Scope)$, where $I$ is  an 
instance over $\Sch$, $\Gamma$ is a set of constraints, and $\Scope$ is a set of relation names in $\Sch$. 
%, that is, such that $S(R)$ is defined for each $R \in \Scope$. 
%Scoped knowledge bases can be seen as a representation system. For an SKB $\skb = (I,\Gamma, \Scope)$, the 
Intuitively, %an SKB 
$\skb$ represents all possible extensions of $I$ that satisfy $\Gamma$ 
and does not modify relations not in the scope. That is, 
\begin{multline*}
\rep(\skb) = \{J \mid I \subseteq J, J \models \Gamma \\ \text{ and } R^I = R^J \text{ for each relation }R \notin \Scope\}.
\end{multline*}
\end{definition}
%For now just doing static schemas, this goes afterwards when we generalise supersets
%Note that, as usually, we do not impose a particular schema for $J$: it could incorporate more relations than those of $I$. 

The first observation is that SKBs are an appropriate representation system for capturing the outcomes of procedures with safe scope.

\begin{proposition}
\label{prop-outcome-fulltgds-justone}
Let $P = (\Scope, \emptyset,\Gamma, \Qsafe)$ be a procedure with safe scope, and $I$ an instance such that 
$P$ can be applied over $I$. Then $\outcome_P(I) = \rep(\skb)$, where $\skb$ is the scoped knowledge base $(I,\Gamma,\Scope)$. 
\end{proposition}

However, we can do much more with SKBs, as this formalism is also appropriate for representing the outcomes of sequences of procedures with safe scope. 
We show this next for the case of procedures with safe scope whose %we also require procedures to have 
constraints are given by full tgds. 

\begin{theorem}
\label{theo-outcome-fulltgds}
Let $P_1,\dots,P_n$ be a sequence of procedures with safe scope, and such that each of $P_1,\dots,P_n$  uses only full tgds. 
Then for every instance $I$, the set $\outcome_{P_1,\dots,P_n}(I)$ can always be represented by an SKB. Furthermore, this 
SKB can be computed in exponential time from $P_1,\dots,P_n$ and $I$.  
\end{theorem}

A natural question at this point is whether scoped knowledge bases are \emph{closed} under the operations 
expressed by procedures: If we start with an SKB $\skb$ and apply a procedure $P$ to each instance represented 
by $\skb$, is it possible to represent $\outcome_P(\rep(\skb))$? Surprisingly, it turns out that the answer %to this question 
is positive, as long as the procedures and the scoped knowledge base 
satisfy the requirements of Proposition \ref{theo-outcome-fulltgds}. 
An SKB $(I,\Gamma,\Scope)$ is \emph{full} if $\Gamma$ is a set of full tgds, \emph{acyclic} if 
$\Gamma$ is acyclic, and \emph{safe} if $\Scope$ contains at least all the relations on the right-hand side of the tgds in $\Gamma$. 
The next result shows that full acyclic safe SKBs are a \emph{strong representation system} \cite{IL84} for the outcomes of such procedures. 

%%%check whether we need acyclicity here

\begin{theorem}
\label{theo-skb-closed}
Let $P$ be a procedure with safe scope using only full tgds. 
The for every full acyclic safe SKB $\skb$, the set $\outcome_{P}(\rep(\skb))$ can always be represented by a full acyclic safe SKB. 
\end{theorem}

The proof relies on the fact that each application of a procedure not only changes the scope, but also forces us to update 
or drop some of the knowledge that we had before. The following example illustrates this issue: 

\begin{example}
Consider two procedures with safe scope, $P_T$ with scope $T[*]$ and 
postcondition $\Gamma_T = \{R(x) \rightarrow T(x), S(x) \rightarrow T(x)\}$, and  
$P_R$ with scope $R[*]$ and postcondition $\Gamma_R = \{S(x) \rightarrow R(x)\}$. 

Consider instance $I$ with $T^I = S^I = \emptyset$ and $R^I = \{1\}$.  
We can show that all the instances in $\outcome_{P_T}(I)$ must satisfy $\Gamma_T$. One possible outcome of applying $P_T$ to $I$ is instance 
$J$ with $R^J = T^J = \{1\}$ and $S^J = \emptyset$. 

Consider now $\outcome_{P_T, P_R}(I)$; it may include instances that do not satisfy the tgd $R(x) \rightarrow T(x)$. For example, 
one can verify that the instance $K$ given by $R^K = \{1,2\}$, $T^K = \{1\}$, $S^K = \emptyset$ is a possible outcome 
of applying to $I$ the procedure $P_T$ followed by $P_R$. 

What happens here is that, in a sense, we lose the information given by $R(x) \rightarrow T(x)$ when applying $P_R$, because $R$ is in the scope of 
$P_R$. %In fact, t
The set  $\outcome_{P_T, P_R}(I)$ is represented by the SKB $\skb = (J,\Gamma,\{R,T\})$, where 
$\Gamma$ contains both $S(x) \rightarrow T(x)$ and $S(x) \rightarrow R(x)$.
\end{example}

To the best of our knowledge, this interaction between knowledge bases, scopes, and relational procedures has not been studied before. 
We do not know whether there are other languages with the above property, or  to what degree Theorem 
\ref{theo-skb-closed} can  be extended to more expressive languages. We believe these are interesting questions 
to study in future work. We finish with a remark about the unavoidable exponential blowup when representing outcomes of 
procedures with SKBs. 

\begin{proposition}
\label{prop-size-representation}
There is an instance $I$, a procedure $P'$, and a family $\big(P_i \big)_{i \geq 1}$ of procedures, all with safe scope, such that 
every SKB $(J,\Gamma,\Scope)$ representing $\outcome_{P_i,P'}(I)$ is such that $J$ has at least $2^i$ tuples. 
\end{proposition}

%%if time, try to extend this to non-ful tgds! 
%\comment{Juan}{What about  querying and/or deciding stuff about scoped KBs? next subsection!}

\subsection{Reasoning About Incomplete Instances With Open and Closed Relations}

Another important problem is to reason about properties satisfied by the outcomes of (sequences of) 
procedures. With Theorem \ref{theo-outcome-fulltgds} 
at hand, we focus  on scoped knowledge bases, as these are sufficient for capturing the outcomes of a wide range of 
procedures. However, our results are of independent interest as they go beyond full acyclic and safe SKBs. We study the following problem: 
\begin{center}
\fbox{
\begin{tabular}{ll}
\multicolumn{2}{l}{\textsc{Satisfaction of Constraints:}}\\
\vspace{-12pt}
& \\
\textbf{Input}: & Acyclic scoped knowledge base $\skb $ and \\
& constraint $\sigma$, both over a schema  $\Sch$; \\
\textbf{Question}:\!\!\!\!\! & Is $\sigma$ satisfied on each instance in $\rep(\skb)$?
\end{tabular}
}
\end{center} 

In this paper we focus on the cases of the above problem where $\sigma$ is an egd or a tgd. %equality-generating dependency (or a tuple-generating dependency. 
Just as with \textsc{Applicability,} we know that we need some restrictions on $\Gamma$ in SKBs, 
because allowing arbitrary tgds would make the problem undecidable. For the next sections, we choose to go with acyclic sets of tgds. %, 
%and work with them in the following sections. 
However, note that our results do not need tgds to be full or SKBs to be safe.

\medskip
\noindent
\textbf{Satisfaction of egds}. 
Our first result is positive, stating that SKBs are as well behaved as other representation systems as far as reasoning about 
equality-generating dependencies is concerned. 

\begin{proposition} 
\label{egd-prop} 
The problem \textsc{Satisfaction of Egds} is $\conp$-complete when restricted to acyclic SKBs.
\end{proposition}  

As with the problem of implication of egds over other forms of incomplete databases, the proof consists on 
showing that $\sigma$ is not valid on every instance in $\rep(\skb)$ if an only if there is a small instance constructed from 
$I$ and the frozen body of $\sigma$ where $\sigma$ does not hold. The difference is that we need to be more careful 
when constructing the frozen body of $\sigma$, to take into account the scope of the SKB. 

%%% complexity!!!!!
%In the formulation of Proposition \ref{egd-prop}, any set of dependencies $\Gamma$ is acceptable as long as it belongs to a class of dependencies that gives a sufficient condition for chase termination with $\Gamma$.  The intuition for the proof is 
%in the small-witness property: It is possible to explicitly generate a finite set of instances, such that either (i) the satisfaction of the input egd $\sigma$ is disproved directly on one of these instances, or (ii) the satisfaction of $\sigma$ on all the instances represented by $SKB$ follows from its satisfaction on all the generated instances. Likewise, it is not diffi

\medskip
\noindent
\textbf{Satisfaction of tgds}.
For the case of tgds we cannot immediately adapt previous results, as the scope of SKBs prevents us 
from applying the chase in a direct way. On the other hand, the problem becomes much easier once we disallow 
unsafe knowledge bases. 

%However, this time we cannot make the statement general enough to include all possible sets of dependencies, since these dependencies may interact with $\sigma$ and 
%make the overall process undecidable. Instead, we choose to focus on the case of  
%sets of weakly-acyclic \cite{FaginKMP05} tgds and egds. %  for a precise definition).

%This is out of context in the current writing (we have defined SKBs so that every relation outside scope is actually ground)
%We now assume that each relation outside $\Scope$ in $(I,\Gamma, \Scope)$ is a ground relation, i.e., no nulls are permitted in such relations. (This restriction could be reasonable in, e.g., those settings where all the available procedures  would be applicable to the same relation. The relation in question would therefore be the only relation in the scope of all the procedures.) 
%\begin{proposition} 
%\label{tgd-entailment-decidability-prop} 
%Given a scoped knowledge base $SKB$ $=$ $(I,\Gamma, \Scope)$ over schema  $\Sch$ such that each relation outside $\Scope$ in $(I,\Gamma, \Scope)$ is a ground relation (i.e., no nulls permitted), and such that $\Gamma$ is a weakly acyclic \cite{FaginKMP05} set of egds and tgds. Then, for a tgd  $\sigma$ over schema  $\Sch$,  it is decidable to check whether $\sigma$ holds on each instance in $(I,\Gamma, \Scope)$. 
%\end{proposition} 

\begin{proposition} 
\label{tgd-entailment-decidability-prop} 
The problem \textsc{Satisfaction of Tgds} is in $\Pi^2_p$ when restricted to acyclic SKBs. 
%SKBs $(I,\Gamma, \Scope)$ such that 
%$\Gamma$ is a weakly acyclic set of tgds. 
It is \np-complete if further restricted to safe SKBs. 
\end{proposition} 

The proof follows from  the results of \cite{ChirkovaY14}, which explores the decidability and complexity of determining containment of CQ queries in presence of materialized views and of dependencies on the instance that generates the materialized views. The results of \cite{ChirkovaY14} build on \cite{ZhangM05} and are obtained for the closed-world case, i.e., all the given materialized views are exact. In the proof of Proposition \ref{tgd-entailment-decidability-prop} in this  paper, 
we explore a setting that is partially closed world, as we use a particularly simple case of exact materialized views to model relations that are not in the scope of an SKB. At the same time, our results differ from, and extend in part, those of \cite{ChirkovaY14}, as our setting is also partially open world, as expressed by the relations that are in scope in the given SKB. 

We remark that both Propositions \ref{egd-prop} and \ref{tgd-entailment-decidability-prop} continue to hold if we allow the knowledge 
set $\Gamma$ in SKBs to be given by more expressive formalisms such as weakly acyclic \cite{FKMP05} sets of tgds and even 
including egds, or virtually any other formalism that guarantees presence of polynomial-length witnesses of 
chase sequences.

\medskip
\noindent
\textbf{Query Answering}.
The final task we consider is query answering: Given an instance
$I$, a sequence $P_1,\dots,P_n$ of procedures, and a query $Q$, we
would like to find the answers to $Q$ over the outcomes of applying
$P_1,\dots,P_n$ to $I$.  To simplify the presentation, we consider
only boolean CQ queries. %, and we only consider conjunctive queries. 
As usual in cases that deal
with infinite sets of instances, we are interested in the
\emph{certain} (or unambiguous) answers, which in our framework
translate into the answers that hold over any possible outcome of
applying $P_1,\dots,P_n$ to $I$.

More formally, let $\certain_{P_1,\dots,P_n}(Q,I)$ denote the
intersection $$\bigcap_{J \in \outcome_{P_1,\dots,P_n}(I)} Q(J).$$

\begin{center}
\fbox{
\begin{tabular}{ll}
\multicolumn{2}{l}{\textsc{Query answering}}\\
\vspace{-12pt}
& \\
\textbf{Input}: & Instance $I$, boolean CQ $Q$, and \\
& sequence $P_1,\dots,P_n$ of procedures; \\
\textbf{Question}: & Is $\certain_{P_1,\dots,P_n}(Q,I)$ nonempty?
\end{tabular}
}
\end{center} 

The complexity of query answering is not high when compared to similar reasoning tasks in other 
similar scenarios such as data exchange or ontology-based query answering. 
The problem is even in polynomial time if we consider %instead 
the \emph{data complexity} of the problem. 

\begin{proposition}
\label{prop-query-answering}
\textsc{Query answering} is in \nexptime when restricted to  
sequences of procedures with safe scope. It is in $\exptime$ if 
the procedures involved contain only full tgds, and in polynomial time if 
$P_1,\dots,P_n$ and $Q$ are fixed.\end{proposition}

Interestingly, this result extends the settings we can represent with Theorem \ref{theo-outcome-fulltgds}, as we allow 
any sequence of procedures with safe scope. The reason we can work with more expressive outcomes is that for computing certain 
answers to conjunctive queries, we only need to keep the set of minimal instances in the set of outcomes, instead 
of representing the complete outcome space of a sequence of procedures. We show that this set of minimal instances 
can be represented with an extension of conditional tables, and for each subsequent application of a procedure 
we update this set by chasing the conditional table, as is done in, e.g., \cite{APR13}. 
Although we can easily get \pspace-hardness from a reduction from the query answering problem for non-recursive 
datalog \cite{VV98}, we do not know if the bounds presented above are tight. 

\section{The data-readiness problem}
\label{readiness-sec}

We now address the problem of assessing achievability of data-quality constraints, which we describe informally as follows. 
We start with an instance $I$ and have a set $\Pi$ of procedures. We are also given 
a property $\alpha$ over instances (specified, for example, as a boolean query or a set of constraints) that 
does not hold in $I$. 
The question we ask in this setting 
is whether we can apply  to $I$ some or all the procedures in $\Pi$ so that 
all the instances in the %resulting 
outcome satisfy $\alpha$:  
%In this is the case, we say that $I$ can be readied for $\alpha$ using $\Pi$. 
%The formal statement of the problem follows. 

\begin{center}
\fbox{
\begin{tabular}{l}
{\textsc{data readiness:}} \\
\begin{tabular}{ll}
\textbf{Input}: & An instance $I$, a set $\Pi$ of procedures,  \\ 
&and a property $\alpha$; \\
\textbf{Question}:\!\!\!\! & Is there a sequence $P_1,\dots,P_n$ of pro- \\
& cedures in $\Pi$ such that all the instan-  \\  
& ces in $\outcome_{P_1,\dots,P_n}(I)$ satisfy $\alpha$?
\end{tabular} \\
\end{tabular}
}
\end{center}

The length of $n$ of the assembled workflow is not part of the input, but 
instead needs to be derived from the system as well. This implies a striking difference 
between readiness and the problems we study in Section \ref{analyze-sec}, and will lead us 
to the undecidability of problems that would instead be decidable if $n$ was considered to 
be fixed. In those cases when we obtain decidability it is by proving small-model properties about the length of the sequences 
that need to be assembled. 

We study two specific versions of the problem, arising from how we  
specify $\alpha$. We begin with the case where $\alpha$ is specified as either a tgd or an egd, 
which we denote as \textsc{constraint readiness}. We then study the case 
where $\alpha$ is a boolean CQ $Q$, denoted as \textsc{query readiness}.

\subsection{Readying Data with Respect to Constraints}

In the previous sections we have seen that fundamental problems in our proposed framework can be solved if we 
restrict ourselves to procedures with safe scope. Unfortunately, as the following proposition shows, 
this is not the case for the data-readiness problem. 

\begin{proposition}
\label{prop-readi-undecidable-constraints}
The problem \textsc{constraint readiness} is undecidable, even if $\Pi$ is a set of procedures with safe scope and 
$\Sigma$ contains a single acyclic tgd. 
\end{proposition}

The proof is by reduction from the universal halting problem for Turing machines, along the lines of the proof 
used in \cite{DNR08} to show that the termination of chase is undecidable. The proof uses $\Pi$ to simulate each of the constraints being 
chased in the proof in \cite{DNR08}. %by Deutsch et al. 

On the other hand, the results of Section \ref{analyze-sec} suggest a way to solve the problem for sets of procedures 
with safe scope that use full tgds only. Intuitively, it suffices to guess a sequence of procedures, compute the 
representation of their outcome, and then apply Proposition \ref{tgd-entailment-decidability-prop}. 
All that is left to do is to prove a small-model property for the size of the sequence of procedures that we need to guess, as specified in the next  
result. 

\begin{theorem}
\label{theo-readi-decidable-constraints}
The problem {\textsc{constraint readiness}} is decidable in \twonexptime, for the cases where $\Pi$ is a set of procedures with safe scope 
with output constraints comprising full tgds only. 
\end{theorem}

%%% add data complexity!!!! 

\subsection{Readying Data with Respect to Queries}

We now study data readiness with respect to boolean CQs. 
Again, it is not enough to restrict our consideration to procedures with safe scope, as we can modify 
Proposition \ref{prop-readi-undecidable-constraints} to obtain undecidability for this case. 

%, even if $\Pi$ is a set of procedures with safe scope and $\Sigma$ contains only tgds. 

\begin{proposition}
\label{prop-readi-undecidable-queries}
The problem {\textsc{query readiness}} is undecidable, even if $\Pi$ is a set of procedures with safe scope and $Q$ 
is a query of one atom. 
\end{proposition}

Again, we can manage this problem if we only consider procedures with safe scope given by full tgds.

\begin{theorem}
\label{theo-readi-decidable-queries}
The problem {\textsc{query-readiness}} is in \nexptime
for the cases where $\Pi$ is a set of procedures with safe scope 
with output constraints comprising full tgds only. 
\end{theorem}

In the following section we will extend this result in a significant way, by showing that  
\textsc{query readiness} continues to be decidable even if we add to $\Pi$ procedures that alter the schema of instances.

\section{Procedures over dynamic schemas} \label{sec:dynamic}

So far we have assumed that databases maintain their schemas through the history of procedures applied to them. 
One could argue that this is not a reasonable assumption, because schemas may and will eventually change along 
with the needs of the data. 

\begin{example}
\label{exa-motivate-dynamic}
In the motivating example in Section \ref{sec:intro} we sketched $P_\text{insur}$, a procedure that 
would augment relation $\locvisits$ with an attribute $\inid$ containing the 
insurance information of patients, with data drawn from  relation $\patients(\pid,\inid)$. 
For this case, the schema of the outcomes must be different from the schema of the input relation. Thus, to be able to capture this scenario,    
we need a different notion of outcome. 
\end{example}

In this section, we revisit the previously introduced notions, to show that our framework can be modified 
to include sequences of transformations under dynamic schemas. We will see that, perhaps surprisingly, most of the results 
of the previous sections still hold under the new, more general definitions.

In this direction of the work, 
using the named perspective on queries and databases will prove critical, as it allows 
us to define queries that hold under several schemas (recall Section \ref{subsec-multiple-schemas}). 
We will then be able to define procedures without a particular 
schema in mind,  requiring only that the conditions and queries therein remain compatible with the input schema. 

So how can we redefine procedures to be applicable to more than one schema? Following 
the same path that we took when defining static outcomes, we now define a dynamic version of the outcome set, 
wherein we treat schemas as other open propositions. We continue to treat procedures as black boxes, and now permit them to update schemas 
as well. We thus allow the outcomes to have any schema, as long as they satisfy certain compatibility restrictions that we now 
introduce. 

\subsection{Initial Definitions}

The first thing we need to do is to extend some of our notation to accommodate the new scenario. 
Somewhat abusing the notation, we will use $\schema(I)$ to denote the schema of an instance $I$. 

Let $\Sigma$ be a set of structure constraints and data constraints. We say that 
$I$ satisfies $\Sigma$, and write again $I \models \Sigma$, if (1) $\schema(I)$ satisfies the structure constraints in $\Sigma$,  
and (2) $I$ satisfies the data constraints in $\Sigma$. We recall that (2) holds only when $\schema(I)$ is compatible 
with all the data constraints in $\Sigma$. 

A schema $\Sch'$ extends a schema $\Sch$ if for each relation $R$ such that $\Sch(R)$ is defined, 
we have $\Sch(R) \subseteq \Sch'(R)$.
That is, $\Sch'$ extends $\Sch$ if $\Sch'$ assigns at least the same attributes to all the relations in $\Sch$.
An instance $J$ \emph{extends} an instance $I$ if (1) $\schema(J)$ extends $\schema(I)$, and (2)  
for each relation $R$ in $\schema(I)$ with assigned attributes $\{A_1,\dots,A_n\}$ and for each tuple $t$ in $R^I$, 
there is a tuple $t'$ in $R^{J}$ such that $t(A_i) = t'(A_i)$ for each $1 \leq i \leq n$.  Intuitively, $J$ 
extends $I$ if the projection of $J$ over the schema of $I$ is contained in $I$.

We now define the {\em dynamic semantics} for procedures. The definition is the same as that of Definition \ref{sec-procedure-semantics}, except that we now allow instances of different schemas to be in the set of a procedure's outcomes. 
\begin{definition}
%Let $I,\Sch$ and $I',\Sch'$ be two instances with their corresponding schemas.  
Let $I$ be an instance over a schema $\Sch$. 
An instance $I'$ over schema $\Sch'$ is a {\em possible outcome} of applying 
$P$ to $I$ \emph{under the dynamic semantics} if the following conditions hold: 
\begin{enumerate}
\item $P$ is applicable on $I$; 
\item $I' \models \C_\post$; 
\item The answers of the query $Q_{\Sch \setminus \Scope}$ do not change: 
$Q_{\Sch \setminus \Scope}(I) = Q_{\Sch \setminus \Scope}(I')$; and 
\item The answers to each query $Q$ in $\Q_\safe$ over $I$ are preserved: 
$Q(I) \subseteq Q(I')$. 
\end{enumerate} 
\end{definition}

In the definition, we state the schemas of instances $I$ and $I'$ explicitly, to reinforce the 
fact that schemas may change during the application of procedures. 
However, most of the time the schema can be understood from the instance, so we normally 
just say that an instance $I'$ is a possible outcome of $I$ under the dynamic semantics (even if the schemas of $I$ and $I'$ are different). 
%Let us also recall that we use $\schema(I)$ to denote the 
%schema of an instance $I$. 

The set of outcomes for a procedure $P$ when applied on an instance $I$ under the dynamic semantics can be thus stated as 
\begin{multline*}
\dynoutcome_P(I) = \{I' \mid I' \text{ is a possible outcome} \\ \text{ of applying } P \text{ under the dynamic semantics to } I \}, 
 \end{multline*}
 
 \noindent 
 and we likewise extend the notion of dynamic outcomes to a set of instances (now over possibly different schemas) and for a 
 sequence of procedures, just as we did in the previous sections. % of the paper. 

\begin{example}[Example \ref{exa-proc-3} continued]
\label{exa-proc-4} 
Recall the definition of procedure $P_\text{migrate}$ in Example \ref{exa-proc-2}. We explained that the instances $J_1$ and $J_2$ in Figure 
 \ref{fig-exa-1} (b) and (c) belong to the set of outcomes of the instance depicted in Figure \ref{fig-exa-1} (a). 
Both $J_1$ and $J_2$ also belong to the set of outcomes under the dynamic semantics. However, this set also contains instances 
over schemas that extend the schema of $I$, such as the instance $J_3$ in Figure  \ref{fig-exa-1} (d). Note that 
the relation \locvisits in $J_3$ has an additional attribute (\age). 
\end{example}

The  notion of dynamic semantics allows us to model, for instance, database procedures that add new 
attributes to relations. 

\begin{example}[Example \ref{exa-motivate-dynamic} continued]
\label{p-insur-complete}
Procedure $P_\text{insur}$ from the example in Section \ref{sec:intro} is formally described as follows:  

 \smallskip
\noindent
\underline{$\Scope$}: The scope is the constraint $\textit{LocVisits}[*]$; 

\smallskip
\noindent
\underline{$\C_\pre$}: The precondition is empty; 

\smallskip
\noindent
\underline{$\C_\post$}: The postcondition comprises the egd
\vspace{-5pt}
\begin{multline*}
\patients(\pid:x, \inid:y)\ \wedge \\ 
\locvisits(\pid: x, \inid: z) \rightarrow y = z ;  
\end{multline*}
(Note that instances compatible with this egd must each have an attribute \inid in relation 
\locvisits. This constraint also says that the insurance ID of patients in \locvisits must correspond 
to that in relation \patients.)  

\smallskip
\noindent
\underline{$\Q_\safe$}: We use again the query  
$\textit{LocVisits}(\fid:x,\pid:y,\timestamp:z)$; regardless of any newly added attributes, we need to maintain the 
projection of \locvisits onto these three attributes. 
\end{example}

\subsection{Fundamental Reasoning Tasks}

Just as with the static semantics, we begin by studying the problems of applicability 
and non-emptiness under the dynamic semantics. This is not just a trivial extension from 
what we had in Section \ref{sec-fundamental-problems}. Indeed, the use of dynamic semantics 
allows for defining procedures for which each instance in the set of outcomes has a schema that is different from the schema of the original 
instance. (See, e.g.,  procedure $P_\text{insur}$ of Example \ref{p-insur-complete}.)

\medskip
\noindent
\textbf{Applicability}. We know that applicability is undecidable under reasonable expectations, unless we restrict ourselves to procedures without preconditions. 
But what if these preconditions were defined using only structure constraints? Such procedures make sense under the dynamic semantics --- 
for instance, a procedure $P$ may require a certain attribute to be added to the schema before its application. 
As it turns out, such preconditions can be added with a very low cost to the applicability problem. 

\begin{proposition}
\label{ref-rep-decidable}
Under dynamic semantics, \textsc{Applicability}  is in polynomial time for sequences of relational procedures whose preconditions contain only structure constraints.
\end{proposition}

The proof of this proposition shows that in this case one can construct a minimal schema such that the schema of 
all instances in the outcomes of sequences of procedures extend the minimal schema. 

\medskip
\noindent
\textbf{Nonemptiness}. Unfortunately, the inclusion of procedures that may alter the schema of instance has 
an immediate impact on the complexity of nonemptiness. In Section \ref{sec-fundamental-problems} we commented that 
nonemptiness was decidable for procedures made of full-tgds. Unfortunately, the following negative result 
shows that nonemptiness is undecidable even in this case, if we also allow procedures whose only 
goal is to alter the schema of instances.  

\begin{proposition}
\label{prop-rep-undec-nothingworks}
There exists a sequence $\Pro_1,\Pro_2,\Pro_3$ of procedures, with
$\Pro_1$ and $\Pro_3$ having only full tgds and $\Pro_2$ having
postconditions built using only structure constraints, such that the
following problem is undecidable: Given an instance $I$, is the set
$\dynoutcome_{\Pro_1,\Pro_2,\Pro_3}(I)$ nonempty?.
\end{proposition}

\medskip
\noindent
\textbf{Safe schema alteration}. 
We define a notion that is analogous to that of safe scope, for the case of 
procedures that alter the schema of databases. We say that a 
procedure $P = (\Scope,\Cpre,\Cpost,\Q_\safe)$ has \emph{safe schema-alterations} if both the following conditions hold: 
\begin{itemize}
\item Both $\Scope$ and $\Q_\safe$ are empty; and 
\item $\Cpost$ is a set of structure constraints. 
\end{itemize} 

Note that the procedure $P_\text{insur}$ of Example \ref{p-insur-complete} can be represented as two different procedures, 
one that alters the schema and another that migrates the needed data. The next result delivers on our promise that 
procedures with safe scope together with procedures with safe schema-alterations form a well-behaved class. 
For readability, we denote by $\P^{\textit{safe},\textit{alter}}$ the class that comprises procedures with safe scope and 
procedures with safe schema-alterations. 

\begin{proposition}
\label{prop-rep-safe-scope-dyn}
Given an instance $I$ and a sequence $\Pro_1,\dots,\Pro_n$ of procedures in $\P^{\textit{safe},\textit{alter}}$, the problem of deciding  
whether $\dynoutcome_{\Pro_1,\dots,\Pro_n}(I)$ is in polynomial time. 
\end{proposition}

%The proof of Proposition \ref{prop-rep-safe-scope} is based on the idea of chasing instances with the dependencies in 
%the procedures, and of adding 
%attributes to schemas as dictated by the procedures. As usual, to enable the chase  
%we need to introduce labeled nulls in instances (see, e.g., \cite{IL84,FKMP05}); composing procedures calls for extending the techniques 
%of \cite{APR13} to enable chasing instances that already have null values. Using the enhanced approach, 
%one can show that the result of the chase is a good over-approximation of the set of outcomes of a sequence of procedures. In Section \ref{query-reason-sec} we will see how  this over-approximation can be used to answer queries. 

%To state this result, 
%we introduce conditional tables \cite{IL84}. 

\medskip
\noindent
\textbf{Procedures with safe scope and dynamic semantics:} Before turning to more complex reasoning tasks, 
we observe that the dynamic semantics does not really interfere with reasoning tasks when we 
deal only with procedures that do not alter the schema. More precisely, 
let us say that a procedure $P$ \emph{does not force alteration of} schemas if, for any instance $I$ compatible with $P$ 
such that $\dynoutcome_P(I) \neq \emptyset$, the set $\dynoutcome_P(I)$ contains at least one instance with the same 
schema as that of $I$. We have the following observation. 

\begin{lemma}
\label{lemma-one}
Let $I$ be an instance, $P_1,\dots,P_n$ a sequence of procedures that do not force alteration of schemas, and $\alpha$ 
a first-order expression over $\schema(I)$. Then $\outcome_{P_1,\dots,P_n}(I) \models \alpha$ if and only if 
 $\dynoutcome_{P_1,\dots,P_n}(I) \models \alpha$.
\end{lemma}

It follows from Lemma \ref{lemma-one} that the algorithmic results shown in Sections \ref{analyze-sec} and 
\ref{readiness-sec} continue to hold under the dynamic semantics. Indeed, the challenge now is in  
reasoning about the outcomes when the mix involves procedures that alter schemas. 
We proceed to study problems related to query answering, leaving those related to constraints to future work. 

\subsection{Reasoning About Queries Under the Dynamic Semantics}
\label{query-reason-sec}

In this section we address the issue of decidability of reasoning about queries in presence of an instance and a 
sequence of procedures belonging to the class $\P^{\textit{safe},\textit{alter}}$ (i.e., safe-scope or safe schema-alteration). The main technique we use is the 
introduction of \emph{conditional tables} as an approximation of the set of outcomes. To state this result, 
we recall the notion of conditional tables introduced by Imieli{\'n}ski and Lipski \cite{IL84}. %some notation following that of %\cite{IL84}. 
 
Let $\Nulls$ be an infinite set of   {\em null values} that is disjoint from the set of domain values $D$. 
A {\em naive instance} $T$ over schema $\Sch$ assigns a finite relation $R^T \subseteq (D \cup \Nulls)^n$ 
to each relation symbol $R$ in $\Sch$ of arity $n$. 
Conditional instances extend naive instances by attaching conditions over the tuples. 
Formally, an \emph{element-condition} is a positive boolean combination of formulas 
of the form $x = y$ and $x \neq y$, where $x \in \Nulls$ and $y \in (D \cup \Nulls)$. 
A {\em conditional instance} $T$ over schema $\Sch$ assigns to each $n$-ary relation 
symbol $R$ in $\Sch$  a pair $(R^T,\rho^T_R)$, where $R^T \subseteq (D \cup \Nulls)^n$ and 
$\rho^T_R$ assigns an element-condition to each tuple $t \in R^T$. 
A conditional instance $T$ is \emph{positive} if none of the element-conditions in its tuples uses 
inequalities (of the form $x \neq y$). 
%A conditional instance with a \emph{global condition} is a pair $(T,\eta)$, where $T$ is a conditional instance and 
%$\eta$ is an element-condition.

%Note that a naive instance can be seen as a conditional instance 
%where every tuples is assigned the condition $x = x$, for a fresh null $x$, and thus we only give the semantics 
%for conditional tables. 

%A naive (resp. conditional) instance \emph{with scope} is a pair $\mathcal T = (T,\Rel)$, where $T$ is a naive (resp. conditional) table and $\Rel$ is a set of relation names. 

To define the semantics, let $\nulls(T)$ be the set of all the nulls in any tuple in $T$ or in an element-condition used in $T$. % , if $T$ is a conditional instance. 
Given a substitution $\nu: \nulls(T) \rightarrow D$, let $\nu^*$ be the 
extension of $\nu$ to a substitution $D \cup \nulls(T) \rightarrow D$ that is the identity on $D$. We say that 
$\nu$ {\em satisfies an element-condition} $\psi$, written $\nu \models \psi$, if for each equality $x= y$ in $\psi$ it is the case that 
$\nu^*(x) = \nu^*(y)$, and for each inequality $x \neq y$ we have that $\nu^*(x) \neq \nu^*(y)$.  
Further, we define the set $\nu(R^T)$ as $\{\nu^*(t) \mid t \in R^T$ and $\nu \models \rho^T_R(t)) 
\}$. %, where $\nu(t)$ is given by replacing each element $n \in \Nulls$ 
%in a tuple $t$ by $\nu(n)$. 
Finally, for a conditional instance $T$, $\nu(T)$ is the 
instance that assigns $\nu(R^T)$ to each relation $R$ in the schema. 

The set of instances represented by $T$, denoted by $\rep(T)$, is defined as 
$\rep(T) = \{I \mid$ there is a substitution $\nu$ such that $I$ extends $\nu(T) \}$. 
%If $G = (T,\eta)$ is a conditional instance with a global condition, then 
%$\rep(G)= \{I \mid$ there is a substitution $\nu$ such that $\nu \models \eta$ and $I$ extends $\nu(T) \}$. 
Note that the instances $I$ in this definition could have potentially bigger schemas than $\nu(T)$. In other words, 
we assume that the set $\rep(T)$ contains instances over any schema extending the schema of $T$. 

The next result states that conditional instances are good over-approximations for the outcomes 
of sequences of procedures. More interestingly, these approximations preserve the \emph{minimal instances} 
of outcomes. To put this formally, we say that an instance $J$ in a set $\Inst$ of instances is {\em minimal} if there is no instance 
$K \in \Inst, K \neq J$, and  
such that $J$ extends $K$. 

\begin{proposition}
\label{prop-minimal}
Let $I$ be an instance, and $\Pro_1,\dots,\Pro_n$ a sequence of procedures in $\P^{\textit{safe},\textit{alter}}$. 
Then either $\dynoutcome_{\Pro_1,\dots,\Pro_n} = \emptyset$, or one can construct 
a positive conditional instance $T$ 
%with global condition  
such that 
\begin{itemize}
\item $\dynoutcome_{\Pro_1,\dots,\Pro_n}(I) \subseteq \rep(T)$; and 
\item If $J$ is a minimal instance in $\rep(T)$, then $J$ is also minimal in 
$\dynoutcome_{\Pro_1,\dots,\Pro_n}(I)$. 
\end{itemize}
Further, $T$ is of double-exponential size with respect to $\Pro_1,\dots,\Pro_n$ and $I$, or 
exponential if $n$ is fixed. 
\end{proposition}

We remark that this proposition can be extended to include procedures defined only with egds, at the cost of 
a much more technical presentation and 
loosing positiveness of the constructed conditional instance.

\medskip
\noindent
\textbf{Query answering}. We can use Proposition \ref{prop-minimal} to derive a decision procedure for the problem of answering conjunctive 
queries over the outcome set of an instance $I$ and a sequence $\Pro_1,\dots,\Pro_n$ of procedures in $\P^{\textit{safe},\textit{alter}}$. The procedure  computes the conditional table that represents the minimal instances of $\dynoutcome_{\Pro_1,\dots,\Pro_n}(I)$, and then 
poses the query over this conditional table.  Of course, computing the entire table is not necessary, as we only need the 
parts that are relevant for answering the query. 

\begin{proposition}
\label{prop-dyn-query-answering} Under dynamic semantics, 
\textsc{Query answering} is decidable and in \nexptime\ if $Q$ is a conjunctive query and 
all procedures involved belong to $\P^{\textit{safe},\textit{alter}}$. It is in \np  
if the number $n$ of procedures is considered fixed. 
\end{proposition}

\medskip
\noindent
\textbf{Query readiness}. We have seen that {\textsc{query readiness}} can be undecidable even if all the procedures 
involved are safe scope. To get decidability, we had to restrict ourselves to procedures specified 
with full tgds only. We are able to obtain an analogous result if we add procedures with safe schema-alterations. The key observation 
here is that it makes sense to use schema-altering procedures only  once, thus we do not lose 
the small-sequence property of Theorem \ref{theo-readi-decidable-queries} if we include procedures with 
safe schema-alterations into the mix. 

\begin{theorem}
\label{theo-readi-decidable-queries-dynamic}
Under dynamic semantics, {\textsc{query-readiness}} is in \twonexptime\ when $\Pi$ is a set of procedures in $\P^{\textit{safe},\textit{alter}}$ 
in which all tgds involved are full. It is in \nexptime\ if the size of $\Pi$ is fixed. 
\end{theorem}

\section{Conclusions}\label{sec:conclusions} 

We have embarked on the task of developing a general framework for data improvement 
that enables one to determine whether there exists workflows that transform 
the input data into  data that satisfy some desired properties. 
We believe  that our proposal is general enough to cover a wide range of operations over multiple domains and 
models, possibly including tools from statistics, machine learning, and
    other data-oriented fields. 

In this paper we instantiated the key definitions in a relational setting. We expect that in this  
setting, our proposed framework will prove its worth by allowing different forms of reasoning about data-preparation workflows. 
For instance, we have shown how to reason about procedures given by standard relational constraints, which include most 
data-migration procedures and several kinds of data-quality tasks. In this context, we show that under broad restrictions,  
the data-readiness problem is decidable, both in the form of constraints and in the form of boolean queries. 

It seems promising to continue studying the problem of knowledge-base updates under these classes of procedures. 
In this respect, we would like to understand the limits of this problem: Can we add more expressivity to procedures and knowledge bases, while still staying in the realm of strong representation systems? For instance, there are numerous forms of description logics and/or datalog variants 
that have been shown to have good query-answering and reasoning properties (see, e.g., \cite{BO15} for description logics, and 
\cite{CGLMP10} for families of datalog); it might be the case that these properties translate  
into our setting as well. It also seems promising to pursue the dynamic semantics. In particular,   
we do not know if there is any reasonable representation 
system suitable for representing the outcomes of procedures with safe scope under the dynamic semantics. 

As a final remark,  instantiating  
our framework in other different scenarios appears to be an exciting direction for future work. The exploration could include, for instance, inclusion of procedures with 
statistical operations (as in, e.g., \cite{BCKOV14}), or information extraction from unstructured data \cite{FKRV15} or CSV  files\cite{MNV15,AMRV16}.

\bibliographystyle{plain}
\bibliography{paper}

%\end{document}
%\newpage

\newpage 
\onecolumn
\appendix

\section{Proofs and Additional Results}

\noindent
\textbf{Remark.} Since up to section 6 we consider only procedures under static schemas, for readability in most of the proofs of these sections we will use the unnamed assumption on schemas and queries. We can switch back from one to the other by using the 
order on attributes, as explained in the Preliminaries.

\subsection{Proof of proposition \ref{prop-ap-und}}

The reduction is from the complement of the embedding problem for finite semigroups, shown to be undecidable in \cite{KPT06}, and 
it is itself an adaptation of the proof of Theorem 7.2 in \cite{APR13}.  
Note that, since we do not intend to add attributes nor relations in the procedures of this proof, we can 
drop the named definition of queries, treating CQs now as normal conjunctions of relational atoms. 

The embedding problem for finite semigroups problem 
can be stated as follows. Consider a pair $\textbf A = (A,g)$, where $A$ is a finite set and 
$g: A \times A \rightarrow A$ is a partial associative function. We say that $\textbf A$ is embeddable in a finite 
semigroup is there exists $\textbf B = (B,f)$ such that $A \subseteq B$ and $f: B \times B \rightarrow B$ is a total 
associative function. The embedding problem for finite semigroups is to decide whether an arbitrary 
$\textbf A = (A,g)$  is embeddable in a finite semigroup. 

Consider the schema 
$\Sch = \{C(\cdot,\cdot), E(\cdot,\cdot), N(\cdot,\cdot), G(\cdot,\cdot,\cdot), F(\cdot), D(\cdot)\}$. 
The idea of the proof is as follows. 
We use relation $G$ to encode binary functions, so that a tuple $(a,b,c)$ in $G$ intuitively 
corresponds to saying that $g(a,b) = c$, for a function $g$. 
Using our procedure we  shall mandate that the binary function encoded in $G$ is total and associative. 
We then encode $\textbf A = (A,g)$ into our input instance $I$: the procedure will then try to embed 
$A$ into a semigroup whose function is total. 

In order to construct the procedures, we first specify the following set $\Sigma$ of tgds. 
First we add to $\Sigma$ a set of dependencies ensuring that all elements in the relation 
$G$ are collected into $D$: 
\begin{eqnarray}
G(x,u,v) & \rightarrow & D(x) \\
G(u,x,v) & \rightarrow & D(x) \\
G(u,v,x) & \rightarrow & D(x) 
\end{eqnarray}
The next set verifies that $G$ is total and associative:
\begin{eqnarray}
D(x) \wedge D(y) & \rightarrow & \exists z G(x,y,z) \\
G(x,y,u) \wedge G(u,z,v) \wedge G(y,z,w) & \rightarrow & G(x,w,v)
\end{eqnarray}
Next we include dependencies that are intended to force relation $E$ to be an equivalence relation 
over all elements in the domain of $G$. 
\begin{eqnarray}
D(x) & \rightarrow & E(x,x) \\
E(x,y) & \rightarrow & E(y,x) \\
E(x,y) \wedge E(y,z) & \rightarrow & E(x,z)
\end{eqnarray}
The next set of dependencies we add $\Sigma$ ensure that $G$ represents a function that is consistent 
with the equivalence relation $E$.
\begin{eqnarray}
G(x,y,z) \wedge E(x,x') \wedge E(y,y') \wedge E(z,z') & \rightarrow & G(x',y',z') \\
G(x,y,z) \wedge G(x',y',z') \wedge E(x,x') \wedge E(y,y') & \rightarrow & E(z,z')
\end{eqnarray}
The final tgd in $\Sigma$ serves us to collect possible errors when trying to embed 
$\textbf A = (A,g)$. The intuition for this tgd will be made clear once we outline the reduction, but 
the idea is to state that the relation $F$ now contains everything that is in $R$, as long 
as a certain property holds on relations $E$, $C$ and $N$.  
\begin{eqnarray}
E(x,y) \wedge C(u,x) \wedge C(v,y) \wedge N(u,v) \wedge R(w) & \rightarrow & F(w)
\end{eqnarray}
Let then $\Sigma$ consists of tgds (1)-(11). We construct fixed procedures $P_1 = (\Scope^1,\Cpre^1,\Cpost^1,\Q_\safe^1)$ and 
$P_2= (\Scope^2,\Cpre^2,\Cpost^2,\Q_\safe^2)$ as follows. 

\bigskip

\noindent
\underline{Procedure \textbf{$P_1$}}: 

\noindent
$\Scope^1$: The scope of $P_1$ consists of relations $G$, $E$, $D$ and $F$, which corresponds to the constraints 
$\{G[*],  E[*], D[*], F[*]\}$. 

\noindent
$\Cpre^1$: There are no preconditions for this procedure. 

\noindent
$\Cpost^1$: The postconditions are the tgds in $\Sigma$. 

\noindent
$\Q_\safe^1$: A single query ensuring that no information is deleted from all of $G$, $E$ and $F$ (and thus that no 
attributes are added to them): $G(x,y,z) \wedge E(u,v) \wedge D(w) \wedge F(p)$.

\medskip
\noindent
\underline{Procedure \textbf{$P_2$}}: 

\noindent
$\Scope^2$: The scope of $P_2$ is empty. 

\noindent
$\Cpre^2$: The precondition for this constraint is $R(x) \rightarrow F(x)$. 

\noindent
$\Cpost^2$: The are no postconditions. 

\noindent
$\Q_\safe^2$: There are no preserved queries. 

\medskip
Note that $P_2$ does not really do anything, it is only there to check that $R$ is contained in $F$. 
We can now state the reduction. On input $\textbf A = (A,g)$, where $A = \{a_1,\dots,a_n\}$, 
we construct an instance 
$I_\textbf{A}$ given by the following interpretations: 
\begin{itemize}
\item $E^{I_\textbf{A}}$ contains the pair $(a_i,a_i)$ for each $1 \leq i \leq n$ (that is, for each element of $A$); 
\item $G^{I_\textbf{A}}$ contains the triple $(a_i,a_j,a_k)$ for each $a_i,a_j,a_k \in A$ such that $g(a_i,a_j) = a_k$;  
\item $D^{I_\textbf{A}}$ and $F^{I_\textbf{A}}$ are empty, while $R^{I_\textbf{A}}$ contains a single element $d$ not in $A$; 
\item $C^{I_\textbf{A}}$ contains the pair $(i,a_i)$ for each $1 \leq i \leq n$; and 
\item $N^{I_\textbf{A}}$ contains the pair $(i,j)$ for each $i \neq j$, $1 \leq i \leq n$ and $1 \leq j \leq n$. 
\end{itemize}

Let us now show $\textbf{A} = (A,g)$ is embeddable in a finite semigroup if and only if $\outcome_{P_1}(I)$ contains 
an instance $I$ such that $I'$ does not satisfy the precondition $R(x) \rightarrow F(x)$ of procedure $P_2$.

\medskip
\noindent
($\Longrightarrow$) Assume that $\textbf{A} = (A,g)$ is embeddable in a finite semigroup, say the semigroup 
$\textbf B = (B,f)$, where $f$ is total. Let $J$ be the instance such that 
$E^J$ is the identity over $B$, $D^J = B$ and $G^J$ contains a pair $(b_1,b_2,b_3)$ if and only if 
$f(b_1,b_2) = b_3$; $F^J$ is empty and relations $N$, $C$ and $R$ are interpreted as in $I_\textbf{A}$. 
It is easy to see that $J \models \Sigma$, $Q_{\Sch \setminus \Scope}$ is preserved and that 
$\Q_\safe(I_\textbf{A}) \subseteq \Q_\safe(J)$, this last because $\textbf A$ was said to be embeddable in 
$\textbf B$. We have that $J$ then does belong to $\outcome_{P_1}(I)$, but $J$ does not satisfy the 
constraint $R(x) \rightarrow F(x)$. 

\medskip
\noindent
($\Longleftarrow$) Assume now that there is an instance $J \in \outcome_{P_1}(I)$ that does not 
satisfy $R(x) \rightarrow F(x)$. Note that, because of the scope of $P_1$, the interpretation of 
$C$, $N$ and $R$ of $J$ must be just as in $I$. Thus it must be that the element $d$ is not in 
$F^J$, because it is the only element in $R^J$. 
Construct a finite semigroup $\textbf{B} = (B,f)$ as follows. 
Let $B$ consists of one representative of each equivalence class in $E^J$, with the additional restriction that 
each $a_i$ in $A$ must be picked as its representative. Further, define $f(b_1,b_2) = b_3$ if and only if 
$G(b_1,b_2,b_3)$ is in $G$. Note that $J$ satisfies the tgds in $\Sigma$, in particular $G$ is associative and 
$E$ acts as en equivalence relation over $G$, which means that $f$ is indeed associative, total, and well defined. 
It remains to show that $\textbf A$ can be embedded in $\textbf B$, but since 
$G^J$ and $E^J$ are supersets of $G^{I_\textbf{A}}$ and $E^{I_\textbf{A}}$ (because of the preservation query of $P_1$), 
all we need to show is that each $a_i$ is in a separate equivalence relation. But this hold because of tgd 
(11) in $\Sigma$: if two elements from $A$ are in the same equivalence relation then the left hand side of (11) would hold in $I_\textbf{A}$, which contradicts the fact that $F^J$ does not contain $d$.

\subsection{Proof of proposition \ref{prop-rep-undec-weak}}

This proof is a simple adaptation of the reduction we used in the proof of Proposition \ref{prop-ap-und}. 
Indeed, consider again the schema $\Sch$ from this proof, and the procedure $P$ given by: 

\noindent
$\Scope$: The scope of $P$ consists of relations $G$, $E$, $D$ and $F$, which corresponds to the constraints 
$G[*], E[*], D[*]$ and $F[*]$. 

\noindent
$\Cpre$: There are no preconditions for this procedure. 

\noindent
$\Cpost$: The postconditions are the tgds in $\Sigma$ plus the tgd $F(x) \rightarrow R(x)$.

\noindent
$\Q_\safe$: This query ensures that no information is deleted from all of $G$, $E$ and $F$: 
$G(x,y,z) \wedge E(u,v) \wedge D(w) \wedge F(p)$.

Given a finite semigroup $\textbf A$, we construct now the following instance $I_\textbf{A}$: 
\begin{itemize}
\item $E^{I_\textbf{A}}$ contains the pair $(a_i,a_i)$ for each $1 \leq i \leq n$ (that is, for each element of $A$); 
\item $G^{I_\textbf{A}}$ contains the triple $(a_i,a_j,a_k)$ for each $a_i,a_j,a_k \in A$ such that $g(a_i,a_j) = a_k$;  
\item All of $D^{I_\textbf{A}}$, $F^{I_\textbf{A}}$ and $R^{I_\textbf{A}}$ are empty; 
\item $C^{I_\textbf{A}}$ contains the pair $(i,a_i)$ for each $1 \leq i \leq n$; and 
\item $N^{I_\textbf{A}}$ contains the pair $(i,j)$ for each $i \neq j$, $1 \leq i \leq n$ and $1 \leq j \leq n$. 
\end{itemize}

By a similar argument as the one used in the proof of Proposition \ref{prop-ap-und}, one can show that 
$\outcome_P(I_\textbf{A})$ has an instance if and only if $\textbf A$ is embeddable in a finite semigroup. 
The intuition is that now we are adding the constraint $F(x) \rightarrow R(x)$ as a postcondition, and 
since $R$ is not part of the scope of the procedure the only way to satisfy this restriction is if we do not 
fire the tgd (11) of the set $\Sigma$ constructed in the aforementioned proof. This, in turn, can only happen if 
$\textbf A$ is embeddable. 

\subsection{Proof of proposition \ref{prop-rep-undec}}

This proof is just a slight adaptation of the Proof of Proposition \ref{prop-rep-undec-nothingworks}.
The only thing needed to be done is to replace the tgd: 
\begin{eqnarray*}
D(x) \wedge D(y) & \rightarrow & G^\text{binary}(x,y) \\
\end{eqnarray*}
in that proof for the (non-full) tgd: 
\begin{eqnarray*}
D(x) \wedge D(y) & \rightarrow & \exists z G^d{binary}(x,y) 
\end{eqnarray*}

The rest follows just as in the proof of Proposition \ref{prop-rep-undec-nothingworks}.

\subsection{Proof of Proposition \ref{prop-rep-safe-scope}} 

Let $I$ and $P_1,\dots,P_n$ be as specified in the statement of the Proposition, assuming the postconditions of 
each $P_i$ correspond to a set $\Sigma_i$ of tgds. Then chasing $I$ with $\Sigma_1$, $\Sigma_2$, etc 
yields an instance in $\outcome_{P_1,\dots,P_n}(I)$. Indeed, if the chase yields instances $I_1,\dots,I_n$ with $I_0 = I$ (so that 
$I_{i+1}$ is the chase of $I_i$ with $\Sigma_{i+1}$), then each $I_{i+1}$ is a possible outcome of applying $P_{i+1}$ to $I_i$. 

\subsection{Proof of Proposition \ref{prop-outcome-fulltgds-justone}}

A consequence of the following Theorem \ref{theo-outcome-fulltgds}.

\subsection{Proof of Theorem \ref{theo-outcome-fulltgds}}

This is a consequence of Theorem \ref{theo-skb-closed}. Exponential time follows from the fact that $\removerels(\Gamma,\Scope)$ 
works in $O(2^{|\Gamma|}$ and that (since tgds are full) the maximum size of an instance produced out of chasing 
is of order $|I|^\Sch$.

\subsection{Proof of Theorem \ref{theo-skb-closed}}

For readability, in this proof we specify the Scopes of procedures using just the relation names (since they can only be formed with 
constraints of the form $R[*]$), and we sometimes omit the definition of 
$\Q_\safe$, since for procedures of safe scope this set depends only on the scope of processes. We also continue using the unnamed assumption.

We begin with some notation. The premise of a tgd is its left-hand side, and the conlcusion is its right-hand side. The dependency graph $\Gamma$ of a set of tgds is a graph whose nodes are the relation names used in the dependencies in 
$\Gamma$ and where there is an edge from a node $R$ to a node $S$ if there is a dependency in $\Gamma$ with $R$ in its premise and $S$ in its consequence. 
A set $\Gamma$ of tgds is acyclic if its dependency graph is acyclic. 

Acyclicity plays an important role in this proof, as we show that one essentially needs to keep up an acyclic set of dependencies (even though subsequent procedures may introduce 
cycles). It also allows us to define a procedure $\removerels(\Gamma, \Scope)$, for an acyclic set $\Gamma$ of tgds and a set $\Scope$ of relation names. 
The procedure $\removerels(\Gamma, \Scope)$ works as follows. Let us first assume without loss of generality that each dependency in $\Gamma$ uses different variables. 
Let also $R_1,\dots,R_n$ be an enumeration of the relations in $\Scope$ that are used in $\Gamma$ and 
that is consistent with the partial order 
of the nodes in the dependency graph of $\Gamma$. We consider 
$R_n$ is the greatest node in this order, i.e., a node without a directed path to any of the nodes $R_1,\dots,R_{n-1}$.

Then, for each $i = n,n-1,n-2,\dots,1$: 
\begin{itemize}
\item Let $\Omega_i \subseteq \Gamma$ contain all dependencies using $R_i$ in its premises. 
\item Initialize the set $\Omega_i'$ as $\Omega_i$. 
\item For each tgd $\lambda$ in $\Omega_i$ of the form $\phi(\bar x) \rightarrow \psi(\bar z)$, for each different atom $R_i(\bar y)$ in $\phi(\bar x)$ and for each dependency 
$\sigma$ in $\Gamma \setminus \Omega_i$ of the form $\theta(\bar u) \rightarrow \eta(\bar v)$ such that there is an atom $R_i(\bar w)$ in $\eta(\bar v)$\footnote{Note that because of acyclicity $\theta$ cannot contain atoms $R_j$ with $j \geq i$}
\begin{itemize}
\item Let $\Pi$ contain all possible equivalence relations for the variables in $\bar w$. For each equivalence relation $\pi \in \Pi$, let 
$\pi(\bar w)$ denote the renaming of the variables in $\bar w$ where each variable is sent to one representative of the equivalence relation. 
If there is a homomorphism 
$h: \bar y \rightarrow \pi(\bar w)$ from $R_i(\pi(\bar w))$ to $R_i(\bar y)$, then 
\item Let $\hat h$ be the extension of $h$ that is the identity on any variable in $\bar x$ but not in $\bar y$, $\hat \pi$ the extension of $\pi$ that is the identity over any variable 
in $\bar v$ not in $\bar w$, and let $\phi'(\bar x)$ be the result of removing the atom $R_i(\bar y)$ 
from $\phi(\bar x)$. 
\item Construct the dependency $\lambda' = \phi'(\hat h(\bar x)) \wedge \theta(\hat \pi(\bar u)) \rightarrow \psi(\hat h(\bar z))$. Note that this is well defined because each variable in $\bar z$ must appear in 
$\bar x$, and since $R_i$ cannot appear in $\theta$ we have removed at least one occurrence of $R_i$.  
\item Create a copy of $\lambda'$ using only fresh variables, and add it to $\Omega_i'$. 
\item Continue until no more dependencies in $\Omega_i'$ contain $R_i$ in their premises. 
\end{itemize}
\item Define $\Gamma = \Gamma \setminus \Omega_i \cup \Omega_i'$. Note that we have removed $R_i$ from the premise of the dependencies in $\Gamma$, but 
$\Gamma$ remains an acyclic set of full tgds (because on the $i$-th step we never introduce dependencies with a relation 
$R_j$ in its premises that is higher in the order than $R_i$). 
\end{itemize}

\bigskip

We now have all the ingredients to prove this theorem. 
Let $\skb = (I,\Gamma, \Scope)$ be a scoped knowledge base and such that $\Gamma$ is an acyclic set of full tgds, and 
$P = (\Scope_P,\emptyset,\Sigma,\Q_\safe)$ a procedure with safe scope with $\Sigma$ a set of full tgds. 
We show how to represent the set $\outcome_{P}(\rep(\skb))$ with a scoped knowledge base $\skb' = (I', \Gamma', \Scope')$, where $\Gamma'$ is also an acyclic set of full tgds. 

\medskip
\noindent
\textbf{Computing $\skb'$}. The scoped knowledge base is $\skb' = (I', \Gamma', \Scope')$, where $I'$ is the chase of $I$ with respect to $\Sigma$, 
$\Gamma'$ is defined as $\Sigma \cup \removerels(\Gamma, \Scope_P)$, and $\Scope' = \Scope \cup \Scope_P)$ (recall that we abuse notation and 
define $\Scope_P$ as a set of relations, where $R \in \Scope_P$ if and only if the scope includes the constraint $R[*]$).  

\medskip
\noindent
\textbf{Soundness}. Let $\skb$, $P$, $\Sigma$ and $\skb'$ be defined as above, and let $N' \in \rep(\skb')$ be an instance represented by $\skb'$. We need to 
show that $N' \in \outcome_P(\rep(\skb))$. In order to do that, we construct an instance $N \in \rep(\skb)$ and then show that $N'$ belongs to $\outcome_P(N)$. 

For readability, let us write $\chase_\Sigma(I)$ to denote the chase of $I$ with respect to $\Sigma$, so that $I' = \chase_\Sigma(I)$. 
Construct instance $K$ by removing from $N'$ all tuples in any relation in $\Scope_P$ that are not in $I$, and then define 
$N = \chase_{\Gamma}(K)$. 

We first show that $N \in \rep(\skb)$. That is, we need to show that (1) $I \subseteq K$, (2) $R^{N} = R^I$ for each $R \notin \Scope$, and 
(3) $N \models \Gamma'$. We know (3) by construction, since $N$ is the chase of $K$ with respect to $\Gamma$. For 
(1), by definition $N' \in \rep(\skb')$ is a superset of $\chase_\Sigma(I)$, and $\chase_\Sigma(I)$ is a superset of 
$I$ itself. Since we construct $K$ by removing tuples not in $I$, we have that $I \subseteq K$ and therefore $I \subseteq N$. 
For (2), clearly  $I^R = N^R$ for each relation $R$ not in $\Scope'$, and thus by construction 
$I^R = K^R$ for each relation not in $\Scope$ (since $\Scope' = \Scope \cup \Scope_P)$. 
Now, if a tuple was added to a relation $R$ in $R^N$ by chasing $K$ with respect to $\Gamma$, then 
since $\Scope'$ contain at least all relations in the consequence of any dependency in $\Gamma$, it follows that $R$ is in $\Scope$. 

Next we show that $N'$ belongs to $\outcome_P(N)$. For this we need to show that (a) $N'$ satisfies $\Sigma$, (b)  $N \subseteq N'$ and 
(c)   $N$ and $N'$ differ only in relations within $\Scope_P$. We obtain (a) from the fact that $N'$ satisfies $\Gamma'$, that contains 
all dependencies in $\Sigma$. 

\smallskip
\noindent 
For (b) we need the following claim, which we prove by induction. We say that $\chase^{i}_\Sigma(I)$ corresponds to the 
instance produced after applying the $i$-th step of a chase (for simplicity we can assume that all rules are applied in lexicographical order). 
\begin{claim}
If the $i$-th step of the chase produces tuples $\psi(h(\bar z))$ out of a dependency $\lambda = \phi(\bar x) \rightarrow \psi(\bar z)$ 
and an assignment $h$ so that $\chase^{i-1}_\Gamma(K)$ satisfies $\phi(h(\bar x))$ but not $\psi(h(\bar z))$, then there is a 
dependency $\theta(\bar u) \rightarrow \psi(\bar v)$ in $\removerels(\Gamma, \Scope_P)$ and an assignment $f$ 
such that $K$ satisfies $\theta(f(\bar u))$ and where $\psi(f(\bar v)) =  \psi(h(\bar z))$.
\end{claim}

\begin{proof}
For the base case when there is a single chase step, we have that none of the atoms in $\theta$ is in $\Scope_P$, and thus 
$\lambda$ is in $\removerels(\Gamma, \Scope_P)$ by construction. 

Now assume the claim holds for all chase steps earlier than step $k$, and let $\lambda = \phi(\bar x) \rightarrow \psi(\bar z)$  as in the claim. 
Further, assume $\phi(\bar x)$ is of the form $S_1(\bar x_1) \wedge \cdots \wedge S_m(\bar x_m)$. Then for each such atom, either $S_j$ is not in $\Scope_P$, 
or the atom $S_j(h(\bar x_j))$ it was produced earlier in the chase. Let us assume that there is only one such atom (if there are more the proof just 
follows by repeating the same argument). By induction, such atom comes from a dependency 
$\theta(\bar u) \rightarrow \eta(\bar v)$ and an assignment $f$ 
such that $K$ satisfies $\theta(f(\bar u))$ and where there is an atom $S_j(\bar w)$ in $\eta(\bar v)$ such that $f(\bar w) = h(\bar x_j)$.   

We note that function $f$ induces an equivalence relation in the variables $\bar v$, where two variables $u_1$ and $u_2$ are in the same equivalence relation 
if $f(v_1) = f(v_2)$. If $\pi$ is the assignment mandated by the equivalence relation, then clearly there is a homomorphism $g$ from $S_j(\bar x_j)$ 
to $S_j(\pi(\bar w))$. 

Let $\hat \pi$ be the extension of $\pi$ that is the identity on all variables in $\bar u$ not in $\bar w$. 
Note that since $\pi$ is the relation induced by $f$, we have that $f(\bar u) = f(\pi(\bar u))$. 
Next, let $\hat g$ be the extension of $g$ that is the identity on every variable of $\bar x$ not in $\bar x_j$. 
 
Now since $S_j$ is in $\Scope_P$ and we have shown $\pi$ and $g$ as in the condition of the procedure, 
at some point during the application of $\removerels$, the dependency $\lambda$ was  
replaced by  $\lambda' = \phi'(\hat g(\bar x)) \wedge \theta(\hat \pi(\bar u)) \rightarrow \psi(\hat g(\bar z))$, 
where $\phi'(\hat g(\bar x)) $ is the result of removing $R_i(g(\bar y)$ from $\phi(\hat g(\bar x))$. 

Define a homomorphism $h^*: g(\bar x) \rightarrow D$ so that $h^*(x) = x$ if $x \in \bar x$ and $h^*(x) = f(g(x))$ otherwise. 
Note that $h(\bar x) = h^*(g \bar x)$, since these two only differ in the variables in $\bar x_j$ and $f(\bar w) = h(\bar x_j)$.

Let $h*: \bar x \cup \bar u: \rightarrow D$ define the union of homomorphisms $h$ and $f$. 
Then we have that $K$ must satisfy 
$\phi'(h^*(\hat g(\bar x))) \wedge \theta(h^*(\pi(\bar u)))$. It satisfies $\phi'(h^*(\hat g(\bar x)))$ because $\hat g$ is the identity on each variable in $\bar x$ not in 
$\bar x_j$, and $f(\hat g(x) = h(x)$ for each $x \in \bar x_j$, and it satisfies $\theta(h^*(\hat \pi(\bar u)))$ because we know that $K$ satisfies $\theta(f(\bar u))$ and 
$f(\bar u) = f(\pi(\bar u))$. 

Again, since $f(\bar w) = h(\bar x_j)$ and $\hat g$ is the identity over any variable not in $\bar x_j$, 
we also obtain that $\psi(h^*(\hat g(\bar z))$ corresponds to $\psi(h(\bar z))$. This proofs the claim. 
\end{proof}

We continue with the proof of fact (b): $N \subseteq N'$. Let us assume otherwise, so that there is a tuple $\bar a$ and a relation $R$ so that $\bar a$ is in 
$R^N$ and not in $R^{N'}$. Clearly $\bar a$ must not be in $R^K$, since by definition $K \subseteq N'$. Thus, $\bar a$ was added to $R^N$ as a product 
of the chase. Assume without loss of generality that $\bar a$ was the first such tuple added by the chase, product of 
chasing a dependency $\phi(\bar x) \rightarrow \psi(\bar z)$ and an assignment $h: \bar x \rightarrow D$, where 
$K$ satisfies $\phi(h(\bar x))$ but not $\psi(h(\bar z))$. 

By the Claim above we know that, instead of chasing $\phi(\bar x) \rightarrow \psi(\bar z)$ we could have chased a relation 
$\theta(\bar u) \rightarrow \psi(\bar v)$ in $\removerels(\Gamma, \Scope_P)$ with the same result. 
But this last dependency is in $\Gamma'$ by construction, and therefore since $N'$ satisfies $\Gamma'$, 
it must be the case that $\bar a$ is actually in $R^{N'}$, which contradicts our initial assumption. 

\smallskip
\noindent
Finally, (c): $N$ and $N'$ differ only in relations within $\Scope_P$ follows from that fact that $K$ and $N'$ differ only in relations 
within $\Scope_P$ and $K \subseteq N \subseteq N'$. This proves $N'$ belongs to $\outcome_P(N)$.

\medskip
\noindent
\textbf{Completeness}. Let $\skb$, $P$, $\Sigma$ and $\skb'$ be defined as above, and consider an instance $N'$ in $\outcome_P(\rep(\skb))$. 
Then there is an instance $N \in \rep(\skb)$ such that $N' \in \outcome_P(N)$. We show that $N'$ belongs also to $\rep(\skb')$. 
We need to prove that (1) $\chase_\sigma(I) \subseteq N'$, (2) $R^{N'} = R^{\chase_\Sigma(I)}$ for each $R \notin \Scope'$, and 
(3) $N' \models \Gamma'$. 

\smallskip 
\noindent
For (1), note that  $N \subseteq N'$ because of the preservation queries in $P$. Since $N'$ satisfies, in particular, $\Sigma$, and $I \subseteq N$, it must be that 
$\chase_\sigma(I) \subseteq N'$

\smallskip 
\noindent 
For (2), we observe that $N$ and $N'$ differ only in relations not in $\Scope_P$, and, furthermore, $I$ and $N$ differ only in relations from $\Scope$. 
Since $\Scope' = \Scope \cup \Scope_P$, if one could find a tuple in a relation $R \notin \Scope'$ such that $N'$ and $\chase_\Sigma(I)$ differ on $R$, 
then either $N'$ and $N$ differ on $R$, or $N$ and $I$ differ on $R$, which we know it is not possible. 

\smallskip 
\noindent 
For (3), note that $N$ satisfies $\Gamma$, and therefore it satisfies $\removerels(\Gamma,\Sigma)$. Now $N'$ is a superset of $N$ that satisfies $\Sigma$. If 
$N'$ does not satisfy a dependency $\lambda$ in $\Gamma'$, then $\lambda$ must belong to $\removerels(\Gamma,\Sigma)$ and thus $N$ satisfies $\lambda$. 
All relations in the premise of $\lambda$ do not belong to $\Scope_P$, and thus since $N$ and $N'$ only differ in relations from $\Scope_P$, 
any assignment $h$ from the premise of $\lambda$ to $N'$ is also an assignment 
to $N$. Furthermore, since $N \subseteq N'$, if $h$ is not an assignment from the consequence of $\lambda$ to $N'$ it is also not an assignment 
from the consequence of $\lambda$ to $N$. This contradicts the fact that $N$ satisfies $\lambda$. 

\subsection{Proof of proposition \ref{prop-size-representation}}

Recall we omit preservation queries from procedures with safe scope. 
Fix a number $i \geq 1$. 
Let $\Sch$ contains unary relation $R$ and $n$-ary relation $T$, and let $I$ such that $R^I = \{0,1\}$ 
and $T^I = \emptyset$. Let $P'$ be the procedure with scope $R$, no preconditions and postcondition 
$\exists x_2 \exists x_3 \exists x_n T(x_1,\dots,x_n) \rightarrow R(x_1)$. 
Let $P_i$ be a procedure with 
with scope $T$, no preconditions and postcondition $R(x_1) \wedge R(x_2) \wedge \cdots \wedge R(x_i) \rightarrow T(x_1,\dots,x_i)$. 

Now note that all instances $K$ in $\outcome_{P_i,P'}(I)$ must be such that $K^T$ contains all tuples of length $i$ that can be formed with 
$0$s and $1$s. Furthermore, since the scope of $P'$ is $R$, for every SKB $(J,\Gamma,\Scope)$ representing 
$\outcome_{P_i,P'}(I)$ it must be that $\Gamma$ does not contain tgds with $R$ in its premises (this follows from the construction of 
Theorem \ref{theo-skb-closed}). Then clearly $J$ must contain all $2^i$ tuples. 

\subsection{Proof of Proposition \ref{egd-prop}} 

Let $\skb = (I,\Gamma,\Scope)$ be as in the statement of the Proposition. 
Let also $\sigma$ be of the form 
%
%\noindent 
$\sigma: \ \phi({\bar X}) \rightarrow Y = Z,$
%
%\noindent 
where each of $Y$ and $Z$ is a term in the vector $\bar X$, and let $\phi^O(\bar X^O)$ be the restriction of $\phi$ to only the relations 
in $\Scope$, and $\phi^C(\bar X^C)$ the restriction of $\phi$ to relations not in $\Scope$. 
Consider now an enumeration $h_1,\dots,h_\ell$ of all mappings from $\phi^C(\bar X)$ to $I$. Next, for each such mapping 
$h$, construct an assignment $g$ that extends $h$ by assigning all variables not in $\bar X^C$ to a fresh element. 
Furthermore, for each such assignment $g$, construct a set of assignments $g_1,\dots,g_n$, where $g_1 = g$ and 
where each assignment $g_j$, with $1$ $<$ $j$ $\leq$ $n$, is a modification of $g_1$ that allows some elements in the image of $\bar X$ to be the same fresh element or to be an element used  in $I$ in $SKB$. We choose $n$ so that all possible such restrictions (up to renaming the null symbols) are included in the set  $\{$ $g_1$, $\ldots$, $g_n$ $\}$.

Let us now assume $\{$ $\tau_1$, $\ldots$, $\tau_k$ $\}$ is the set of all assignments constructed in this way. We construct an 
instance $D_i$ for each such assignment: Each $D_i$ is created by adding to $I$ the set of tuples $\phi(\tau_i({\bar X}))$. We then chase each $D_i$ with $\Gamma$, 
% \ryc{--- and with ${\cal Q}_{safe}$?}, 
and check whether the resulting instance $(D_i)^{\Gamma}$ satisfies the egd $\sigma$. 
%(Note that here we do not need to chase with the ``views for the relations outside the scope in $SKB$,'' because we assume that each relation outside the scope in $SKB$ is a ground relation.) 
%, and because (ii) (do we even need this assumption?) we assume that in $\Gamma$, each relation symbol in the RHS of each tgd is in scope in $SKB$.} 
There are two cases here:

\begin{enumerate} 
	\item There exists an $i$ $\in$ $[1, k]$ such that $(D_i)^{\Gamma}$ $\models$ $\sigma$ does not hold. But it is easy to see 
	that $D_i$ belongs to the set $\rep(\skb)$, so we have a counterexample and we conclude that $(I,\Gamma, \Scope)$ does not satisfy $\sigma$. 

	\item (This is the ``small-witness counterexample'' property at the center of this proof.) Suppose it is the case that $(D_i)^{\Gamma}$ $\models$ $\sigma$ for each $i$ $\in$ $[1, k]$. In this case, we conclude that $(I,\Gamma, \Scope)$ satisfies $\sigma$. 
	
	Indeed, assume toward a contradiction that there exists an instance $D_*$ in $(I,\Gamma, \Scope)$ such  that $D_*$ does not satisfy $\sigma$. By the fact that $(D_i)^{\Gamma}$ $\models$ $\sigma$ for each $i$ $\in$ $[1, k]$. Then there cannot exist a homomorphism to $D_*$ from $(D_i)^{\Gamma}$ for any $i$ $\in$ $[1$, $k]$, such that the the image of the homomorphism would include one or more violations of $\sigma$ in $D_*$. The reason for the nonexistence of any such homomorphism is that each $(D_i)^{\Gamma}$ is ``minimal'' in terms of satisfying both $I$ and $\sigma$, and $D_*$ does not satisfy $\sigma$ while $(D_i)^{\Gamma}$ satisfies $\sigma$ for each $i$ $\in$ $[1, k]$. Note that all the instances $(D_i)^{\Gamma}$ taken together cover all the ``minimal'' instances that satisfy all of $I$, $\Gamma$, and $\sigma$. Thus, to obtain $D_*$, we need to add at least one tuple to some $(D_i)^{\Gamma}$, and then chase the result (call this result $D_{**}$) with $\Gamma$. 
	
	By construction of $D_{**}$ and $D_*$, there must be a homomorphism, $h$, from  $D_{**}$ to the part of $D_*$ that has all the tuples in $\nu(\phi)$ (recall that $\phi$ is the body of $\sigma$) with $\nu(Y)$ $\neq$ $\nu(Z)$, for some valuation $\nu$ from $\phi$ to $D_*$. Thus, it must be possible to extend $h$ to a homomorphism, $h'$, from  $D_{**}$ to $D_*$, such that the preimage of $h'$ includes all of $I$. Take one $h'$ such that its preimage $D'$ is exactly $I$ plus the minimal set (potentially empty)  of extra tuples that are needed to cover all of the preimage of $h$ by the preimage of $h'$. Then, by construction of the instances $(D_i)^{\Gamma}$ (by that construction, there must be a homomorphism from the $I$ in $SKB$ to $(D_i)^{\Gamma}$ for each $i$ $\in$ $[1$, $k]$), $D'$ must be isomorphic to $(D_i)^{\Gamma}$ for some $i$ $\in$ $[1$, $k]$. It follows that there exists a homomorphism from that $(D_i)^{\Gamma}$ to a part of $D_*$ that contains at least one violation of $\sigma$. We get the desired contradiction, as the existence of such a homomorphism is impossible by construction of the $(D_i)^{\Gamma}$ and by our assumption that there exists a ground instance $D_*$ of $SKB$ that violates $\sigma$.
	\end{enumerate} 
	%
%	\mbox{} 
	%
	%\ryc{These are pieces of the old proof:}
	%
%	Thus, we can describe $D_*$ as having been generated by extending one of the $(D_i)^{\Gamma}$'s by at least one more tuple, and by then chasing the resulting instance with $\Gamma$. Observe that this process must result in  an instance ($D_*$) that is bound to satisfy $\sigma$. 
%
%	Indeed, the chase that would generate $D_*$ would involve at least all the tuples (and thus the chase would ``go at least as far as'') in the appropriate $D_i^{\Gamma}$. Consider any one violation of $\sigma$ in $D_*$; this violation must be of the form $\nu(\phi({\bar X}))$, with $\nu(X)$ and $\nu(Y)$ being two distinct terms. \ryc{Stopped here F02/17/17}
	%
	 %there must exist at least one $j$ $\in$ $[1, k]$ such that there is a homomorphism, $h$, from $(D_j)^{\Gamma}$ into 
		%
	 %violation of $\sigma$ and from. ((from $(D_i)^{\Gamma}$ $\models$ $\sigma$ for each $i$ $\in$ $[1, k]$) ) We thus obtain the desired contradiction. 
%\end{enumerate} 
%
%\ryc{Need to finish writing this prooffrom handwritten}
%

\medskip
\noindent
\textbf{Membership in \conp}. From the remarks above, to solve the complement of the satisfaction problem it suffices to guess one of the $D_i$s (which are of 
polynomial size), a mapping $h$ from $\phi({\bar X})$ to the chase of $D_i$ with $\Gamma$ such that $h(Z) \neq h(y)$, 
and the appropriate chase rules 
(plus their assignments) to take us from $D_i$ to the images of $h$ (thus, instead of chasing the entire $D_i$ we just guess 
the witnesses to cover the image of $h$). 

\medskip
\noindent
\textbf{Proof of \conp-hardness}. We reduce from the compliment of the 3-colorability problem. Given a graph $G$, let 
$Q_G$ be the corresponding boolean CQ whose underlying graph is $G$ (using fresh variables), using a binary relation $E$ 
to specify its edges. 
Further, let $R$ and $S$ unary relations, and let $I$ be an instance such  that $I^R = \{1\}$, $I^S = \{2\}$, 
$I^E = \{(r,b), (b,r), (b,w), (w,b), (w,r), (r,w)\}$, and define the skb $\skb$ as $(I,\emptyset,\{E\})$. 
Consider then the egd $\sigma = Q_G \wedge R(x) \wedge S(y) \rightarrow x = y$, where $x$ and $y$ are again fresh variables not used in $Q_G$. 
It follows that $\skb \models \sigma$ if and only if $G$ is not 3-colorable.

\subsection{Proof of proposition \ref{tgd-entailment-decidability-prop}}

\begin{example} 
\label{tgd-entailment-decidability-example} 
We illustrate via an example the subtleties in the interplay between the closed-world and open-world aspects of our setting. (Recall that we explore a setting that is partially closed world, as we use a particularly simple case of exact materialized views to model relations that are not in the specified scope of a scoped knowledge base. At the same time, our setting is also partially open world, as expressed by the relations that are in scope in the given scoped knowledge base.) 

Consider a scoped knowledge base $SKB$ $=$ $(I,\Gamma, \Scope)$ over schema  $\Sch$ $=$ $\{ R(A), T(B) \}$, with $\Gamma$ $=$ $\emptyset$, the relation $T$ being the only relation in $\Scope$, and $I$ consisting of four tuples $R(a),$ $R(b),$ $T(a),$ and $T(b)$. 
 
(i) Let us check first whether tgd $\sigma_1: \ R(X) \rightarrow T(X)$ holds on $SKB$. Intuitively, $\sigma_1$ should hold on $SKB$. Indeed, the relation $R$ in $I$ is exact (closed world), and thus stays the same (i.e., always is exactly $\{$ $R(a),$ $R(b)$ $\}$) in all the ground instances in $SKB$. At the same time, relation $T$ in $I$ is open world, and thus each ground instance in $SKB$ has an instance of $T$ that is a superset of $\{$ $T(a),$ $T(b)$ $\}$. 

To verify formally that the tgd $\sigma_1$ holds on $SKB$, we transform each of the body and head of $\sigma_1$ into two respective CQ queries, $Q_1$ and $Q_2$:  

\noindent 
$Q_1(X) \leftarrow R(X); \ Q_2(X) \leftarrow T(X).$ 

By Proposition \ref{tgd-entailment-decidability-prop-app}, the tgd $\sigma_1$ holds on $SKB$ iff the result of transforming $Q_1$ using the information we have in $SKB$ is contained in $Q_2$.

We begin the transformation of $Q_1$ by conjoining its body with the conjunction of all the facts in $I$, and denote the result by $Q'_1$: 

\noindent 
$Q'_1(X) \leftarrow R(X), R(a), R(b), T(a), T(b).$ 

We now chase $Q'_1$ with a dependency, $\tau$, generated from the materialized view $V(X) \leftarrow R(X)$. (The view $V$ expresses the fact that the relation $R$ in $I$ is exact --- closed world --- and thus is exactly $\{$ $R(a),$ $R(b)$ $\}$ in all the ground instances in $SKB$.) The dependency $\tau$ is constructed from the definition of the view $V$ and from its answer, which is the relation $\{$ $R(a),$ $R(b)$ $\}$ given in $SKB$: 

\noindent 
$\tau: \ R(X) \rightarrow X = a \vee X = b.$ 

The chase transforms the query $Q'_1$ into a UCQ query $Q''_1$: 

\noindent 
$Q^{''(1)}_1(a) \leftarrow R(a), R(a), R(b), T(a), T(b).$ 

\noindent 
$Q^{''(2)}_1(b) \leftarrow R(b), R(a), R(b), T(a), T(b).$ 

(After the removal of duplicate tuples from each CQ component of $Q''_1$, the bodies of the components become identical to each other, but the heads are still distinct.) 

By the containment mapping $\mu_1:$ $\{$ $X$ $\rightarrow$ $a$ $\}$, we have that $Q^{''(1)}_1$ is contained in $Q_2$. Further, by  the containment mapping $\mu_2:$ $\{$ $X$ $\rightarrow$ $b$ $\}$, we have that $Q^{''(2)}_1$ is contained in $Q_2$. We conclude that $Q''_1$ is contained in $Q_2$. Thus, by Proposition \ref{tgd-entailment-decidability-prop-app} we have that the tgd $\sigma_1$ holds on $SKB$. 

(ii) Let us now check whether tgd $\sigma_2: \ T(X) \rightarrow R(X)$ holds on $SKB$. Intuitively, $\sigma_2$ should not hold on $SKB$. (Recall that  the relation $R$ in $I$ is closed world and thus is exactly $\{$ $R(a),$ $R(b)$ $\}$ in all the ground instances in $SKB$. At the same time, relation $T$ in $I$ is open world, and thus each ground instance in $SKB$ has an instance of $T$ that is a superset of $\{$ $T(a),$ $T(b)$ $\}$, in particular a proper supersets of this set in an infinite number of ground instances of $I$ in $SKB$.) 

By the process that is symmetric to that in part (i) of this example, we obtain from $\sigma_2$ two CQ queries, $P_1$ and $P_2$:  

\noindent 
$P_1(X) \leftarrow T(X); \ P_2(X) \leftarrow R(X).$ 

By Proposition \ref{tgd-entailment-decidability-prop-app}, the tgd $\sigma_2$ holds on $SKB$ iff the result of transforming $P_1$ using the information we have in $SKB$ is contained in $P_2$. 

We begin the transformation of $P_1$ by conjoining its body with the conjunction of all the facts in $I$, and denote the result by $P'_1$: 

\noindent 
$P'_1(X) \leftarrow T(X), R(a), R(b), T(a), T(b).$ 

We now chase $P'_1$ with the dependency, $\tau$, featured in part (i) of this example: 

\noindent 
$\tau: \ R(X) \rightarrow X = a \vee X = b.$ 

The chase keeps the query $P'_1$ intact, as $\tau$ is not applicable. We denote by $P''_1$ the result of thus terminated chase. 

\noindent 
$P''_1(X) \leftarrow T(X), R(a), R(b), T(a), T(b).$ 

There does not exist a containment mapping from $P_2$ to $P''_1$. (Any containment mapping from from $P_2$ to $P''_1$ would have to map $X$ in $P_2$ to the head variable of $P''_1$, which would force an invalid mapping from the only subgoal $R(X)$ of $P_2$ to the non-matching subgoal $T(X)$ of $P''_1$.) We prove below that in this case the tgd $\sigma_2$ does not hold on $SKB$.
\end{example} 

Let us begin with the general case ($\Pi^p_2$ membership).
\begin{proof}{(Proposition \ref{tgd-entailment-decidability-prop} - first part)}  
Let $\sigma$ be of the form $\sigma: \ \phi({\bar X}, {\bar Y}) \rightarrow \ \exists {\bar Z} \ \psi({\bar X}, {\bar Z})$. Consider two CQ queries 
$Q({\bar X}) \leftarrow \phi({\bar X}, {\bar Y})$ 
and 
$P({\bar X}) \leftarrow \psi({\bar X}, {\bar Z}),$ each constructed from the respective part of the tgd $\sigma$. The proof constructs from $Q({\bar X})$ a $UCQ^{\neq}$ query $Q''({\bar X})$, such that $Q''$ is equivalent to $Q$ on all ground instances in the scoped knowledge base $SKB$. We then show that $Q''$ is contained in $P$ if and only if the tgd $\sigma$ holds on each instance in $SKB$. As the containment test for $UCQ^{\neq}$ queries in CQ queries is decidable in $\Pi^p_2$  (see e.g. \cite{SY80}) the result of the first part follows.

We begin constructing the query $Q''({\bar X})$ by conjoining $\phi({\bar X}, {\bar Y})$ with the instance $I$ in $SKB$, treating $I$ as a conjunction of atoms. %, after renaming each named null in $I$ into a fresh variable not already occurring in ${\bar X}$ or ${\bar Y}$. 
We then chase the resulting CQ query $Q'({\bar X})$ with $\Gamma$, after  $\Gamma$ has been (i) transformed into a set of dependencies $\Gamma^{\neq}$, and (ii) enhanced with a set ${\cal V}^{\neq}$ of dependenies that were introduced in \cite{ChirkovaY14}  as a straightforward generalization of disjunctive egds %as introduced in 
\cite{DeutschT01,FaginKMP05}. We denote by $\Psi$  the result of transforming $\Gamma$ via (i) and (ii), i.e., $\Psi$ $:=$ $\Gamma^{\neq}$ $\cup$ ${\cal V}^{\neq}$; as discussed above, we obtain the query $Q''({\bar X})$ by chasing $Q'({\bar X})$ with $\Psi$. We now provide the details of (i) and (ii) in the construction of $\Psi$, using the exposition in \cite{ChirkovaY14}. 

{\bf Constructing $\Psi$:} 
The construction of $\Psi$ uses {\em normalized} versions of conjunctions of relational atoms (see, e.g., \cite{ZhangO97}). That is, let $\phi$ be a conjunction of relational atoms. We replace in $\phi$ each duplicate occurrence of a variable or constant with a fresh distinct variable name. As we do each replacement, say of $X$ (or $c$) with $Y$, we add to the conjunction the equality atom $Y$ $=$ $X$ (or $Y$ $=$ $c$). As an illustration, if $\phi$ $=$ $P(X,X) \wedge S(c,c,X)$, then its normalized version is ${\phi}^{(norm)}$ $=$ $P(X,Y) \wedge S(c,Z,W) \wedge Y = X \wedge Z = c \wedge W = X$. %Clearly, 
By construction, the normalized version of each $\phi$ is unique up to variable renamings. For the normalized version ${\phi}^{(norm)}$ of a conjunction $\phi$, we will denote by ${\cal R}({\phi}^{(norm)})$ the conjunction of all the relational atoms in ${\phi}^{(norm)}$, and will denote by ${\cal E}({\phi}^{(norm)})$ the conjunction of all the equality atoms in ${\phi}^{(norm)}$. (If ${\phi}^{(norm)}$ has no equality atoms, we set ${\cal E}({\phi}^{(norm)})$ to $true$.) %Thus, ${\cal C}^{(norm)}$ $=$ ${\cal R}({\cal C}^{(norm)})$ $\wedge$ ${\cal E}({\cal C}^{(norm)})$. 

A {\em non-egd} {\em (negd)} % for short)  
is a dependency of the form 
%
%\vspace{-0.2cm} 
%
\begin{equation} 
\label{negd-eqn}  
%\vspace{-0.1cm} 
\sigma: \phi(\bar{W}) \rightarrow X \neq Y.  
\end{equation} 

%\vspace{-0.1cm} 

\noindent 
Here, $\phi$ is a conjunction of relational atoms, and each of $X$ and $Y$ is an element of the set of variables $\bar W$. 

We also use %in Section~\ref{containment-procedure-sec} 
chase with ``implication constraints,'' %or Horn rules with  empty heads, 
see, e.g., \cite{ZhangO97}. An {\em implication constraint (ic)} is a dependency of the form $\tau: \phi(\bar{W}) \rightarrow false$, with $\phi(\bar{W})$ a conjunction of relational atoms. % (and $false$ a truth value). 

Intuitively, to obtain the query $Q''({\bar X})$, we will be performing chase of $CQ^{\neq}$ queries (starting with $Q'({\bar X})$) with the set of dependencies $\Psi$, which includes potentially ics, negds, egds, and tgds. (The chase rules are as defined in \cite{ChirkovaY14}.) 

(i) Dependencies $\Gamma^{\neq}$: We convert each dependency in $\Gamma$ using a conversion rule that follows, and then produce $\Gamma^{\neq}$ as the union of the outputs. The conversion rule for a dependency $\gamma$ $\in$ $\Gamma$ of the form $\gamma: \phi({\bar X},{\bar Y}) \rightarrow \exists \bar{Z} \ \psi(\bar{X},\bar{Z})$ converts $\phi$ into ${\cal R}(\phi^{(norm)})$ $\wedge$ ${\cal E}(\phi^{(norm)})$, and then returns 

\begin{tabbing} 
$\gamma^{(\neq)}: {\cal R}(\phi^{(norm)}) \rightarrow  \exists \bar{Z} \ \psi(\bar{X},\bar{Z}) \vee \neg {\cal E}(\phi^{(norm)})$.  
\end{tabbing} 

(ii) Dependencies ${\cal V}^{\neq}$: We use a type of dependencies, as introduced in \cite{ChirkovaY14}, that collectively enable us to reflect the requirement that in all the ground instances of $SKB$, the contents of all the outside-scope relations are fixed. Let $R_1,$ $\ldots$, $R_m$ ($m$ $\geq$ $0$) be the names of all such outside-scope relations in $SKB$. We denote by $\cal V$ the set of $m$ views $V_1$, $\ldots$, $V_m$, such that for each $i$ $\in$ $[1,$ $m]$ and for the relation $R_i$, $V_i$ is defined as $V_i({\bar X})$ $\leftarrow$ $R_i({\bar X})$. We now proceed for each $V_i$ as follows: 
\begin{itemize} 
	\item 
If, for the instance $I$ in $SKB$, we have $R_i(I)$ $=$ $\emptyset$, we define the implication constraint  $\iota_{V_i}$ {\em for} $V_i$ as 
\begin{equation} 
\label{new-iota-eqn}  
\iota_{V_i}: \ R_i({\bar X}) \rightarrow \ false.  
\end{equation}

% let $k_V$ $=$ $0$, and let $MV[V]$ be the set $\{ () \}$. Then we define the {\em $MV$-induced generalized egd} $\tau_V$ {\em for} $V$ as  

%\begin{tabbing} 
%$\tau_V:$ $\phi({\bar X},{\bar Y})$ $\rightarrow$ $true$.  
%\end{tabbing} 

%Finally, 

\item 

Now suppose that the arity $p_i$ of the relation $R_i$ is greater than zero, and that for the instance $I$ in $SKB$, we have  $R_i(I)$ $=$ $\{ {\bar t}_1$, ${\bar t}_2$, $\ldots$, ${\bar t}_{m_i} \}$, with ${m_i}$ $\geq$ $1$. Then we define the generalized negd \cite{ChirkovaY14} $\tau_{V_i}$ {\em for} $V_i$ as  
%
%\begin{tabbing} 
%$\tau_V:$ $\phi({\bar X},{\bar Y})$ $\rightarrow$ $\bigvee_{i=1}^{m_V} ({\bar X} = {\bar t}_i)$.  
%\end{tabbing} 
%\vspace{-0.1cm} 
\begin{equation} 
\label{new-tau-eqn}  
\tau_{V_i}: R({\bar X}) \rightarrow \vee_{j=1}^{m_i} ({\bar X} = {\bar t}_j) .  
\end{equation} 

%\vspace{-0.1cm} 

\noindent 
Here, ${\bar X}$ $=$ $[S_1,\ldots,S_{p_i}]$ is the (distinct-variable-only) head vector of the query for $V_i$.  For each $j$ $\in$ $[1,$ ${m_i}]$ and for the ground tuple $\bar t_j$ $=$ $(c_{j1},$ $\ldots,$ $c_{jp_i})$ $\in$   $R_i(I)$, we abbreviate by ${\bar X} = {\bar t}_j$ the conjunction $\wedge_{l=1}^{p_i} (S_l = c_{jl})$. 

\end{itemize}

The set of dependencies ${\cal V}^{\neq}$ is the union of the implication constraints and generalized negds, one for each relation ouside the scope in $SKB$, with the exception that nonempty Boolean relations (if any are present ouside the scope in $SKB$) are not represented in ${\cal V}^{\neq}$. (It is shown in \cite{ChirkovaY14} that it is not necessary for the correctness of the chase to include such dependencies for nonempty Boolean relations.) 

%\vspace{-0.1cm} 

%\noindent 
%(Observe that whenever $\sigma$ is an egd, $\sigma_{(\neq)}$ is a gnegd.) % See Example~\ref{disj-neq-ex} for an illustration.) 

%\vspace{-0.1cm} 

It is shown in \cite{ChirkovaY14} that chase of CQ queries with sets of dependencies such as in $\Psi$ terminates and has a unique output (up to variable renaming) that is a $UCQ^{\neq}$ query, provided $\Gamma$ is a weakly acyclic \cite{FaginKMP05} set of egds and tgds. Denote by $Q''({\bar X})$ the output of the chase of the CQ query $Q'({\bar X})$ (as constructed in the beginning of this proof) with the set of dependencies $\Psi$ $=$ $\Gamma^{\neq}$ $\cup$ ${\cal V}^{\neq}$. By construction of  $Q''({\bar X})$ (see \cite{ChirkovaY14}), we have that for each ground instance $D$ in $SKB$, the queries $Q({\bar X})$ and $Q''({\bar X})$ have on $D$ the same answer set, call it ${\cal A}_Q(D)$. 

Now: 
\begin{itemize} 
	\item Suppose that we have $Q''({\bar X})$ $\sqsubseteq$ $P({\bar X})$. Then it follows immediately from the definition of tgds and from the reasoning in the previous paragraph that $\sigma$ holds on all ground instances in $SKB$. 
	
	\item Conversely, suppose that $\sigma$ holds on  all ground instances $D$ in $SKB$. Then, by definition of $\sigma$, for each such $D$ we have that 
	
	$Q(D)$ $\subseteq$ $P(D)$. 
	
	But we have obtained that the answer ${\cal A}_Q(D)$ to the query $Q$ on the instance $D$ is also the anwer to $Q''({\bar X})$ on $D$. Thus, we have that $Q''({\bar X})$ $\sqsubseteq$ $P({\bar X})$. 	The latter result is immediate from the following  properties of the relationship between $Q$ and $Q''$, by construction of $Q''$ from $Q$: 
	\begin{itemize} 
		\item On all the instances $D$ that are represented by the given SKB, $Q(D)$ $=$ $Q''(D)$; 
		\item On all the instances $D$ that do not have at least one fact included in $I$ in the given SKB, $Q''(D)$ $=$ $\emptyset$; and 
		\item On each instance $D$ that includes all the facts that are present in $I$ in the given SKB, each valuation from $Q''$ to $D$ is an instance represented by the given SKB.  
	\end{itemize}

\end{itemize} 
\end{proof}

Next we show ($\np$-completeness when the SKB is safe).
\begin{proof}{(Proposition \ref{tgd-entailment-decidability-prop} - second part)}  
The key ingredient is that, when the SKB is safe, everything produced out of the case belongs to a relation on scope. 
Thus, we can then proceed with the standard proof for implication of tgds in an instance, where the premise of the tgd is 
added to $I$, chased according to $\Gamma$ so that the conclusion of the tgd can be found in the chase. 
Since $\Gamma$ is weakly acyclic, chase terminates and query answering is $\np$-complete (combied complexity) \cite{FKMP05}.
\end{proof}

\subsection{Proof of Proposition \ref{prop-query-answering}}

The membership in \exptimenp\ follows from the proof of Proposition \ref{prop-dyn-query-answering}, as the setting therein strictly generalizes the one in the 
statement of the Proposition. 

For the membership in \exptime, we can use the following algorithm. 
\begin{itemize}
\item Compute the skb $\skb = (J,\Gamma,\Scope)$ representing the set $\outcome_{P_1,\dots,P_n}(I)$. 
\item We know that $J$ may be of exponential size, but the amount of data values in $J$ are the same as in $I$. Let us denote this number 
by $d$.  
\item Since CQs are preserved under homomorphisms, it suffices 
to check the satisfaction of $Q$ over $J$. In order to do that, we can enumerate the number of homomorphisms $h$ from $Q$ to 
$J$, which are bounded by $d^{|Q|}$, and see whether $h(Q)$ is realised in $J$. 
\end{itemize}

For the membership in \ptime, we use the same algorithm as for the \exptime case, albeit this time all of $Q$ and $P_1,\dots,P_n$ are fixed. 
This means that $J$ is of size plynomial in $I$, and the number of homomorphisms if also polynomial. This results altogether in a polynomial algorithm. 

\subsection{Proof of Proposition \ref{prop-readi-undecidable-constraints}}

\newcommand{\alphabet}{\textbf{A}}
\newcommand{\Moveright}{\textup{R}}
\newcommand{\Moveleft}{\textup{L}}
\newcommand{\blankcell}{\mathord{\hbox{\textvisiblespace}}}

%Proposition \ref{prop-readi-undecidable-constraints} stated that the
%problem \textsc{constraint readiness} is undecidable, even if $\Pi$ is
%a set of procedures with safe scope and $\Sigma$ contains only tgds.
%In outline, the proof goes as follows.

%\begin{proof}{(Proposition \ref{prop-readi-undecidable-constraints})}

  The proof reduces the halting problem to {\textsc{constraint
      readiness}} for tgd constraints, encoding the computation of
  Turing machine on an input word by using a set of relations that
  describe successive machine configurations and tgds that describe
  machine transitions.  The proof outline here is similar in spirit to
  the proof outline given in \cite[Theorem 1]{DNR08b} to show that the
  problem of deciding whether there exists a terminating chase
  sequence for a set of tgds is undecidable.

  Let $M = (Q,\alphabet,q_0,q_h,\delta)$ be a deterministic Turing
  machine with a single biinfinite tape, an alphabet $\alphabet$
  including a blank symbol that we write as $\blankcell$, and a set of
  controller states $Q$ containing an initial state $q_0$.  We assume
  without loss of generality that $M$ has a single halting state $q_h$,
  distinct from $q_0$, and that no transitions are defined for this
  state.  Let $w = a_1,\dots,a_m$ denote the input word in
  $\alphabet^*$.

  We represent the space-time structure of a computation of $M$ by a
  grid of tape cells, in which the top row of cells represents the
  initial configuration of $M$, and each successive row beneath the
  first represents the next machine configuration.  Each tape cell in
  a row is connected by horizontal edges to the cells representing its
  left and right neighbors.  With some exceptions, every cell in a row
  is also connected by vertical edges to its correspondents (if any)
  in the preceding and succeeding rows.

\medskip

\textbf{Schema.}  We represent this grid structure in relations, using
a schema $\Sch$ that includes ternary relations $T$ (``tape'') and $H$
(``head'').  The relation $T$ describes the tape contents and layout
in a machine configuration: $T(x,a,y)$ indicates that cell $x$
contains letter $a$ and lies immediately to the left of cell $y$.  The
relation $H$ describes the head location and state: $H(x,q,y)$
indicates that the tape head rests on cell $x$ (immediately to the
left of cell $y$) and that the machine controller is in state $q$.  We
can depict the initial configuration of $M$ as follows, with $c_i$
denoting the tape cell numbered $i$, and with $B$ and $E$ denoting special
elements, different from any tape symbols, that represent,
respectively, the beginning and the end of the used or visited portion
of the infinite tape.
  
\begin{center}
  \begin{tikzpicture}[node distance=15mm]
    \node (c0) {$c_0$};
    \node (c1) [right of=c0] {$c_1$};
    \node (c2) [right of=c1] {$c_2$};
    \node (c3) [right of=c2] {\dots};
    \node (cm) [right of=c3] {$c_m$};
    \node (cm1) [right of=cm] {$c_{m+1}$};
    \node (cm2) [right of=cm1] {$c_{m+2}$};

    \draw (c0) edge node [above] {$B$} (c1);
    \draw (c1) edge node [above] {$a_1$}  node [below] {$q_0$} (c2);
    \draw (c2) edge node [above] {$a_2$} (c3);
    \draw (c3) edge node [above] {$a_{m-1}$} (cm);
    \draw (cm) edge node [above] {$a_m$} (cm1);
    \draw (cm1) edge node [above] {$E$} (cm2);
  \end{tikzpicture}
\end{center}

The schema $\Sch$ also includes two auxilliary binary relations $L$
(``left'') and $R$ (``right'') used to ensure that the tape cells to
the left and right of the active head region in each successor step
are copies of the tape cells to the left and right in the predecessor
step.

To ensure that the tgds within each procedure are acyclic, the schema
$\Sch$ includes additional ternary relations $T'$ and $H'$ and binary
relations $L'$ and $R'$.  These, respectively, have exactly the same
interpretation as do $T$, $H$, $L$, and $R$, but serve only as dummies
that we use to divide what would otherwise be cyclic tgds in one
procedure into matching acyclic sets that appear in separate
procedures.

In the remainder of the proof, as in the explanations of $T$ and $H$
above, we use constants to denote the grid cells, states, and tape
symbols of $M$.  We do not, however, permit tgds to involve constants,
and so assume that for each state and symbol constant $c$, the schema
$\Sch$ contains a distinct unary relation $C$ that is never in scope
and that simulates the use of the constant $c$ in the sense that the
interpretation of $C$ is $C^I = \{c\}$ in the initial instance $I$ and
hence in every subsequent instance.  For example, an atom of the form
$T(x,a,y)$ for some symbol constant $a$ should be read as shorthand
for the conjunction $T(x,v,y) \wedge A(v)$.

\medskip

\textbf{Initial instance.}  The initial instance $I$ represents the initial
machine configuration as follows.
\begin{itemize}
\item $T^I$ consists of triples
  $\{ (c_0,B,c_1), (c_1,a_1,c_2), \dots, (c_m,a_m,c_{m+1}),
  (c_{m+1},E,c_{m+2}) \}$.  Here each $c_j$ is a fresh element
  representing a tape cell.
\item $H^I$ consists of the triple $(c_1,q_0,c_2)$, meaning that $M$
  starts with the head on cell 1 in the initial state $q_0$.
\item $L^I$ and $R^I$ are empty, as are ${T'}^I$, ${H'}^I$, ${L'}^I$,
  and ${R'}^I$.
\end{itemize}

\medskip

\textbf{Set of procedures.}  The set of procedures $\Pi$ contains one
procedure $P^d$ for each transition $d$ in $\delta$, plus a procedure
$P^{lr}$ that copies tape cell contents across configurations, and a
procedure $P^{tr}$ that translates the dummy relations back into the
primary relations.

We define the translation procedure
$P^{tr} = (\Scope^{tr},\Cpre^{tr},\Cpost^{tr},\Q_\safe^{tr})$ so that
$\Scope^{tr} = \{ T[*], H[*], L[*], R[*] \}$, $\Cpre^{tr}$ is empty,
$\Q_\safe^{tr} = \{ T(x_1,y_1,z_1), H(x_2,y_2,z_2), L(x_3,y_3),
R(x_4,y_4)\}$, and
\begin{displaymath}
  \Cpost^{tr} =
  \left\{
    \begin{array}{lll}
      T'(x,y,z) & \rightarrow & T(x,y,z) \\
      H'(x,y,z) & \rightarrow & H(x,y,z) \\
      L'(x,y) & \rightarrow & L(x,y) \\
      R'(x,y) & \rightarrow & R(x,y).
    \end{array}
  \right.
\end{displaymath}
Comparison with the definition shows that $P^{tr}$ has safe scope.

For each transition
$d = (q,a) \mapsto (q',a',\Moveleft/\Moveright)$ in $\delta$, we
define the transition procedure $P^d$ so that
$\Scope^d = \{ T'[*], H'[*], L'[*], R'[*] \}$, $\Cpre^d$ is empty,
$\Q_\safe^d = \{ T'(x_1,y_1,z_1), H'(x_2,y_2,z_2), L'(x_3,y_3),
R'(x_4,y_4)\}$, and $\Cpost^d$ consists of tgds that characterize
local transition changes.
\begin{enumerate}

\item If $d$ is a right-moving transition, the first tgd in $\Cpost^d$
  encodes motion that does not extend the used portion of the tape by
  visiting new cells.
  \begin{eqnarray*}
    \lefteqn{T(x,a,y) \wedge H(x,q,y) \wedge T(y,v,z) \rightarrow {}} \\
    & \exists x'\ \exists y'\ \exists z'
    & T'(x',a',y') \wedge T'(y',v,z') \wedge H'(y',q',z') \wedge {} \\
    && L'(x,x') \wedge R'(y,y').
  \end{eqnarray*}
  \begin{center}
    \begin{tikzpicture}[node distance=15mm]
      \node (x) {$x$};
      \node (y) [right of=x] {$y$};
      \node (z) [right of=y] {$z$};
      \draw (x) edge node [above] {$a$} node [below] {$q$} (y);
      \draw (y) edge node [above] {$v$} (z);

      \node (mid) [right of=z] {\Large $\Longrightarrow$};

      \node (x1) [right of=mid] {$x$};
      \node (y1) [right of=x1] {$y$};
      \node (z1) [right of=y1] {$z$};
      \draw (x1) edge node [above] {$a$} node [below] {$q$} (y1);
      \draw (y1) edge node [above] {$v$} (z1);
      \node (x1p) [below of=x1] {$x'$};
      \node (y1p) [below of=y1] {$y'$};
      \node (z1p) [below of=z1] {$z'$};
      \draw (x1p) edge node [above] {$a'$} (y1p);
      \draw (y1p) edge node [above] {$v$}  node [below] {$q'$}(z1p);
      \draw (x1) edge node [left] {$L$} (x1p);
      \draw (y1) edge node [right] {$R$} (y1p);
    \end{tikzpicture}
  \end{center}
  Here variables $x'$, $y'$, and $z'$ name the cells in the successor
  configuration that correspond to the cells $x$, $y$, and $z$.  The
  symbols $a$, $a'$, $q$ and $q'$ denote constants, with
  $a, a' \in \alphabet$ and $q, q' \in Q$.

  If $d$ is a left-moving transition, the first tgd is defined
  correspondingly.
  \begin{eqnarray*}
    \lefteqn{T(x,a,y) \wedge H(x,q,y) \wedge T(z,v,x) \rightarrow {}} \\
    & \exists x'\ \exists y'\ \exists z'
    & T'(x',a',y') \wedge T'(z',v,x') \wedge H'(z',q',x') \wedge {} \\
    && L'(x,x') \wedge R'(y,y').
  \end{eqnarray*}
  \begin{center}
    \begin{tikzpicture}[node distance=15mm]
      \node (x) {$x$};
      \node (y) [right of=x] {$y$};
      \node (z) [left of=x] {$z$};
      \draw (x) edge node [above] {$a$} node [below] {$q$} (y);
      \draw (x) edge node [above] {$v$} (z);

      \node (mid) [right of=y] {\Large $\Longrightarrow$};

      \node (z1) [right of=mid] {$z$};
      \node (x1) [right of=z1] {$x$};
      \node (y1) [right of=x1] {$y$};
      \draw (x1) edge node [above] {$a$} node [below] {$q$} (y1);
      \draw (x1) edge node [above] {$v$} (z1);
      \node (x1p) [below of=x1] {$x'$};
      \node (y1p) [below of=y1] {$y'$};
      \node (z1p) [below of=z1] {$z'$};
      \draw (x1p) edge node [above] {$a'$} (y1p);
      \draw (x1p) edge node [above] {$v$}  node [below] {$q'$}(z1p);
      \draw (x1) edge node [left] {$L$} (x1p);
      \draw (y1) edge node [right] {$R$} (y1p);
    \end{tikzpicture}
  \end{center}

\item If $d$ is a right-moving transition, the second tgd in
  $\Cpost^d$ encodes right moves that extend the used portion of the
  tape by moving the $E$ marker right and inserting a blank cell.
  \begin{eqnarray*}
    \lefteqn{T(x,a,y) \wedge H(x,q,y) \wedge T(y,E,z) \rightarrow {}} \\ 
    & \exists x'\ \exists y'\ \exists z'\ \exists u'
    & T'(x',a',y') \wedge T'(y',\blankcell,z') \wedge
      H'(y',q',z') \wedge T'(z',E,u') \wedge {} \\
    && L'(x,x') \wedge R'(z,z').
  \end{eqnarray*}
  \begin{center}
    \begin{tikzpicture}[node distance=15mm]
      \node (x) {$x$};
      \node (y) [right of=x] {$y$};
      \node (z) [right of=y] {$z$};
      \draw (x) edge node [above] {$a$} node [below] {$q$} (y);
      \draw (y) edge node [above] {$E$} (z);

      \node (mid) [right of=z] {\Large $\Longrightarrow$};

      \node (x1) [right of=mid] {$x$};
      \node (y1) [right of=x1] {$y$};
      \node (z1) [right of=y1] {$z$};
      \draw (x1) edge node [above] {$a$} node [below] {$q$} (y1);
      \draw (y1) edge node [above] {$E$} (z1);
      \node (x1p) [below of=x1] {$x'$};
      \node (y1p) [below of=y1] {$y'$};
      \node (z1p) [below of=z1] {$z'$};
      \node (u1p) [right of=z1p] {$u'$};
      \draw (x1p) edge node [above] {$a'$} (y1p);
      \draw (y1p) edge node [above] {$\blankcell$} node [below] {$q'$}(z1p);
      \draw (z1p) edge node [above] {$E$} (u1p);
      \draw (x1) edge node [left] {$L$} (x1p);
      % \draw (y1) edge node {$L$} (y1p);
      \draw (z1) edge node [right] {$R$} (z1p);
    \end{tikzpicture}
  \end{center}
  We include the $R'(y,y')$ atom in the right-hand side of the tgd for
  uniformity, though it serves no other purpose in this case, as there
  are no cells to copy to the right of the cell marked with $E$ in the
  prior configuration.

  If $d$ is a left-moving transition, this second kind of tgd is
  defined correspondingly, this time moving the beginning of tape
  marker to the left.  In this tgd, the $L(z,z')$ atom is included for
  uniformity.
  \begin{eqnarray*}
    \lefteqn{T(x,a,y) \wedge H(x,q,y) \wedge T(z,B,x) \rightarrow {}} \\ 
    & \exists x'\ \exists y'\ \exists z'\ \exists u'
    & T'(x',a',y') \wedge T'(z',\blankcell,x') \wedge
      H'(z',q',x') \wedge T'(u',B,z') \wedge {} \\
    && L'(z,z') \wedge R'(y,y').
  \end{eqnarray*}
  \begin{center}
    \begin{tikzpicture}[node distance=15mm]
      \node (x) {$x$};
      \node (y) [right of=x] {$y$};
      \node (z) [left of=x] {$z$};
      \draw (x) edge node [above] {$a$} node [below] {$q$} (y);
      \draw (x) edge node [above] {$B$} (z);

      \node (mid) [right of=y] {\Large $\Longrightarrow$};

      \node (z1) [right of=mid] {$z$};
      \node (x1) [right of=z1] {$x$};
      \node (y1) [right of=x1] {$y$};
      \draw (x1) edge node [above] {$a$} node [below] {$q$} (y1);
      \draw (x1) edge node [above] {$B$} (z1);
      \node (x1p) [below of=x1] {$x'$};
      \node (y1p) [below of=y1] {$y'$};
      \node (z1p) [below of=z1] {$z'$};
      \node (u1p) [left of=z1p] {$u'$};
      \draw (x1p) edge node [above] {$a'$} (y1p);
      \draw (x1p) edge node [above] {$\blankcell$}  node [below] {$q'$}(z1p);
      \draw (z1p) edge node [above] {$B$} (u1p);
      \draw (z1) edge node [left] {$L$} (z1p);
      \draw (y1) edge node [right] {$R$} (y1p);
    \end{tikzpicture}
  \end{center}

\end{enumerate}  
As sets of such tgds are acyclic, comparison with the definition shows
that each $P^d$ has safe scope.

Finally, we define the left-right copying procedure
$P^{lr} = (\Scope^{lr},\Cpre^{lr},\Cpost^{lr},\Q_\safe^{lr})$ so that
$\Scope^{lr} = \{ T'[*], L'[*], R'[*] \}$, $\Cpre^{lr}$ is empty,
$\Q_\safe^{lr} = \{ T'(x_1,y_1,z_1), L'(x_2,y_2), R'(x_3,y_3)\}$, and
$\Cpost^{lr}$ contains two tgds that simply copy the contents of any
cells outside of the active region in one configuration to the
corresponding cells in the successor configuration.

The first tgd copies contents of cells to the left of the active
region.
\begin{eqnarray*}
  T(x,v,y) \wedge L(y,y') & \rightarrow & \exists x'\ T'(x',v,y') \wedge L'(x,x')
\end{eqnarray*}
\begin{center}
  \begin{tikzpicture}[node distance=15mm]
    \node (x) {$x$};
    \node (y) [right of=x] {$y$};
    \draw (x) edge node [above] {$v$} (y);
    \node (yp) [below of=y] {$y'$};
    \draw (y) edge node [left] {$L$} (yp);

    \node (mid) [right of=y] {\Large $\Longrightarrow$};

    \node (x1) [right of=mid] {$x$};
    \node (y1) [right of=x1] {$y$};
    \draw (x1) edge node [above] {$v$} (y1);
    \node (x1p) [below of=x1] {$x'$};
    \node (y1p) [below of=y1] {$y'$};
    \draw (x1p) edge node [above] {$v$} (y1p);
    \draw (x1) edge node [left] {$L$} (x1p);
    \draw (y1) edge node [left] {$L$} (y1p);
  \end{tikzpicture}
\end{center}

The second tgd copies contents of cells to the right of the active
region.
\begin{eqnarray*}
  T(x,v,y) \wedge R(x,x') & \rightarrow & \exists y'\ T'(x',v,y') \wedge R'(y,y')
\end{eqnarray*}
\begin{center}
  \begin{tikzpicture}[node distance=15mm]
    \node (x) {$x$};
    \node (y) [right of=x] {$y$};
    \draw (x) edge node [above] {$v$} (y);
    \node (xp) [below of=x] {$x'$};
    \draw (x) edge node [right] {$R$} (xp);

    \node (mid) [right of=y] {\Large $\Longrightarrow$};

    \node (x1) [right of=mid] {$x$};
    \node (y1) [right of=x1] {$y$};
    \draw (x1) edge node [above] {$v$} (y1);
    \node (x1p) [below of=x1] {$x'$};
    \node (y1p) [below of=y1] {$y'$};
    \draw (x1p) edge node [above] {$v$} (y1p);
    \draw (x1) edge node [right] {$R$} (x1p);
    \draw (y1) edge node [right] {$R$} (y1p);
  \end{tikzpicture}
\end{center}
Comparison with the definition shows that $P^{lr}$ has safe scope.

\medskip

\textbf{Undecidability.}  Let $\Sigma$ be the readiness constraint
consisting of a single tgd $t$ that encodes halting of the machine as
entry of the machine into the halting state $q_h$ as the following.
\begin{eqnarray*}
  H(x,q_0,y) &\rightarrow& \exists x'\ \exists y'\ H(x',q_h,y')
\end{eqnarray*}
\begin{center}
  \begin{tikzpicture}[node distance=15mm]
    \node (x) {$x$};
    \node (y) [right of=x] {$y$};
    \draw (x) edge node [below] {$q_0$} (y);

    \node (mid) [right of=y] {\Large $\Longrightarrow$};

    \node (x1) [right of=mid] {$x$};
    \node (y1) [right of=x1] {$y$};
    \draw (x1) edge node [below] {$q_0$} (y1);
    \node (x1p) [below of=x1] {$x'$};
    \node (y1p) [below of=y1] {$y'$};
    \draw (x1p) edge node [below] {$q_h$} (y1p);
  \end{tikzpicture}
\end{center}
Here $x'$ and $y'$ denote cells in some successor configuration of
the initial configuration, not necessarily immediate successors of
$x$ and $y$.

This tgd is cyclic, but we explain how to transform it into an acyclic tgd at the end of the proof. 

\medskip

We now show that $I$ can be readied for $t$ using $\Pi$ if and only
if $M$ halts on input $w$.  

\noindent
\textbf{(Readiable $\Rightarrow$ halts):} Assume that $I$ can be
readied for $t$ using $\Pi$.  Let $P_1,\dots,P_n$ be a witnessing
sequence of procedures, and let $I_0, \dots, I_n$ be a sequence of
instances such that $I_0 = I$ and for each $i>0$,
$I_i = \chase_i (I_{i-1})$, where $\chase_i$ denotes chase with
respect to the set $\Cpost$ of tgds in the output constraints of
procedure $P_i$.  In particular, $I_n$ is an instance produced out of
chasing dependencies in the output constraints of the procedures in
$\Pi$.

By construction, the initial instance $I$ characterizes the starting
configuration of $M$ on $w$, in which the state specification is
$H(c_1,q_0,c_2)$.  Also by construction, if a transition procedure
applies to a configuration, then $M$ must have a corresponding
transition.  Because $M$ is deterministic, if no transition procedure
applies even after all copying and translation dependencies have been
chased, it must be because the final configuration is one in which
$H(c_i,q_h,c_{i+1})$ for some $i$, meaning that the computation has
halted.  Therefore $I$ can be readied for $t$ only if some sequence of
procedures yields a transition to $q_h$.

\medskip

\noindent
\textbf{(Halts $\Rightarrow$ readiable):} Assume now that $M$ halts on
input $w$.  Let $d_1,\dots,d_n$ be the sequence of transitions taken
by $M$ in moving from the initial configuration to the final halting
configuration, and let $P_i$ denote the transition procedure for
$d_i$.

% We claim that there is a sequence of the form
% ${P_1,S_1,P^{tr},P_2,S_2,P^{tr},\dots,P_n,S_n,P^{tr}}$ where each $S_i$ stands
% for a finite number of repetitions of $P^{lr}$.

We claim that $t$ satisfies each instance in
\begin{displaymath}
  \outcome_{P_1,P^{lr},P^{tr},P_2,P^{lr},P^{tr},\dots,P_n,P^{lr},P^{tr}}(I).
\end{displaymath}
To see this, note that $t$ is satisfied as long as the conjunctive
query $\exists x' \exists y' H(x',q_h,y')$ is satisfied.  That query,
in turn, is satisfied if the chase of $I$ satisfies each of the
$\Cpost$ dependencies of $P^{lr}$, $P^{tr}$, and $P^{d_i}$ for each
$i$.  By construction, each of the transition procedures produces, by
chase and in conjunction with the copying and translation procedures,
each successive configuration, so the final transition $d_n$ will
require, after translation, that $H$ contains a triple of the form
$(c_i,q_h,c_{i+1})$.

We conclude that $I$ can be readied for $t$ using $\Pi$ if and only
if $M$ halts on input $w$.  As the halting problem is undecidable,
so must also be the constraint readiness problem.

\bigskip

We note that the claim ofProposition
\ref{prop-readi-undecidable-constraints} holds even if $\Sigma$ is
acyclic.  The tgd $t$ employed in the proof is cyclic, but essentially
the same argument would hold were one to use instead the acyclic tgd
\begin{eqnarray*}
  T(c_0,B,c_1) &\rightarrow& \exists x'\ \exists y'\ H(x',q_h,y'),
\end{eqnarray*}
where here $c_0$ and $c_1$ are the constants so named in the initial
instance.

\subsection{Proof of Theorem \ref{theo-readi-decidable-constraints}}

It is easy to see that using procedures with safe scope we will never be able to ready an instance 
for an egd: if the instance $I$ does not satisfy an egd then this violation will be carried over no matter what 
procedures we apply. For this reason, we focus on tgds only. 

Let then $I$, $\Pi$ and $\sigma$ be as outlined in the statement of the theorem (where $\sigma$ is a tgd). 
Let $D$ be the number of different elements in $I$, and $\Sch$ the schema of $I$. 

Further, assume a sequence 
$P_1,\dots,P_\ell$ such that every instance in $\outcome_{P_1,\dots,P_\ell}(I)$ satisfies $\sigma$. 
By Theorem \ref{theo-outcome-fulltgds}, there is an $\skb = (J,\Gamma,\Scope)$ that represents the set 
$\outcome_{P_1,\dots,P_\ell}(I)$. 

We now show that there is a sequence $P'_1,\dots,P'_n$ of procedures in $\Pi$ with the same property, but where now $n$ is bounded 
exponentially on the size of $I$ and $\Pi$. First, 
since all tgds involved in all procedures in $\Pi$ are full tgds, the size of $J$ is bounded by $|D|^|\Sch|$. 

From Proposition \ref{tgd-entailment-decidability-prop}, it must be that the instance $K$ resulting of taking 
$J$ together with the frozen body of the premise of $\sigma$ is such that the consequence of $\sigma$ holds in 
$\chase_\Gamma(K)$. The size of $K$ is bounded by $(|D| + |Q|)^|\Sch|$, and thus we can assume that 
$\Gamma$ contains only tgds of size bounded by $(|D| + |Q|)^|\Sch|$, as any bigger tgd can be equivalen to 
a tgd of such size when chasing $K$. 

We can enumerate all sets $\Omega$ containing tgds of size at most $(|D| + |Q|)^|\Sch|$, and we know 
that the number of different such sets is bounded by $2^{(|D| + |Q|)^|\Sch|}$. 

Let $\skb_i = (K_i,\Gamma_i,\Scope_i)$ be the SKB representing the set  $\outcome_{P_1,\dots,P_i}(I)$. 
As mentioned, we assume without loss of generality that each $\Gamma_i$ contains tgds of size at most 
$(|D| + |Q|)^|\Sch|$. 

If no $K_i$ is equal, then the sequence $\ell$ must then be of length at most $|D|^|\Sch|$, because each new procedure must at least 
introduce some data in $K_i$. 
Further, if for every maximal sequence $p,p+1,p+2,\dots,q$ such that $K_p = K_q$ one cannot find two numbers 
$r_1$ and $r_2$ such that $\Gamma_{r_1}$ is logically equivalent to $\Gamma_{r_2}$, then 
$\ell$ must then be of length at most $|D|^|\Sch| \cdot \Pi * 2^{(|D| + |Q|)^|\Sch|}$. 

On the other hand, if there are  $r_1$ and $r_2$ such that $\Gamma_{r_1}$ is logically equivalent to $\Gamma_{r_2}$, 
we can just prune the sequence from $r_1 +1$ to $r_2$. 

With the above observations in hand we then outline the \twonexptime algorithm: 

\begin{itemize}
\item Guess a sequence $P_1,\dots,P_n$ of procedures. We know it is of length at most doubly-exponential in the size of the input. 
\item Compute the SKB $\skb$ representing $\outcome_{P_1,\dots,P_n}(I)$. 
\item We need also to guess all appropriate chase steps to show that $\skb \models \sigma$, as explained in the proof of proposition \ref{tgd-entailment-decidability-prop}. 
\end{itemize}

\subsection{Proof of Proposition \ref{prop-readi-undecidable-queries}}

Follows from the proof of \ref{prop-readi-undecidable-constraints}, simply by using the query 
$H(x',q_h,y')$ instead of the tgd employed therein. 

\subsection{Proof of Theorem \ref{theo-readi-decidable-queries}}

We use the same argument as in the proof of Propositions \ref{prop-query-answering} and \ref{prop-dyn-query-answering}. 
This time, in addition we need to 
guess a sequence of procedures that yield the appropriate chase rule. 

More precisely, as in the proof of Theorem \ref{theo-readi-decidable-constraints}, we can bound the size of the 
sequence of procedures. 
In order to do that, assume a sequence 
$P_1,\dots,P_\ell$ of procedures from $\Pi$ such that every instance in $\outcome_{P_1,\dots,P_\ell}(I)$ satisfies $Q$. 
By Theorem \ref{theo-outcome-fulltgds}, there is an $\skb = (J,\Gamma,\Scope)$ that represents the set 
$\outcome_{P_1,\dots,P_\ell}(I)$. 
Let also $\skb_i = (K_i,\Gamma_i,\Scope_i)$ be the SKB representing the set  $\outcome_{P_1,\dots,P_i}(I)$. 

Note however that, in contrast with the proof of Theorem \ref{theo-readi-decidable-constraints}, we do not need to focus on maintaining 
the sets $\Gamma_i$ of intermediate SKBs, as we only care about the instances $K_i$. 

We can then construct a corresponding exponential sequence by pruning out all procedures $P_i$ where $J_i$ is the same 
instance as $J_{i-1}$. 

For the \nexptime\ algorithm we can then guess this exponential sequence of procedures whose outcome is represented by an SKB of the form 
$(J,\Gamma',\Scope')$, guess an appropriate homomorphism from $Q$ to $J$ and guess the necessary chase steps to produce 
the image of $Q$ over $J$ as in the proof of Propositions \ref{prop-query-answering} and \ref{prop-dyn-query-answering}.

\subsection{Proof of Proposition \ref{ref-rep-decidable}}

\newcommand{\Schmin}{\Sch_\text{min}}

%We say that a schema $\Sch'$ is an extension of a schema $\Sch$ if (1) all relations in $\Sch$ are contained in 
%$\Sch'$, and (2), all attributes associated to a relation $R$ in $\Sch$ are also associated to $R$ in 
%$\Sch'$ (but there could be more attributes). 

Let  $P = (\Scope,\Cpre,\Cpost,\Q_\safe)$. We first show how to construct, for each instance $I$ over a schema $\Sch$, the \emph{minimal} 
schema $\Schmin$ such that all pairs $(J,\Sch')$ that are possible outcomes of applying $P$ over $(I,\Sch)$  are such that 
$\Sch'$ extend $\Schmin$. 

The algorithm receives a procedure $P$ and a schema $\Sch$ and outputs either $\Schmin$, if the procedure is applicable, or a failure signal in case there is no schema satisfying the output constraints of the procedure. 
Along the algorithm we will be assigning numbers to some of the relations in $\Schmin$. This is important to be able to decide failure. 

\medskip
\noindent
\textbf{Algorithm $A(P,\Sch)$ for constructing $\Schmin$}\\
\textbf{Input}: procedure $P = (\Scope,\Cpre,\Cpost,\Q_\safe$) and schema $\Sch$. \\
\textbf{Output}: either \emph{failiure} or a schema $\Schmin$.

\begin{enumerate}
\item If $\Sch$ does not satisfy the structural constraints in $\Cpre$ or is not compatible with either $\Q_\safe$ or $Q_{\Sch \setminus \Scope}$, output failure. Otherwise, continue. 
\item Start with $\Schmin = \emptyset$. 
\item For each total query $R$ in $\Q_\safe$, assume that $|\Sch(R)| = k$. Set $\Schmin(R) = \Sch(R)$, and label $R$ with $k$.  

%\item Add to $\Schmin$ all relations $R$ mentioned in a total tgd in $\Cpost$ (if they are not already part of $\Schmin$), without associating any attributes to them
\item Add to $\Schmin$ all relations $R$ mentioned in an atom $R[*]$ in $\Cpost$ (if they are not already part of $\Schmin$), without associating any attributes to them 

\item In the following instructions we construct a set $\Gamma(P,\Sch)$ of pairs of relations and attributes. Intuitively, 
a pair $(R,\{a_1,\dots,a_n\})$ in $\Gamma(P,\Sch)$ states that each schema in the output of $P$ must contain a relation $R$ with attributes 
$a_1,\dots,a_n$. 
\begin{itemize}
\item For each relation $R$ in $\Sch$ that is not mentioned in $\Scope$, add to $\Gamma(P,\Sch)$ the pair $(R,\Sch(R))$.
\item For each constraint $R[a_1,\dots,a_n]$ in $\Scope$, add the pair 
$(R,\Sch(R) \setminus \{a_1,\dots,a_n\})$ to $\Gamma(P,\Sch)$.
\item For each atom $R(a_1:x_1,\dots,a_n:x_n)$ in $\Q_\safe$ add to $\Gamma(P,\Sch)$ the pair 
$(R,\{a_1,\dots,a_n\})$. 
\item For each atom $R(a_1:x_1,\dots,a_n:x_n)$ in a tgd or egd in $\Cpost$ add to $\Gamma(P,\Sch)$ the pair 
$(R,\{a_1,\dots,a_n\})$. 
\item For each constraint $R[a_1,\dots,a_n]$ in $\Cpost$, add to $\Gamma(P,\Sch)$ the pair 
$(R,\{a_1,\dots,a_n\})$. 
\end{itemize}

\item For each pair $(R,A)$ in $\Gamma(P,\Sch)$, do the following.
\begin{itemize}
\item If $R$ is not yet in $\Schmin$, add $R$ to $\Schmin$ and set $\Schmin(R) = A$; 
\item If $R$ is in $\Schmin$, update $\Schmin(R) = \Schmin(R) \cup A$. 
%\item If $R$ is in $\Schmin$ but it is not marked, let $B$ be the set of attributes currently associated to $R$. Update $B = B \cup A$. 
%\item Is $R$ is in $\Schmin$, $R$ is marked and the attributes associated to $R$ are not a superset of $A$, output failure. 
\end{itemize}

%\item Repeat until there are no more labelings: For each total tgd $R \rightarrow S$ in $\Cpost$, 
%\begin{itemize}
%\item If both $R$ and $S$ are labelled with different numbers, output failure. 
%\item Otherwise if $R$ is labelled with $n$ but $S$ is not, label $S$ with $n$. 
%\item Otherwise if $S$ is labelled with $n$ but $R$ is not, label $S$ with $n$. 
%\end{itemize}

\item If $\Schmin$ contains a relation $R$ labelled with a number $n$ where, $\Schmin(R) > n$, output failure. Otherwise output $\Schmin$.  
\end{enumerate}

By direct inspection of the algorithm, we can state the following. 
\begin{observation}
\label{obs-gamma}
Let $P = (\Scope,\Cpre,\Cpost,\Q_\safe)$ be a relational procedure and $\Sch$ a relational schema. Then for each 
relation $R$ in $\Schmin$ with attributes $\{a_1,\dots,a_n\}$, every instance $I$ over $\Sch$ and every pair 
$(J,\Sch')$ in the outcome of applying $P$ to $(I)$, we have that 
$\Sch(R)$ is defined, with $\{a_1,\dots,a_n\} \subseteq \Sch(R)$. 
\end{observation}

%Furthermore, from the definition of total queries and the semantics of procedures it is not difficult to show that. 
%\begin{observation}
%\label{obs-label}
%Let $P = (\Scope,\Cpre,\Cpost,\Q_\safe)$ be a relational procedure and $\Sch$ a relational schema. 
%For each relation $R$ in $\Schmin$ that is labeled with a number $k$, 
%every instance $I$ over $\Sch$ and every instance  
%$J$ in the outcome set of applying $P$ to $I$, we have that $|\Sch(R)| = k$. 
%\end{observation}

Furthermore, the following lemma specifies, in a sense, the correctness of the algorithm. 
\begin{lemma}
\label{lem-schmin}
Let $P = (\Scope,\Cpre,\Cpost,\Q_\safe)$ be a relational procedure and $\Sch$ a relational schema. Then: 
\begin{itemize}
\item[i)] If $A(P,\Sch)$ outputs failure, either $P$ cannot be applied over any instance $I$ over $\Sch$, or 
for each instance $I$ over $\Sch$ the set  $\outcome_P(I)$ is empty. 
\item[ii)] If $A(P,\Sch)$ outputs $\Schmin$, then the schema of any instance in $\outcome_P(I)$ extends $\Schmin$.
\end{itemize}
\end{lemma}
\begin{proof}
For i), if some of the components of $P$ are not compatible with $\Sch$, or $\Sch$ does not satisfy the constraints in 
$\Cpre$, then clearly $P$ cannot be applied over any instance $I$ over $\Sch$. 
Assume then that $\Sch$ satisfies all compatibilities and preconditions in $P$, but $A(P,\Sch)$ outputs failure. 
Then $\Schmin$ contains a relation $R$ such that $|\Schmin(R)| = m$, but $R$ is labelled with number $k$, for 
$k < \ell$. From the algorithm, we this implies that $|\Schmin(R)| > |\Sch(R)$, but that there is a query 
$R$ in $\Q_\safe$. Clearly, $\Q_\safe$ cannot be preserved under any outcome, since by Observation \ref{obs-gamma} 
we require the schemas of outcomes to assign more attributes to $R$ than those assigned by $\Schmin$, and thus 
the cardinality of tuples in the answer of $R$ differs between $I$ and its possible outcomes. 
Finally, item ii) is a direct consequence of Observation \ref{obs-gamma}. 
\end{proof}

The algorithm $(A,P)$ runs in polynomial time, and that the total size of $\Schmin$ (measured as the number of relations and attributes) 
is at most the size of $\Sch$ and $P$ combined. Thus, to decide the applicability problem for a sequence $P_1,\dots,P_n$ of procedures, all we need 
to do is to perform subsequent calls to the algorithm, setting $\Sch_0 = \Sch$ and then using $\Sch_{i} =  A(P_i,\Sch_{i-1})$ as the input for the 
next procedures. If $A(P_n,\Sch_{n-1})$ outputs a schema, then the answer to the applicability problem is affirmative, otherwise if some call to 
$A(P_i,\Sch{i-1})$ outputs failure, the 
answer is negative. 

\subsection{Proof of proposition \ref{prop-rep-undec-nothingworks}}

The reduction, just as that of Proposition \ref{prop-ap-und}, is by reduction from 
the embedding problem for finite semigroups, and builds up from this proposition. 
Let us start by defining the procedures $\Pro_1$, $\Pro_2$ and $\Pro_3$. 

\medskip
\noindent
For procedure $\Pro_1$ we first build a set $\Gamma_1$ of tgs. This set is similar to the set 
$\Sigma$ used in Proposition \ref{prop-ap-und}, but using three additional \emph{dummy} 
relations $G^d$, $E^d$ and $G^\text{binary}$. 

First we add to $\Gamma_1$ dependencies that collect elements of $G$ into $D$, and that initialize 
$E$ as a reflexive relation. 

\begin{eqnarray*}
G(x,u,v) & \rightarrow & D(x) \\
G(u,x,v) & \rightarrow & D(x) \\
G(u,v,x) & \rightarrow & D(x) \\ 
D(x) & \rightarrow & E(x,x) 
\end{eqnarray*}

Next the dependency that states that $F$ contains everything in $R$ if some conditions about $E$ occur. 

\begin{eqnarray}
E(x,y) \wedge C(u,x) \wedge C(v,y) \wedge N(u,v) \wedge R(w) & \rightarrow & F(w)
\end{eqnarray}

The dependencies that assured that $E$ was an equivalence relation where acyclic, so we 
replace the right hand side with a dummy relation.  

\begin{eqnarray*}
E(x,y) & \rightarrow & E^d(y,x) \\
E(x,y) \wedge E(y,z) & \rightarrow & E^d(x,z)
\end{eqnarray*}

Next come the dependencies assuring $G$ is a total and associative function, using also dummy relations. 

\begin{eqnarray*}
D(x) \wedge D(y) & \rightarrow & G^\text{binary}(x,y) \\
G(x,y,u) \wedge G(u,z,v) \wedge G(y,z,w) & \rightarrow & G^d(x,w,v)
\end{eqnarray*}

Finally, the dependencies that were supposed to ensure that $E$ worked as the equality over 
function $G$, using again the dummy relations.  

\begin{eqnarray*}
G(x,y,z) \wedge E(x,x') \wedge E(y,y') \wedge E(z,z') & \rightarrow & G^d(x',y',z') \\
G(x,y,z) \wedge G(x',y',z') \wedge E(x,x') \wedge E(y,y') & \rightarrow & E^d(z,z')
\end{eqnarray*}

\noindent
We can now define procedure \textbf{$P_1$}: 

\noindent
$\Scope$: The scope of $P_1$ consists of relations $G$, $E$, $D$, $F$, $G^d$, $E^d$ and $G^\text{binary}$ which corresponds to the constraints 
$G[*], E[*], D[*], F[*], E^d[*], G^d[*]$ and $G^\text{binary}[*]$. 

\noindent
$\Cpre$: There are no preconditions for this procedure. 

\noindent
$\Cpost$: The postconditions are the tgds in $\Gamma_1$. 

\noindent
$\Q_\safe$: This query ensures that no information is deleted from all of $G$, $E$, $F$, $G^d$, $E^d$ and $G^\text{binary}$: 
$G(x,y,z) \wedge E(u,v) \wedge D(w) \wedge F(p) 
\wedge G^d(x',y',z') \wedge E^d(u',v') \wedge G^\text{binary}(a,b)$.

Note that, even though relations $G$ and $E$ are not mentioned in the right hand side of any tgd in $\Gamma_1$, 
they are part of the scope and thus they could be modified by the procedures $\Pro_1$. 

\medskip
\noindent
The procedure $\Pro_2$ has no scope, no safety queries, no precondition, and the only postcondition is the presence 
of a third attribute, say $C$, in $G^\text{binary}$, by using a structural constraint $G^\text{binary}[A,B,C]$ (to maintain consistency with our unnamed 
perspective, we assume that these three attributes are ordered $A <_\mathcal{A} B <_\mathcal{A} C$). 

\medskip
\noindent
To define the final procedure, consider the following set of tgds $\Gamma_3$. 

\begin{eqnarray*}
E^d(x,y) & \rightarrow & E(x,y) \\ 
G^d(x,y,z) & \rightarrow & G(x,y,z) \\  
G^\text{binary}(x,y,z) & \rightarrow & G(x,y,z) \\ 
F(x) & \rightarrow & F^\text{check}(x) 
\end{eqnarray*}

\noindent
Then we define procedure $\Pro_3$ is as follows. 

\noindent
$\Scope$: The scope of $\Pro_3$ is again empty.  

\noindent
$\Cpre$: There are no preconditions for this procedure. 

\noindent
$\Cpost$: The postconditions are the tgds in $\Gamma_3$. 

\noindent
$\Q_\safe$: There are also no safety queries for this procedure. 

Let $\Sch$ be the schema containing relations 
$G$, $E$, $D$, $F$, $F^\text{check}$, $G^d$, $E^d$ and $G^\text{binary}$ and $R$. The attribute names are of no importance 
for this proof, except for $G^\text{binary}$, which associates attributes $A$ and $B$. 

Given a finite semigroup $\textbf A$, we construct now the following instance $I_{\textbf A}$: 
\begin{itemize}
\item $E^{I_\textbf{A}}$ contains the pair $(a_i,a_i)$ for each $1 \leq i \leq n$ (that is, for each element of $A$); 
\item $G^{I_\textbf{A}}$ contains the triple $(a_i,a_j,a_k)$ for each $a_i,a_j,a_k \in A$ such that $g(a_i,a_j) = a_k$;  
\item All of $D^{I_\textbf{A}}$, $F^{I_\textbf{A}}$ and ${F^\text{check}}^{I_\textbf{A}}$ are empty; 
\item $R^{I_\textbf{A}}$ has a single element $d$ not used elsewhere in $I_\textbf{A}$
\item $C^{I_\textbf{A}}$ contains the pair $(i,a_i)$ for each $1 \leq i \leq n$; and 
\item $N^{I_\textbf{A}}$ contains the pair $(i,j)$ for each $i \neq j$, $1 \leq i \leq n$ and $1 \leq j \leq n$. 
\end{itemize}

Let us now show $\textbf{A} = (A,g)$ is embeddable in a finite semigroup if and only if $\outcome_{P_1,P_2,P_3}(I)$ is nonempty. 

\medskip
\noindent
($\Longrightarrow$) Assume that $\textbf{A} = (A,g)$ is embeddable in a finite semigroup, say the semigroup 
$\textbf B = (B,f)$, where $f$ is total. Let $J$ be the instance over $\Sch$ such that 
both ${E^d}^J$ and $E^J$ are the identity over $B$, $D^J = B$, both ${G^d}^J$ and  $G^J$ contains a pair $(b_1,b_2,b_3)$ if and only if 
$f(b_1,b_2) = b_3$; ${G^\text{binary}}^J$ is the projection of $G^J$ over its two first attirbutes, 
$F^J$ and ${F^\text{check}}^J$ are empty and relations $N$, $C$ and $R$ are interpreted as in $I_\textbf{A}$. 

It is easy to see that $J$ is in the outcome of applying $\Pro_1$ over $I$. 
Now, let $\Sch'$ be the extension of $\Sch$ where $G^\text{binary}$ has an extra attribute, $C$, 
and $K$ is an instance over $\Sch'$ that is just like $J$ except that ${G^\text{binary}}^K$ is now 
the same as $G^J$ (and therefore $G^K$). 
By definition we obtain that $K$ is a possible outcome of applying $\Pro_2$ over $J$, and therefore 
$K$ is in $\outcome_{\Pro_1,\Pro_2}(I)$. 
Furthermore, one can see that the same instance $K$ is again an outcome of applying $\Pro_3$ over 
$K$, therefore obtaining that $\outcome_{\Pro_1,\Pro_2,\Pro_3}(I)$ is nonempty. 

\medskip
\noindent
($\Longleftarrow$) Assume now that there is an instance $L \in \outcome_{P_1,P_2,P_3}(I)$. Then by definition there are instances $J$ and $K$ 
such that $J$ is in $\outcome_{\Pro_1}(I)$, $K$ is in $\outcome_{\Pro_2}(J)$ and 
$L$ is in $\outcome_{\Pro_3}(K)$. 

Let $J^*$ be the restriction of $J$ over the schema $\Sch$. From a simple inspection of $\Pro_1$ we have 
that $J^*$ satisfies as well the dependencies in $\Pro_1$, so that 
$J^*$ is in $\outcome_{\Pro_1}(I)$.

Let now $\Sch'$ be the extension of $\Sch$ that assigns also attribute $C$ to $G^\text{binary}$. Now, since $K$ is an outcome of $P_2$ over $J$ 
and $P_2$ has no scope, if we define $K^*$ as the restriction of $K$ over $\Sch'$, then 
clearly $K^*$ must be in the outcome of applying $\Pro_2$ over $J^*$. 
Note that, by definition of $\Pro_3$ (since its scope is empty), the restriction of $L$ up to the schema of $K$ must be the same instance as $K$, 
and therefore the restriction $L^*$ of $L$ to $\Sch'$ must be the same instance than $K^*$. Furthermore, since $L$ (and thus $L^*$) satisfies the 
constraints in $\Pro_3$, and the constraints only mention relations and atoms in $\Sch'$, 
we have that $K^*$ must be an outcome of applying $\Pro_3$ over $(K^*,\Sch')$. 

We now claim that $K^*$ satisfy all tgds (1)-(11) in the proof of Propositon \ref{prop-ap-und}. 
Tgds (1-3) and (6) are immediate from the scopes of procedures, and the satisfaction for all the remaining ones 
is shown in the same way. For example, to see that $K^*$ satisfies $E(x,y) \rightarrow E(y,x)$, note that 
$J^*$ already satisfies $E(x,y) \rightarrow E^d(y,x)$. From the fact that the interpretations of $E^d$ and $E$ 
are the same over $J^*$ and $K^*$ and that $K^*$ satisfies $E^d(x,y) \rightarrow E(x,y)$ we 
obtain the desired result. 

Finally, since $K^*$ satisfies $F(x)\ \rightarrow\ F^\text{check}(x)$, and the interpretation of 
$F^\text{check}$ over all of $I$, $J^*$ and $K^*$ must be empty, we have that 
the interpretation of $F$ over $K^*$ is empty as well. Given that $K^*$ satisfies all dependencies in $\Sigma$, 
it must be the case that the left hand side of the tgd (11) is not true $K^*$, for any possible assignment. 
By using the same argument 
as in the proof of Proposition \ref{prop-ap-und} we obtain that $\textbf{A} = (A,g)$ is embeddable in a finite semigroup. 

\subsection{Proof of proposition \ref{prop-rep-safe-scope-dyn}}

Follows from Proposition \ref{prop-minimal} and Proposition \ref{ref-rep-decidable}. We need to check first whether each procedure in the 
sequence is applicable. Once we do that, from Proposition \ref{prop-minimal} we know that the resulting outcome is non-empty. 

\subsection{Proof of proposition \ref{prop-minimal}}

For the proof we assume that all procedures does not use preconditions. We can treat them  by first doing an initial check on compatibility that only complicates the proof. 

We also specify an alternative set of representatives for conditional instances (which is actually the usual one). The 
set $\hat \rep(G)$ of representatives of a conditional instance $G$ is simply 
$\hat \rep(G)= \{I \mid$ there is a substitution $\nu$ such that $\nu(T) \subseteq I \}$. That is,  $\hat \rep (G)$ only 
specifies instances over the same schema as $G$. 
The following lemma allows us to work with this representation instead; it is immediate from the definition of safe scope procedures. 

\begin{lemma}
\label{lem-usual-rep}
If $G$ is a conditional instance, then (1) $\hat \rep(G) \subseteq \rep(G)$, and (2) an instance $J$ is minimal for $\rep(G)$ if and only if it is minimal 
for $\hat \rep(G)$. 
\end{lemma}

Moreover, from the fact that procedures with safe scope are acyclic, we can state Theorem 5.1 in \cite{APR13} in the following terms:  

\begin{lemma}[\cite{APR13}]
\label{lem-chase}
Given a set $\Sigma$ of tgds and a positive conditional instance $G$, one can construct, in exponential time, a positive conditional instance 
$G'$ such that (1) $\hat \rep(G') \subseteq \hat \rep(G)$ and (2) all minimal models of $\hat \rep(G')$ satisfy $\Sigma$. 
\end{lemma}

Moreover, by slightly adapting the proof of Proposition 4.6 in \cite{APR13}, we can see that the conditional instance constructed above has even better 
properties. In order to prove this theorem all that one needs to do is to adapt the notion of solutions for data exchange into a scenario where 
the target instance may already have some tuples (which will not fire any dependencies because of the safeness of procedures). 

\begin{lemma}[\cite{APR13}]
\label{lem-minimal-chase}
Let $P = \Scope,\Cpre,\Cpost,\Q_\safe$ be a procedure with safe scope, and let $G$ be a positive conditional instance. 
Then one can construct (in exponential time) a positive conditional instance $G'$ such that, for every minimal instance $I$ of $\hat \rep(G)$, the 
set $\hat \rep(G')$ contains all minimal instances in $\outcome_{\Pro}(I)$, and for every minmal instance $J$ in $\hat \rep(G')$ 
there is a minimal instance $I$ of $\rep(G)$ such that $J$ is minimal for $\outcome_{\Pro}(I)$.
\end{lemma}

Finally, we can show the key result for this proof. 
\begin{lemma}
\label{lem-seq-safe}
Let $\Inst$ be a set of instances, and $G$ a positive conditional table that is minimal for $\Inst$, and 
$\Pro  = (\Scope,\Cpre,\Cpost,\Q_\safe)$ a procedure with safe scope.  
Then either $\outcome_{\Pro}(\Inst) = \emptyset$ or one can construct, in exponential time, a  
positive conditional instance $G'$ 
%with global condition  
such that 
\begin{itemize}
\item[i)] $\outcome_{\Pro}(\Inst) \subseteq \rep(G')$; and 
\item[ii)] If $J$ is a minimal instance in $\rep(G')$, then $J$ is also minimal in 
$\outcome_{\Pro}(\Inst) $. 
\end{itemize}
\end{lemma}

\begin{proof}
Using the chase procedure mentioned in Lemma \ref{lem-minimal-chase}, we see that the conditional table $G'$ produced in this lemma 
satisfies the conditions of this Lemma, for $\hat \rep(G)$.  

For i), let $J$ be an instance in $\outcome_{\Pro}(\Inst)$. Then there is an instance $I$ in $\Inst$ such that 
$J \in \outcome{P}(I)$. Let $I^*$ be a minimal instance in $\Inst$ such that $I$ extends $I^*$. 
By our assumption we know that $I^*$ belongs to $\rep(G)$, and 
since $I^*$ is minimal it must be the case that $I^*$ belongs (and is minimal) for $\hat \rep(G)$. 
Therefore, by Lemma \ref{lem-minimal-chase} we have that $\hat \rep(G')$ contains all minimal instances 
for $\outcome_P(I^*)$. But now notice that for every assignment $\tau$ and tgd $\lambda$ such that 
$(I^*,\tau)$ satisfies $\lambda$, we have that $(I,\tau)$ satisfy $\lambda$ as well. This means that 
every instance in the set $\outcome_P(I)$ must extend a minimal instance in $\outcome_P(I^*)$ 
(if not, then a tgd would not be satisfied due to some assignment that would not be possible to extend). 
Since every minimal instance in $\outcome_P(I^*)$ is in $\hat \rep(G')$, then by the semantics of conditional tables 
it must be the case that $J$ belongs to $\hat \rep(G')$ as well, and therefore to $\rep(G')$. 

Item [ii)] follows from the fact that any minimal instance in $\rep(G')$ must also be minimal for $\hat \rep(G')$ 
and a direct application of Lemma \ref{lem-minimal-chase}.
\end{proof}

The next Lemma constructs the desired outcomes for alter schema procedures. 

\begin{lemma}
\label{lem-seq-alter}
Let $\Inst$ be a set of instances, and $G$ a conditional table that is minimal for $\Inst$, and 
$\Pro  = (\Scope,\Cpre,\Cpost,\Q_\safe)$ an alter schema procedure.  
Then either $\outcome_{\Pro}(\Inst) = \emptyset$ or one can construct, in polynomial time, a  
conditional instance $G'$ 
%with global condition  
such that 
\begin{itemize}
\item[i)] $\outcome_{\Pro}(\Inst) \subseteq \rep(G')$; and 
\item[ii)] If $J$ is a minimal instance in $\rep(G')$, then $J$ is also minimal in 
$\outcome_{\Pro}(\Inst) $. 
\end{itemize}
\end{lemma}

\begin{proof}
Assume that $\outcome_{\Pro}(\Inst) \neq \emptyset$ (this can be easily checked in polynomial time). 
Then one can compute the schema $\Schmin$ from the proof of Proposition 
\ref{ref-rep-decidable}. This schema will add some attributes to some relations in the schema of $G$, and 
possibly some other relations with other sets of attributes. Let $\schema(G) = \Sch$. 

We extend $G$ to a positive conditional table $G'$ over $\Schmin$ as follows: 
\begin{enumerate}
\item For every relation $R$ such that $\Schmin(R) \setminus \Sch(R) = \{A_1,\dots,A_n\}$, with $n \geq 1$, for tuples from $G'$ by 
adding to each tuple in $G$ a fresh null value in each of the attributes $A_1,\dots,A_n$. 
\item For every relation $R$ such that $\Sch(R)$ is not defined, but $\Schmin(R)$ is defined, set $R^{G'} = \emptyset$
\end{enumerate}

The properties of the lemma now follow from a straightforward check. 
\end{proof}

The proof of Proposition \ref{prop-minimal} now follows from successive applications of Lemmas \ref{lem-seq-alter} and  \ref{lem-seq-safe}: 
one just need to compute the appropriate conditional table for each procedure in the sequence $P_1,\dots,P_n$. 
That each construction is in exponential size if the number $n$ of procedures is fixed, or doubly-exponential in other case, 
follows also from these Lemmas, as the size of the conditional table $G'$, for a procedure $P$ and a conditional table $G$, is at most 
exponential in $|G|^{|P|}$ (and thus if we have $n$ procedures of size $|P|$ the size is of order ${(|G|^{|P|})}^n$, or $|G|^{n|P|})$. 

\subsection{Proof of proposition \ref{prop-dyn-query-answering}}

Let $T$ be a conditional instance, $I$ an instance in $\hat \rep(T)$ and $Q$ a boolean conjunctive query. 
By definition of conditional instances and the fact that CQs are preserved under homomorphisms, we have the following fact. 

\begin{observation}
$Q$ holds in every minimal instance in $\hat \rep(T)$ if and only if $Q$ holds in every instance in $\rep(T)$
\end{observation}

Thus, if we want to compute the certain answers for a query $Q$ over a conditional instance $I$, all we need to do is 
to guess a counterexample: an assignment $\nu$ for $T$ such that the minimal instance $\nu(T)$ does not satisfy the query. 
We obtain (see e.g. “C. Grahne.
The Problem of Incomplete Information in Relational \\ Databases. Springer, 1991.”)

\begin{observation}
Computing certain answers of conditional instance is in $\Pi_2^p$.
\end{observation}

However, again by construction we can show the following for positive conditional instances: 

\begin{lemma}
Let $T$ be a positive conditional instance, and $N$ the naive instance given by dropping all conditions from 
$T$. Then  $Q$ holds in every minimal instance in $\hat \rep(T)$ if and only if 
$Q$ holds in every instance in $\rep(N)$
\end{lemma}

\begin{proof}
The if direction follows because an arbitrary assignment for the nulls in $N$ that sends each null to a fresh constant not appearing anywhere (not even in conditions) in 
$T$ yields an instance in both $\hat \rep(T)$ and $\rep(N)$ that is also minimal for $\hat \rep(T)$. 
For the only if direction we show that there is a homomorphism from $N$ to every minimal instance in $\hat \rep(T)$. 
Indeed, let $J$ be a minima instance in $\hat \rep(T)$, built from an assignment $\nu$ for $T$. Then the function  
mapping each null in $N$ as mandated by $\nu$ is indeed a homomorphism from $N$ to $J$: by construction it could 
be that $J$ contains more tuples than $\nu(N)$, but not the other way around. 
\end{proof}

We immediately obtain 
\begin{observation}
Computing certain answers of positive conditional instance is in $\np$, and for boolean queries $Q$ it suffices to 
check a homomorphism from $Q$ to $N$, where $N$ is the naive instance resulting of dropping all tuple with 
conditions in positive conditional instances. 
\end{observation}

\textbf{Membership in \nexptime (\np when $n$ is fixed)}. We can now outline our algorithm for query answering, given 
$P_1,\dots,P_n$, $I$ and $Q$ from the statement of the problem. Let $T$ an instance representing 
$\outcome_{P_1,\dots,P_n}(I)$, and let $N$ the naive instance constructed by dropping tuples with conditions in $T$. 
\begin{itemize}
\item Guess a homomorphism $h$ from $Q$ to $N$
\item For each atom in $h(Q)$, guess a set rules producing this atom, and for each such rule all homomorphisms 
needed to fire the rule during a chase.  
This set is at most exponential size, in $n$ (and 
polynomial if $n$ is fixed) because each atom in the conditional instance representing $\outcome_{P_1,\dots,P_i}(I)$ 
is produced by a rule in $P_{i}$, having at most $|P_i|$ atoms from $\outcome_{P_1,\dots,P_{i-1}}(I)$. The resulting 
size of the set is then bounded by $|P|^n$, where $P$ is the size of the biggest procedure i the sequence $P_1,\dots,P_n$. 
\item We can then check that by chasing the sequence $P_1,\dots,P_n$ of procedures one does produce the 
set of atoms and rules needed to witness $h(Q)$. The check is polynomial in the size of the set of rules producing 
$h(Q)$. 
\end{itemize}

\subsection{Proof of Theorem \ref{theo-readi-decidable-queries-dynamic}}

The key idea for this proof is the fact that, when computing the conditional instances representing the outcome of 
procedures as dictated by Proposition \ref{prop-minimal}, procedures with safe schema-alteration can only produce nulls 
the first time they appear in a sequence. 

To be more precise, assume a sequence  $P_1,\dots,P_\ell$ of procedures from $\Pi$ 
such that every instance in $\outcome_{P_1,\dots,P_\ell}(I)$ satisfies $Q$. 
By Proposition \ref{prop-minimal}, there is a conditional table $T$ whose minimal instances 
coincide with the minimal instances in  $\outcome_{P_1,\dots,P_\ell}(I)$.  

Let also $T_i$ the conditional instance representing the minimal instances of $\outcome_{P_1,\dots,P_i}(I)$. 

While $T_i$ may contain nulls, at most one null can be computed for each procedure with safe-schema alteration 
in $\Pi$ and each assignment tuple in $T_{i-1}$. The first time we apply such a procedure we can create at most 
$D^|\Sch|$ nulls,  where $D$ is the number of elements in $I$ and $|\Sch|$ is the schema of $T$, and thus 
the size of the resulting instance is at most $(D^{|\Sch|})^{|\Sch|}$. Then the number of nulls 
created is at most ${D^{|\Sch|}}^{|\Pi|}$. 

We can then continue the argument in the proof of Theorem \ref{theo-readi-decidable-queries}, except that instead of 
querying the minimal instance of the SKB we query the naive table resulting out of removing tuples with conditions 
from $T$. Since we now have a doubly-exponential number of elements, sequences may be of double exponential 
size (unless the number of procedures is fixed), from which the \twonexptime follows.

\end{document}